\documentclass[english]{article}
\usepackage[T1]{fontenc}
\usepackage[utf8]{inputenc}
\usepackage[a4paper]{geometry}
\usepackage{bm}
\usepackage{amsmath}
\usepackage{amsthm}
\usepackage{amssymb}
\usepackage{cite}

\usepackage{mathdots}
\usepackage{setspace}
\makeatletter
\numberwithin{equation}{section}
\theoremstyle{plain}
\newtheorem{prop}{\protect\propositionname}[section]
\ifx\proof\undefined\
  \newenvironment{proof}[1][\proofname]{\par
    \normalfont\topsep6\p@\@plus6\p@\relax
    \trivlist
    \itemindent\parindent
    \item[\hskip\labelsep
          \scshape
      #1]\ignorespaces
  }{%
    \endtrivlist\@endpefalse
  }
  \providecommand{\proofname}{Proof}
\fi
\theoremstyle{plain}
\newtheorem{lem}{\protect\lemmaname}[section]
\theoremstyle{definition}
\newtheorem{rem}{\protect\remarkname}[section]
\theoremstyle{definition}

\theoremstyle{definition}
\newtheorem{defn}{\protect\definitionname}[section]
\theoremstyle{plain}
\newtheorem{thm}{\protect\theoremname}[section]

\date{}

\usepackage{authblk}

\author[$\dagger$]{Błażej M. Szablikowski\footnote{Corresponding author.}}
\author[$\dagger$]{Maciej Błaszak}
\author[$\ddagger$]{Krzysztof Marciniak}

\affil[$\dagger$]{Faculty of Physics, Adam Mickiewicz University\authorcr
Uniwersytetu Pozna\'{n}skiego 2, 61-614 Poznań, Poland\authorcr
\texttt{bszablik@amu.edu.pl, blaszakm@amu.edu.pl}\authorcr\mbox{}}

\affil[$\ddagger$]{Department of Science and Technology\authorcr
Link\"{o}ping University, Campus Norrk\"{o}ping\authorcr
601 74 Norrk\"{o}ping, Sweden\authorcr
\texttt{krzma@itn.liu.se}}

\makeatother

\usepackage{babel}
\providecommand{\definitionname}{Definition}
\providecommand{\examplename}{Example}
\providecommand{\lemmaname}{Lemma}
\providecommand{\propositionname}{Proposition}
\providecommand{\remarkname}{Remark}
\providecommand{\theoremname}{Theorem}

\begin{document}
\title{\textsc{Stationary coupled KdV systems and their Stäckel representations}}
\maketitle

\begin{abstract}
In this article we investigate stationary coupled Korteweg--de Vries (cKdV) systems
and prove that every $N$-field stationary cKdV system can be written,
after a careful reparametrization of jet variables, as a classical
separable Stäckel system in $N+1$ different ways. For each of these
$N+1$ parametrizations we present an explicit map between the jet
variables and the separation variables of the system. Finally, we
show that each pair of Stäckel representations of the same stationary
cKdV system, when considered in the phase space extended by Casimir
variables, is connected by an appropriate finite-dimensional Miura map, which leads to
an $(N+1)$-Hamiltonian formulation for the stationary cKdV system.
\end{abstract}


\tableofcontents{}

\section{Introduction}

Since the classical works of Bogoyavlensky, Novikov \cite{BN} and
Mokhov \cite{Mo} there has been a tremendous amount of research devoted
to connections between soliton hierarchies and their integrable finite-dimensional
reductions (see for example the bibliography in the survey \cite{blaszak1998}
as well as the introductory part in \cite{1}).
In \cite{ant87,ant92} the authors presented a fairly general construction of
finite-dimensional completely integrable Hamiltonian systems obtained
by reductions of the Korteweg--de Vries (KdV)  and the coupled Korteweg--de Vries (cKdV) soliton hierarchies to their stationary and restricted flows.

In this article we thoroughly investigate the concept of
stationary cKdV systems and their Stäckel representations,
obtaining new results in this area, as described below.

In Section~\ref{S2} we remind the construction of the $N$-component
cKdV hierarchy $\bm{u}_{t_{n}}=\bm{K}_{n}[\bm{u}]$
from the energy-dependent Schrödinger spectral problem
\cite{Alonso,AF}. We present these -- classical -- results in a novel setting of symmetric
bilinear differential operators $\mathcal{J}_{i}$, defined in (\ref{s10}),
which allows for presenting the pure algebraic recursion formulas
for conserved densities (see Proposition \ref{wazna} and formula
(\ref{beautiful})). In Subsection \ref{mhs} we present the integrated
form of kernels of all $N+1$ Hamiltonian operators $\mathbb{B}_{m}$ of the cKdV hierarchy (see Proposition \ref{propker})
(the $N$-component cKdV hierarchy is $N+1$ Hamiltonian) in the language
of bi-linear operators $\mathcal{J}_{i}$. Both these results are
new. In Subsection \ref{S24} we remind the reader the structure
of the Lax formulation of cKdV hierarchy.

In Section~\ref{3S} we study the main object of this paper: the stationary
cKdV\ system. This is a finite-dimensional system originating by
restricting the cKdV hierarchy to one of its stationary manifolds
$\mathcal{M}_{n}=\left\{ \left[\bm{u}\right]:K_{n+1}[\bm{u}]=0\right\} $.
In Subsection \ref{31S} we explain how each Hamiltonian operator
$\mathbb{B}_{m}$ generates a Hamiltonian foliation of the $\left(2n+N\right)$-dimensional
stationary manifold $\mathcal{M}_{n}$ into $2n$-dimensional leaves.
These Hamiltonian foliations are transversal to each other. In Subsection
\ref{32S} we introduce another set of $N+1$ foliations of $\mathcal{M}_{n}$
into $2n$-dimensional leaves, each related with imposing a particular
stationary constraint on Lax representation of the cKdV hierarchy.
We call these foliations Stäckel foliations of $\mathcal{M}_{n}$.
We conclude this subsection by proving that Hamiltonian foliation
and Stäckel foliation for the same $m$ coincide (see Theorem~\ref{foliacje}).
In Subsection~\ref{33S} we present the Lax representation
of the stationary cKdV system on leaves of each of Stäckel foliations
and then we show how these foliations lead in a very natural way to
separation (spectral) curves of Stäckel systems.

Section~\ref{4S} contains some basic facts about Stäckel systems
and their Lax formulation. This section is
introduced in order to keep the article self-contained.

In Section \ref{5S}, comparing the Lax formulation
of the stationary cKdV system as given in Subsection~\ref{33S} with
the Lax formulation of Stäckel systems as given in Section~\ref{4S}, we prove
the main theorem of this article (Theorem~\ref{great}), stating that
each stationary cKdV\ system can be represented as a Stäckel system
on $\mathcal{M}_{n}$ on $N+1$ ways.

In Section~\ref{6S} we prove that all Stäckel representations of
the stationary cKdV systems obtained in Section \ref{5S}, are connected
by a corresponding Miura map. This yields immediately $(N+1)$-Hamiltonian
formulations of the cKdV system on $\mathcal{M}_{n}$. Thus, in this
article we demonstrate that the $(N+1)$-Hamiltonian structure of
the cKdV hierarchy generates $(N+1)$-Hamiltonian structure of the
stationary cKdV system on $\mathcal{M}_{n}$. In consequence, the
leaves of a $m$-th Stäckel foliation become symplectic leaves of
the corresponding $m$-th Hamiltonian operator of the stationary cKdV
system.

Section~\ref{7S} is devoted to examples. The subsection \ref{P1} focuses on the Dispersive Water Waves (DWW) hierarchy (so that $N=2$) and its first part contains a very detailed presentation of all
main formulas and ingredients for the case $n=2$. The second part of this subsection presents the case $n=3$ and $m=0$. The subsection \ref{ostatni} is devoted to the case $N=4$, $n=2$ and $m=0$. This last example is non-generic since $n+m<N-1$.

In article \cite{1} we performed a similar analysis to described above, but for the KdV case, i.e.~the one-component case.

Recently, in \cite{FH,FH2}, the authors revisited the idea of the
stationary cKdV system. Their results focused mainly on the $N=2$
and $N=3$ case and on first two (i.e.~lowest) flows of the
hierarchy, while our analysis is general, i.e.~valid for all $N$
and all $n$. Besides, they do not consider the separability problem in the general setting, intensively studied in our article.

Let us also mention that
a new approach, partly related to our results, based on the
Nijenhuis geometry applicable to multi-component integrable
equations, i.e.~of KdV type, was proposed recently in
\cite{bolsIII,bolsIV}.

Let us mention that the construction inverse to the one presented in this article is also possible:
starting from a carefully chosen family of Stäckel systems one can reconstruct the related hierarchies of stationary systems and hence reconstruct the associated soliton hierarchy. The idea of such a construction appeared for the first time in 1999 during a visit of one of the authors (M.B.) to 
A.P.~Fordy in Leeds University. This idea was then explored for the
first time in \cite{blaszak2005,blaszak2006,blaszak2008}.

Finally, let us also mention that a similar idea, linking Stäckel systems with dispersionless field systems was introduced in
papers by Ferapontov and Fordy \cite{FF1,FF2,FF3} and the paper \cite{fordy}.

\section{Coupled KdV hierarchy\label{S2}}

 In this section we review, following  \cite{Alonso,AF}, and develop the classical construction of the cKdV hierarchy from the energy dependent
Schrödinger spectral problem as well its multi-Hamiltonian representation.

\subsection{Energy-dependent Schrödinger spectral problem}

The $N$-component cKdV hierarchy originates as the compatibility
condition of the energy dependent Schrödinger spectral problem with
the appropriate evolutionary part:
\begin{align}
 & \psi_{xx}+\mathbb{Q}\psi=0,\nonumber \\
 & \psi_{t_{k}}=\frac{1}{2}\mathbb{P}_{k}\psi_{x}-\frac{1}{4}(\mathbb{P}_{k})_{x}\psi,\qquad k=1,2,\ldots,\label{c1}
\end{align}
where
\[
\mathbb{Q}:=\sum_{i=0}^{N}u_{i}\lambda^{i},\qquad u_{N}\equiv-1.
\]
Here and below, $u_{i}=u_{i}(x,t_{1},t_{2},\ldots)$ are the dynamical
fields, while
$\mathbb{P}_{k}$ are so far unspecified functions of the spectral parameter $\lambda$
and jet variables in $u_i$.
The compatibility conditions $\left(\psi_{xx}\right)_{t_{k}}=\left(\psi_{t_{k}}\right)_{xx}$
of \eqref{c1} yield the following hierarchy of evolution equations
\begin{equation}
\mathbb{Q}_{t_{k}}=\frac{1}{4}(\mathbb{P}_{k})_{3x}+\mathbb{Q}(\mathbb{P}_{k})_{x}+\frac{1}{2}\mathbb{Q}_{x}\mathbb{P}_{k}\equiv J\mathbb{P}_{k},\qquad k=1,2,\ldots \label{c2}
\end{equation}
where
\begin{equation}
J\equiv\sum_{i=0}^{N}J_{i}\lambda^{i},\qquad J_{i}:=\frac{1}{4}\delta_{i0}\partial_{x}^{3}
+u_{i}\partial_{x}+\frac{1}{2}(u_{i})_x.\label{c22}
\end{equation}

Further, in accordance with \cite{Alonso,AF}, we assume that each
$\mathbb{P}_{k}$ is a polynomial of order $k-1$ in $\lambda$:
\begin{equation}
\mathbb{P}_{k}=\sum_{i=0}^{k-1}P_{k-1-i}\lambda^{i}\equiv P_{0}\lambda^{k-1}+\ldots+P_{k-2}\lambda+P_{k-1}.\label{c3}
\end{equation}
The conditions for the coefficients $P_{i}$ in (\ref{c3}) can be
obtained by requiring consistency of the evolution equations \eqref{c2}.
It turns out that the coefficients $P_{i}$ in \eqref{c2} or \eqref{c3}
actually do \emph{not} depend on $k$ and that they satisfy the equation
\begin{equation}
J\mathcal{P}\equiv\frac{1}{4}\mathcal{P}_{3x}+\mathbb{Q}\mathcal{P}_{x}+\frac{1}{2}\mathbb{Q}_{x}\mathcal{P}=0,\label{c5}
\end{equation}
where $\mathcal{P}$ is the Laurent series in $\lambda$:
\begin{equation}
\mathcal{P}=\sum_{i=0}^{\infty}P_{i}\lambda^{-i},\label{c4}
\end{equation}
and thus
\[
\mathbb{P}_{k}=\left[\lambda^{k-1}\mathcal{P}\right]_{\geqslant0},
\]
where $[\cdot]_{\geqslant0}$ means the projection on the part polynomial
in $\lambda$.

Notice that
\[
J\mathcal{P}=\sum_{i=0}^{N}\sum_{j=0}^{\infty}J_{i}P_{j}\lambda^{i-j}\equiv\sum_{k=-\infty}^{N}\text{ }\sum_{i=\text{max}\{0,k\}}^{N}J_{i}P_{i-k}\lambda^{k}
\]
and thus a straightforward consequence of \eqref{c5} is the equality
\begin{equation}
\sum_{i=0}^{N}J_{i}P_{i-k}=0,\qquad\text{where}\quad k\leqslant N.\label{c6}
\end{equation}
To simplify the notation, in \eqref{c6} and further we assume that
$P_{i}=0$ for $i<0$.\footnote{Thus in particular for $0\leqslant k\leqslant N$ the formula \eqref{c6}
takes the form
\[
\sum_{i=k}^{N}J_{i}P_{i-k}=0.
\]
} The condition \eqref{c5} not only provides us with (differential)
equations on the coefficients $P_{i}$ in \eqref{c3} or \eqref{c4},
but also assures the consistency of the equations from the hierarchy
\eqref{c2}.

There exists an alternative description of the cKdV hierarchy within
the above scheme. If we define
\begin{equation*}
\bar{\mathbb{P}}_{k}:=\bigl[\lambda^{k-1}\mathcal{P}\bigr]_{<0}\equiv\lambda^{k-1}\mathcal{P}-\mathbb{P}_{k}=\sum_{i=k}^{\infty}P_{i}\lambda^{k-1-i} 
\end{equation*}
then, due to \eqref{c5}, $J\mathbb{P}_{k}=-J\bar{\mathbb{P}}_{k}$, and
thus the hierarchy \eqref{c2} can alternatively be written as
\begin{equation}
\mathbb{Q}_{t_{k}}=-J\bar{\mathbb{P}}_{k}.\label{c10}
\end{equation}

Consequently, the cKdV hierarchy can be defined by \eqref{c2},
or equivalently by \eqref{c10}, provided that the condition \eqref{c5} holds.
The members of the cKdV hierarchy have the form of mutually commuting
$N$-component evolution equations,
\begin{equation}
\bm{u}_{t_{k}}=\bm{K}_{k},\qquad k=1,2,\ldots,\label{c11a}
\end{equation}
defined on an infinite-dimensional functional (smooth) manifold $\mathcal{F}$.
Coordinates on $\mathcal{F}$ are given by jet variables $[\bm{u}]:=(\bm{u},\bm{u}_{x},\bm{u}_{xx},\ldots)$,
with the (field) vector $\bm{u}:=(u_{0},\ldots,u_{N-1})^{T}$. More
explicitly, by \eqref{c2} and \eqref{c10}, the evolution equations
\eqref{c11a} can be written in the following two equivalent ways:
\begin{equation}
(u_{i-1})_{t_{k}}=(\bm{K}_{k})_{i}\equiv\sum_{j=0}^{i-1}J_{j}P_{j-i+k}\equiv-\sum_{j=i}^{N}J_{j}P_{j-i+k},\qquad i=1,\ldots,N.\label{c11b}
\end{equation}
The above equivalence is immediately apparent form the equality \eqref{c6}
and will be used in subsection \ref{mhs} for reconstructing the known \cite{AF} multi-Hamiltonian structure of the cKdV hierarchy.

For $N=1$ the cKdV hierarchy reduces to the KdV
hierarchy and for $N=2$ to the Dispersive Water Waves (DWW) hierarchy.

\subsection{Algebraic recursion}

To solve \eqref{c5} for the coefficients $P_{i}$ we need to integrate differential
equations provided by this condition or equivalently by \eqref{c6}.
Integrating \eqref{c5} we obtain
\begin{equation}
\frac{1}{2}\mathcal{PP}_{xx}-\frac{1}{4}\mathcal{P}_{x}^{2}+\mathbb{Q}\mathcal{P}^{2}
= c(\lambda) \equiv  -4\lambda^{N},\label{c7}
\end{equation}
where $c(\lambda)$ is a polynomial in $\lambda$ with coefficient being
constants of integration of differential equations provided by \eqref{c5}.
Here, for convenience, we make the simplest possible choice $c(\lambda)\equiv-4\lambda^{N}$ so that $P_{0}=2$. Other choices lead to hierarchies \eqref{c11a} with members being linear combinations of symmetries originating from the simplest possible choice as described above
and by a linear change of basis in the cKdV hierarchy we can always choose the polynomial 
constant $c(\lambda)$ as in \eqref{c7}.

We will now attempt to solve \eqref{c7} recursively for the coefficients $P_{i}$. Let us start
by defining the following auxiliary differential (symmetric) bi-linear operators:
\begin{equation}
\mathcal{J}(f,g):=\sum_{i=0}^{N}\mathcal{J}_{i}(f,g)\lambda^{i},\qquad\mathcal{J}_{i}(f,g):=-\frac{1}{16}\delta_{i0}\left(f_{xx}g+fg_{xx}-f_{x}g_{x}\right)-\frac{1}{4}u_{i}fg\label{s10}.
\end{equation}
Thus
\begin{equation}
\mathcal{J}(f,g)\equiv-\frac{1}{16}(f_{xx}g+fg_{xx}-f_{x}g_{x})-\frac{1}{4}\mathbb{Q}fg.\label{Jpisane}
\end{equation}
The bi-linear operators \eqref{s10} and the linear operators \eqref{c22}
are related by the following useful formulas:
\begin{equation}
\left [\mathcal{J}_{i}(f,g) \right ]_{x}=-\frac{1}{4}\left(fJ_{i}g+gJ_{i}f\right)\label{s12b}
\end{equation}
and
\begin{equation*}
\left [\mathcal{J}(f,f)\right ]_{x}=-\frac{1}{2}fJf. 
\end{equation*}

Using the above b-linear operators, the equation \eqref{c7} can be written in the equivalent form:
\begin{equation}
\mathcal{J}(\mathcal{P},\mathcal{P})=\lambda^{N}.\label{c8}
\end{equation}
Now, since
\[
\mathcal{J}(\mathcal{P},\mathcal{P})\equiv\sum_{i=0}^{N}\sum_{j=0}^{\infty}\sum_{k=0}^{\infty}\mathcal{J}_{i}(P_{j},P_{k})\lambda^{i-j-k}\equiv\sum_{k=-\infty}^{N}\sum_{i=0}^{N}\sum_{j=0}^{i-k}\mathcal{J}_{i}(P_{j},P_{i-j-k})\lambda^{k},
\]
by \eqref{s10} and since $P_{0}=2$ we find that $\mathcal{J}_{N}(P_{0},P_{0})=1$
and thus, from the above equation, it follows that
\begin{equation}
\sum_{i=0}^{N}\sum_{j=0}^{i-k}\mathcal{J}_{i}(P_{j},P_{i-j-k})=0\qquad\text{for}\qquad k<N.\label{c8a}
\end{equation}

In \cite{AF} it was noticed that we can always solve \eqref{c7}
for coefficients $P_i$ in terms of previously calculated (differential) expressions in 
$u_0,\ldots, u_{N-1}$. The following proposition provides us with a compact formula for that.

\begin{prop}
\label{wazna}The coefficients $P_{i}$ of the series \eqref{c4}
satisfy the following recursive formula
\begin{equation}
P_{k}=-\sum_{j=1}^{k-1}\mathcal{J}_{N}(P_{j},P_{k-j})-\sum_{i=0}^{N-1}\sum_{j=0}^{i+k-N}\mathcal{J}_{i}(P_{j},P_{i-j+k-N}),\qquad k=1,2,\ldots.\label{c8b}
\end{equation}
\end{prop}
Note that (\ref{c8b}) is indeed of a recursive form as the right
hand side contains $P_{i}$ only up to $P_{k-1}$. Note also that
this formula is purely differential-algebraic.

\begin{proof}
The recursion \eqref{c8b} is a consequence of the formula \eqref{c8a},
which can be rewritten in the form
\begin{equation}
\sum_{i=0}^{N}\sum_{j=0}^{i+k-N}\mathcal{J}_{i}(P_{j},P_{i-j+k-N})=0,\qquad\text{where}\quad k\geqslant1.\label{c8c}
\end{equation}
For fixed $k$ \eqref{c8c} involves only coefficients $P_{i}$ for
$0\leqslant i\leqslant k$ and $P_{k}$ can be found only in the terms
$\mathcal{J}_{N}(P_{0},P_{k})=\frac{1}{2}P_{k}$. Thus, solving for
$P_{k}$ we find the recursion \eqref{c8b}.
\end{proof}

Explicitly, (\ref{c8b}) can be written as
\begin{equation}
P_{k}=\frac{1}{4}\left[-\sum_{j=1}^{k-1}P_{k-j}P_{j}+\sum_{j=0}^{k-N}\left(\frac{1}{2}P_{k-j-N}(P_{j})_{xx}-\frac{1}{4}(P_{k-j-N})_{x}(P_{j})_{x}\right)+\sum_{i=0}^{N-1}\sum_{j=0}^{i+k-N}u_{i}P_{i-j+k-N}P_{j}\right],\label{beautiful}
\end{equation}
where $k=1,2,\ldots$. The formula (\ref{c8b}) (or (\ref{beautiful}))
was not present in literature before and it provides us with an effective
way of calculating higher coefficients $P_{i}$ from the lower ones without
any need of integration. The functions $P_{i}$ turn out to be components
of cosymmetries of the cKdV hierarchy (\ref{c11a}), see the formula (\ref{mham}) below.

\subsection{Multi-Hamiltonian structure\label{mhs}}

The evolution equations from the $N$-component cKdV hierarchy
\eqref{c11a} are multi-Hamiltonian with respect to $N+1$ mutually
compatible Hamiltonian operators $\mathbb{B}_{m}$: \cite{Alonso,AF}
\begin{equation}
\bm{u}_{t_{r}}=\bm{K}_{r}\equiv\mathbb{B}_{0}\bm{\gamma}_{r}=\ldots=\mathbb{B}_{m}\bm{\gamma}_{r-m}=\ldots=\mathbb{B}_{N}\bm{\gamma}_{r-N},\qquad m=0,1,\ldots,N,\label{mham}
\end{equation}
where $\bm{\gamma}_{r}=(P_{r},\ldots,P_{r+N-1})^{T}$ are cosymmetries
of the hierarchy with $P_{i}$ given by \eqref{c8b} or \eqref{beautiful}. The first and the last Hamiltonian structure (that
is with respect to the Hamiltonian operators $\mathbb{B}_{0}$ and
$\mathbb{B}_{N}$) are direct consequences of the two representations
the hierarchy expressed in \eqref{c11b}. The remaining Hamiltonian
structures with respect to the Hamiltonian operators $\mathbb{B}_{m}$
can be constructed taking $m$ first equations from the structure
with respect to the operator $\mathbb{B}_{N}$ and $N-m$ last equations
from the structure with respect to the operator~$\mathbb{B}_{0}$.
So, the Hamiltonian operators $\mathbb{B}_{m}$, for $m=0,1,\ldots N$,
act on an arbitrary covector field $\bm{\eta}=(\eta_{1},\ldots,\eta_{N})^{T}$
as
\begin{align*}
\bigl(\mathbb{B}_{m}\bm{\eta}\bigr)_{j} & =\sum_{i=0}^{j-1}J_{i}\eta_{i-j+m+1}\qquad\text{for}\qquad1\leqslant j\leqslant m,\\
\bigl(\mathbb{B}_{m}\bm{\eta}\bigr)_{j} & =-\sum_{i=j}^{N}J_{i}\eta_{i-j+m+1}\qquad\text{for}\qquad m+1\leqslant j\leqslant N,
\end{align*}
Thus, the operators $\mathbb{B}_{m}$ have the explicit form:
\begin{equation*}
\mathbb{B}_{m}=\left(\begin{array}{c|c}
\begin{matrix} &  & J_{0}\\
 & \iddots & \vdots\\
J_{0} & \cdots & J_{m-1}
\end{matrix} & 0\\
\hline 0 & \begin{matrix}-J_{m+1} & \cdots & -J_{N}\\
\vdots & \iddots\\
-J_{N}
\end{matrix}
\end{array}\right),\qquad m=0,\ldots,N. 
\end{equation*}

Any two consecutive Hamiltonian operators~$\mathbb{B}_{m}$
define (the same) hereditary recursion operator $\mathbb{R}$ through
\[
\mathbb{R}:=\mathbb{B}_{m+1}\mathbb{B}_{m}^{\,-1},\qquad m=0,1,\ldots,N-1,
\]
given explicitly by
\begin{equation}
\mathbb{R}=\left(\begin{array}{c|c}
\begin{matrix}0 & \cdots & 0\end{matrix} & -J_{0}J_{N}^{-1}\\
\hline \begin{matrix}1\\
 & \ddots\\
 &  & 1
\end{matrix} & \begin{matrix}-J_{1}J_{N}^{-1}\\
\vdots\\
-J_{N-1}J_{N}^{-1}
\end{matrix}
\end{array}\right), \label{R}
\end{equation}
so that one can generate all the vector fields of cKdV hierarchy and
their cosymmetries with the help of the recursion operator $\mathbb{R}$
through
\begin{equation*}
\bm{K}_{r+1}=\mathbb{R}^{r}\bm{K}_{1}\qquad\text{and}\qquad\bm{\gamma}_{r+1-N}=(\mathbb{R}^{\dagger})^{r}\bm{\gamma}_{1-N},\qquad r=1,2,\ldots, 
\end{equation*}
where $\bm{K}_{1}=\bm{u}_{x}$ and $\bm{\gamma}_{1-N}=(0,\ldots,P_{0})^{T}$.
Let us however point out that the above method of constructing the
hierarchy requires integrating the nonlocal operator (\ref{R}) while
our formula (\ref{c8b}) gives us an explicit (although recursive)
form of all $P_{k}$ that are obtained by purely differential operations.

In what follows we will need two propositions (Proposition
\ref{propker} and Proposition \ref{altpropker}) that characterize
the kernels of the Hamiltonian operators $\mathbb{B}_{m}$. In
order to formulate and prove these propositions we need to define
the following functions:
\begin{subequations}
\label{d1}
\begin{align}
f_{k,m}(\bm{\xi}) & :=\sum_{i=0}^{k-1}\sum_{j=i+1}^{k}\mathcal{J}_{i}(\xi_{m+j-k},\xi_{m+i-j+1}),\qquad 1\leqslant k\leqslant m,\label{d1a}\\
g_{k,m}(\bm{\xi}) & :=-2\sum_{i=k}^{N}\sum_{j=k}^{i}\mathcal{J}_{i}(P_{j-k},\xi_{m+i-j+1}),
\qquad m+1\leqslant k\leqslant N,\label{d1b}\\
\tilde{g}_{k,m}(\bm{\xi},) & :=\sum_{i=k}^{N}\sum_{j=k}^{i}\mathcal{J}_{i}(\xi_{m+j-k+1},\xi_{m+i-j+1}),\qquad m+1\leqslant k\leqslant N,\label{d1c}
\end{align}
\end{subequations}
 where $m\in\left\{ 0,\ldots,N\right\} $
and $\bm{\xi}=(\xi_{1},\ldots,\xi_{N})^{T}$ with $\xi_{i}=\xi_{i}[\bm{u}]$
is an arbitrary covector.
\begin{lem}
We have:
\begin{subequations}
\label{d2}
\begin{align}
[f_{k,m}(\bm{\xi})]_{x} & =-\frac{1}{2}\sum_{j=1}^{k}\xi_{m+j-k}\bigl(\mathbb{B}_{m}\bm{\xi}\bigr)_{j},\label{d2a}\\{}
[g_{k,m}(\bm{\xi})]_{x} & =-\frac{1}{2}\sum_{j=k}^{N}P_{j-k}\bigl(\mathbb{B}_{m}\bm{\xi}\bigr)_{j},\label{d2b}\\{}
[\tilde{g}_{k,m}(\bm{\xi})]_{x} & =\frac{1}{2}\sum_{j=k}^{N}\xi_{m+j-k+1}\bigl(\mathbb{B}_{m}\bm{\xi}\bigr)_{j}.\label{d2c}
\end{align}
\end{subequations}
\end{lem}
\begin{proof}
Differentiating \eqref{d1a} and \eqref{d1c} with respect to $x$
and using the relation \eqref{s12b} we see that
\begin{align*}
[f_{k,m}(\bm{\xi})]_{x} & =-\frac{1}{4}\sum_{i=0}^{k-1}\sum_{j=i+1}^{k}\bigl[\xi_{m+j-k}J_{i}\xi_{m+i-j+1}+\xi_{m+i-j+1}J_{i}\xi_{m+j-k}\bigr]\\
 & \equiv-\frac{1}{2}\sum_{i=0}^{k-1}\sum_{j=i+1}^{k}\xi_{m+j-k}J_{i}\xi_{m+i-j+1}=-\frac{1}{2}\sum_{j=1}^{k}\xi_{m+j-k}\sum_{i=0}^{j-1}J_{i}\xi_{m+i-j+1}
\end{align*}
and
\begin{align*}
[\tilde{g}_{k,m}(\bm{\xi})]_{x} & =-\frac{1}{4}\sum_{i=k}^{N}\sum_{j=k}^{i}\bigl[\xi_{m+j-k+1}J_{i}\xi_{m+i-j+1}+\xi_{m+i-j+1}J_{i}\xi_{m+j-k+1}\bigr]\\
 & \equiv-\frac{1}{2}\sum_{i=k}^{N}\sum_{j=k}^{i}\xi_{m+j-k+1}J_{i}\xi_{m+i-j+1}=-\frac{1}{2}\sum_{j=k}^{N}\xi_{m+j-k+1}\sum_{i=j}^{N}J_{i}\xi_{m+i-j+1}.
\end{align*}
Hence, we obtain \eqref{d2a} and \eqref{d2c}. Differentiating \eqref{d1b}
we obtain
\begin{align*}
[g_{k,m}(\bm{\xi})]_{x} & =\frac{1}{2}\sum_{i=k}^{N}\sum_{j=k}^{i}\bigl[\xi_{m+i-j+1}J_{i}P_{j-k}+P_{j-k}J_{i}\xi_{m+i-j+1}\bigr]\\
 & =\frac{1}{2}\sum_{i=k}^{N}\sum_{j=k}^{i}P_{j-k}J_{i}\xi_{m+i-j+1}=\frac{1}{2}\sum_{j=k}^{N}P_{j-k}\sum_{i=j}^{N}J_{i}\xi_{m+i-j+1}
\end{align*}
as
\[
\sum_{i=k}^{N}\sum_{j=k}^{i}\xi_{m+i-j+1}J_{i}P_{j-k}\equiv\sum_{i=k}^{N}\sum_{j=k}^{i}\xi_{m+j-k+1}J_{i}P_{i-j}\equiv\sum_{j=k}^{N}\xi_{m+j-k+1}\sum_{i=j}^{N}J_{i}P_{i-j}=0,
\]
where the last equality is a consequence of \eqref{c6} since here $j\geqslant0$.
Hence, \eqref{d2b} follows.
\end{proof}
The following proposition describes the form of kernels
of all Hamiltonian operators $\mathbb{B}_{m}$.
\begin{prop}
\label{propker} For fixed $m\in\{0,1,\ldots,N\}$, $\bm{\xi}\in\ker\mathbb{B}_{m}$
(that is $\mathbb{B}_{m}\bm{\xi}=0$) if and only if
\begin{subequations}
\label{d3}
\begin{align}
f_{k,m}(\bm{\xi}) & =c_{k}\qquad\text{for}\qquad1\leqslant k\leqslant m,\label{d3a}\\
g_{k,m}(\bm{\xi}) & =c_{k}\qquad\text{for}\qquad m+1\leqslant k\leqslant N,\label{d3b}
\end{align}
\end{subequations}
 where $c_{1},\ldots,c_{N}$ are arbitrary constants.
\end{prop}
\begin{proof}

Let us assume that the conditions (\ref{d3}) hold. Differentiating
\eqref{d3a} and using \eqref{d2a} we obtain the system
\begin{equation}
[f_{k,m}(\bm{\xi})]_{x}\equiv-\frac{1}{2}\sum_{j=1}^{k}\xi_{m+j-k}\bigl(\mathbb{B}_{m}\bm{\xi}\bigr)_{j}=0,\qquad1\leqslant k\leqslant m,\label{d4a}
\end{equation}
which recursively implies that
\[
(\mathbb{B}_{m}\bm{\xi}\bigr)_{j}=0\qquad\text{for}\qquad1\leqslant j\leqslant m.
\]
Notice that the implication is correct since the formula \eqref{d4a} can be interpreted as the matrix product of the triangular matrix $\xi_{m+j-k}$ with the constant non-zero diagonal term
$\xi_m$ with the vector $\mathbb{B}_{m}\bm{\xi}$.
For $m\geqslant 1$ always $\xi_m\neq 0$, which follows from the condition: $J_0\xi_m = 0$, required by the fact that $\bm{\xi}\in\ker\mathbb{B}_{m}$.

Next, differentiating \eqref{d3b} and using \eqref{d2b} we have
\begin{equation}
[g_{k,m}(\bm{\xi})]_{x}=-\frac{1}{2}\sum_{j=k}^{N}P_{j-k}\bigl(\mathbb{B}_{m}\bm{\xi}\bigr)_{j}=0,\qquad m+1\leqslant k\leqslant N,\label{d4b}
\end{equation}
from which
\[
(\mathbb{B}_{m}\bm{\xi}\bigr)_{j}=0\qquad\text{for}\qquad m+1\leqslant j\leqslant N.
\]
Thus, the conditions (\ref{d3}) imply that $\mathbb{B}_{m}\bm{\xi}=0$.
The reverse implication is a matter of straightforward integration
of \eqref{d4a} and \eqref{d4b}.
\end{proof}

Later we will also need an alternative description of kernels
of $\mathbb{B}_{m}$, contained in the following proposition.
\begin{prop}
\label{altpropker} Let us fix $m\in\{0,1,\ldots,N\}$ and a natural
$n>0$ such that $n+m\leqslant N-2$. Then $\bm{\xi}\in\ker\mathbb{B}_{m}$
(that is $\mathbb{B}_{m}\bm{\xi}=0$) if and only if
\begin{subequations}
\label{d5}
\begin{align}
f_{k,m}(\bm{\xi}) & =c_{k}\qquad\text{for}\qquad1\leqslant k\leqslant m,\label{d5a}\\
g_{k,m}(\bm{\xi})+\tilde{g}_{k+n+1,m}(\bm{\xi}) & =c_{k}\qquad\text{for}\qquad m+1\leqslant k\leqslant N-n-1,\label{d5b}\\
g_{k,m}(\bm{\xi}) & =c_{k}\qquad\text{for}\qquad N-n\leqslant k\leqslant N,\label{d5c}
\end{align}
\end{subequations}
 where $c_{1},\ldots,c_{N}$ are arbitrary constants.
\end{prop}
\begin{proof}
Just like in the previous proof, one can see that \eqref{d5a} and
\eqref{d5c} imply that
\begin{equation}
(\mathbb{B}_{m}\bm{\xi}\bigr)_{j}=0\qquad\text{for}\qquad1\leqslant j\leqslant m\quad\text{and}\quad N-n\leqslant j\leqslant N.\label{d6}
\end{equation}
Finally, differentiating \eqref{d5b} and taking into account \eqref{d6},
\begin{equation}
[g_{k,m}(\bm{\xi})+\tilde{g}_{k+n+1,m}(\bm{\xi})]_{x}=-\frac{1}{2}\sum_{j=k}^{N-n-1}P_{j-k}\bigl(\mathbb{B}_{m}\bm{\xi}\bigr)_{j}+\frac{1}{2}\sum_{j=k+n+1}^{N-n-1}\xi_{m+j-k-n}\bigl(\mathbb{B}_{m}\bm{\xi}\bigr)_{j}=0,\label{d7}
\end{equation}
where $m+1\leqslant k\leqslant N-n-1$. Thus, \eqref{d7} recursively
implies that
\begin{equation}
(\mathbb{B}_{m}\bm{\xi}\bigr)_{j}=0\qquad\text{for}\qquad m+1\leqslant j\leqslant N-n-1.\label{d8}
\end{equation}
Hence, collecting together \eqref{d6} and \eqref{d8} we actually
see that from the conditions (\ref{d5}) it follows that $\mathbb{B}_{m}\bm{\xi}=0$.
The reverse implication is straightforward.
\end{proof}
\begin{rem}
Let us notice that the characterizations of the kernels of the Hamiltonian
operators $\mathbb{B}_{m}$, as in the above propositions, are not
the only possible ones. For instance, Proposition \ref{propker} would
still be correct if the functions $g_{k,m}$ in \eqref{d3b} were
entirely replaced by $\tilde{g}_{k,m}$, as defined in \eqref{d1c}.
Choices made in Propositions \ref{propker} and \ref{altpropker}
are dictated by later needs.
\end{rem}

\subsection{Lax representation\label{S24}}

Introducing the vector eigenfunction $\Psi=(\psi,\psi_{x})^{T}$ we
can rewrite the linear problems \eqref{c1} for the cKdV hierarchy
(\ref{c11a}) in the form
\begin{equation}
\Psi_{t_{k}}=\mathbb{V}_{k}\Psi,\qquad k=1,2,3,\ldots,\label{l1}
\end{equation}
where $t_{1}\equiv x$ and
\begin{equation}
\mathbb{V}_{1}=\begin{pmatrix}0 & 1\\
-\mathbb{Q} & 0
\end{pmatrix}\text{ \ \ \ and \ \ \ }\mathbb{V}_{k}=\begin{pmatrix}-\frac{1}{4}\left(\mathbb{P}_{k\,}\right)_{x} & \frac{1}{2}\mathbb{P}_{k}\\
-\frac{1}{4}\left(\mathbb{P}_{k}\right)_{xx}-\frac{1}{2}\mathbb{Q}\mathbb{P}_{k} & \frac{1}{4}\left(\mathbb{P}_{k}\right)_{x}
\end{pmatrix},\qquad k=2,3,\ldots.\label{l11}
\end{equation}
see \cite{AF}. Then, by the compatibility conditions $(\Psi_{x})_{t_{k}}=(\Psi_{t_{k}})_{x}$
and $(\Psi_{t_{k}})_{t_{s}}=(\Psi_{t_{k}})_{t_{s}}$ we obtain, respectively,
the Lax equations
\begin{equation}
\frac{d}{dt_{k}}\mathbb{V}_{1}-\frac{d}{dx}\mathbb{V}_{k}+[\mathbb{V}_{1},\mathbb{V}_{k}]=0,\qquad k=1,2,\ldots\,\label{l2}
\end{equation}
and the zero-curvature equations
\begin{equation}
\frac{d}{dt_{k}}\mathbb{V}_{s}-\frac{d}{dt_{s}}\mathbb{V}_{k}+[\mathbb{V}_{s},\mathbb{V}_{k}]=0,\qquad s,k=2,3,\ldots.\label{l3}
\end{equation}
Thus, $\mathbb{V}_{1}$ can be considered as the Lax matrix of the
cKdV hierarchy, the matrices $\mathbb{V}_{k}$ for $k>1$ play the role of auxiliary matrices while the hierarchy itself can be obtained by the matrix
Lax equations \eqref{l2}. The zero-curvature equations \eqref{l3}
are differential consequences of the hierarchy.

\section{Stationary cKdV systems\label{3S}}

In this section we consider (see also the special case considered
in \cite{1}) stationary cKdV systems. A stationary cKdV system is
a system that originates by restricting the (infinite) cKdV hierarchy
(\ref{c11a}) to one of its stationary manifolds. We will then show
that the resulting finite-dimensional integrable system can be in
a very natural way associated with an appropriate Stäckel system.
Actually, due to the fact that the $N$-component cKdV hierarchy is
$(N+1)$-hamiltonian, see (\ref{mham}), we can perform this association
on $N+1$ different ways.

The $(n+1)$-th stationary flow of the cKdV hierarchy \eqref{c11a}
is determined by the following condition:
\begin{equation}
\bm{u}_{t_{n+1}}=0\qquad\text{or equivalently}\qquad\bm{K}_{n+1}=0,\label{r1}
\end{equation}
which by \eqref{c11b} takes the form of a system of $N$ differential equations:
\begin{equation}
(\bm{K}_{n+1})_{j}\equiv\sum_{i=0}^{j-1}J_{i}P_{i-j+n+1}\equiv-\sum_{i=j}^{N}J_{i}P_{i-j+n+1}=0,\qquad j=1,\ldots,N.\label{s15}
\end{equation}

The stationary condition \eqref{r1} provides a constraint (or rather
a system of constraints) on the infinite-dimensional (functional)
manifold $\mathcal{F}$, on which the cKdV hierarchy is defined, reducing
it to the finite-dimensional submanifold, $n$-th stationary manifold:
\[
\mathcal{M}_{n}=\left\{ [\bm{u}]\in\mathcal{F}\,|\,\bm{K}_{n+1}=0\right\} .
\]
Due to complete integrability of the cKdV hierarchy, the constraints
provided by \eqref{r1} are invariant with respect to all the flows
of the hierarchy. As a result, the infinite hierarchy \eqref{c11a}
reduces to the finite system \eqref{r2} described in the following
definition.
\begin{defn}
\label{maindef} The $n$-th stationary cKdV system is the system
consisting of the first $n$ evolution equations from the cKdV hierarchy
\eqref{c11a} together with its $(n+1)$-th stationary flow:
\begin{equation}
\bm{u}_{t_{1}}=\bm{K}_{1},\qquad\bm{u}_{t_{2}}=\bm{K}_{2},\qquad\ldots,\qquad\bm{u}_{t_{n}}=\bm{K}_{n},\qquad\bm{K}_{n+1}=0.\label{r2}
\end{equation}
\end{defn}
From the recursive formula \eqref{beautiful} one can observe that the
cumulative differential order (i.e.~the sum of differential orders
of all components) of the vector field $\bm{K}_{k}$ increases by
two as $k$ increases by one. Thus, the cumulative order of $(n+1)$-th
vector field $\bm{K}_{n+1}$ is equal $N+2n$, which means that the
vector field $\bm{K}_{n+1}$ depends on $2(N+n)$ jet variables. Since
the stationary condition \eqref{s15} provides $N$ independent constraints
it follows that the stationary manifold $\mathcal{M}_{n}$ is $(2n+N)$-dimensional.

From the integrability of the cKdV hierarchy \eqref{c11a} it
follows that the manifold $\mathcal{M}_{n}$ is invariant with
respect to all the flows of the hierarchy and thus all the vector
fields $\bm{K}_{r}$ in (\ref{r2}) are tangent to
$\mathcal{M}_{n}$. They still pairwise commute since they commute
on the ambient space $\mathcal{F}$. Note that also the higher
vector fields $\bm{K}_{n+2},\bm{K}_{n+3},\ldots$ properly reduce
to $\mathcal{M}_{n}$, however, we do not study these reductions in
this article.


The system (\ref{r2}) is the main object of our study. In \cite{1}
the authors studied the particular case of (\ref{r2}) for $N=1$,
that is the stationary KdV system.

\subsection{Hamiltonian foliations of $\mathcal{M}_{n}$\label{31S}}

The (differential) order of the system of differential constraints
\eqref{s15} can be lowered
by integrating them with respect to the spatial variable $x$. This procedure
provides us with a system of differential constraints parameterized
by $N$ integration constants. It turns out that it can be done on
$N+1$ different ways, due to the $(N+1)$ Hamiltonian structures
of cKdV hierarchy \eqref{mham}. In particular, for the $m$-th Hamiltonian
representation $\mathbb{B}_{m}\bm{\gamma}_{n+1-m}=0$ of the $(n+1)$-th
stationary cKdV flow \eqref{r1} we can see that on the stationary
manifold $\mathcal{M}_{n}$ the covector $\bm{\gamma}_{n+1-m}$ belongs
to the kernel of the respective Hamiltonian operator $\mathbb{B}_{m}$.
As result, the `integrated' constraints with respect to the $m$-th
Hamiltonian structure can be obtained requiring that $\bm{\gamma}_{n+1-m}$
fulfill the conditions from Proposition~\ref{propker} or Proposition~\ref{altpropker}.

Let us, by setting $\bm{\xi}=\bm{\gamma}_{n+1-m}$, where $\xi_{i}=P_{n-m+i}$,
in (\ref{d1}), define the following auxiliary functions:
\begin{subequations}
\begin{align}
f_{k} & :=f_{k,m}(\bm{\gamma}_{n+1-m})\equiv\sum_{i=0}^{k-1}\sum_{j=i+1}^{k}\mathcal{J}_{i}(P_{j+n-k},P_{i-j+n+1}),\label{fk}\\
g_{k} & :=g_{k,m}(\bm{\gamma}_{n+1-m})\equiv-2\sum_{i=k}^{N}\sum_{j=k}^{i}\mathcal{J}_{i}(P_{j-k},P_{i-j+n+1}),\label{gk}\\
\tilde{g}_{k} & :=\tilde{g}_{k,m}(\bm{\gamma}_{n+1-m})\equiv\sum_{i=k}^{N}\sum_{j=k}^{i}\mathcal{J}_{i}(P_{j+n-k+1},P_{i-j+n+1}),\label{ggk}
\end{align}
\end{subequations}
 where in each case $1\leqslant k\leqslant N$. Observe that they
do not depend on $m$.

Thus, by Proposition \ref{propker}, by setting $\bm{\xi}=\bm{\gamma}_{n+1-m}$
in (\ref{d3}) we find the following integrated form of the $(n+1)$-th
stationary cKdV flow \eqref{r1}:
\begin{subequations}\label{r3}
\begin{align}
f_{k} & =c_{k}\qquad\text{for}\qquad1\leqslant k\leqslant m,\label{r3a}\\
g_{k} & =c_{k}\qquad\text{for}\qquad m+1\leqslant k\leqslant N, \label{r3b}
\end{align}
\end{subequations}
 where $c_{1},\ldots,c_{N}$ are (arbitrary) integration constants.
Moreover, for cases that $n+m<N-1$, by Proposition \ref{altpropker},
setting $\bm{\xi}=\bm{\gamma}_{n+1-m}$ in (\ref{d5}) yields the
following alternative integrated form of the $(n+1)$-th stationary
cKdV flow \eqref{r1}:
\begin{subequations}\label{r5}
\begin{align}
f_{k} & =c_{k}\qquad\text{for}\qquad1\leqslant k\leqslant m,\label{r5a}\\
g_{k}+\tilde{g}_{k+n+1} & =c_{k}\qquad\text{for}\qquad m+1\leqslant k\leqslant N-n-1,\label{r5b}\\
g_{k} & =c_{k}\qquad\text{for}\qquad N-n\leqslant k\leqslant N,\label{r5c}
\end{align}
\end{subequations}
 Therefore, for each $m\in\{0,1,\ldots,N\}$ the above relations define
a $2n$-dimensional foliation of the stationary manifold $\mathcal{M}_{n}$,
parameterized by the vector $\bm{c}\equiv(c_{1},\ldots,c_{N})$. The leaves
of this foliation are given by
\begin{equation}
\mathcal{M}_{n,m}^{\bm{c}}:=\left\{ [\bm{u}]\in\mathcal{M}_{n}\,|\,\text{s.t.~(\ref{r3}) for \ensuremath{n+m\geqslant N-1} or (\ref{r5}) for \ensuremath{n+m<N-1}}\right\} ,\label{leaf}
\end{equation}
so that for each $m$:
\begin{equation}
\mathcal{M}_{n}\equiv\bigcup\limits _{\bm{c}\in\mathbb{R}^{N}}\mathcal{M}_{n,m}^{\bm{c}}.\label{fol}
\end{equation}
We will refer to this foliation as \emph{Hamiltonian foliation} of
$\mathcal{M}_{n}$. The case $n+m\geqslant N-1$ will be referred
to as the\emph{ generic }case while the case $n+m<N-1$ we will call
the \emph{non-generic }case.
\begin{rem}
The Hamiltonian foliation \eqref{fol} of the $n$-th stationary
manifold $\mathcal{M}_{n}$ could be defined in a simpler way, through
the leaves
\[
\mathcal{M}_{n,m}^{\bm{c}}:=\left\{ [\bm{u}]\in\mathcal{M}_{n}\,|\,\text{s.t.~(\ref{r3})}\right\} ,
\]
given only by the relations (\ref{r3}), i.e.~without introducing
the non-generic case. We will however use the definition
\eqref{leaf} which is motivated by later needs.
\end{rem}

\subsection{Stäckel foliations of $\mathcal{M}_{n}$\label{32S}}

The stationary manifold $\mathcal{M}_{n}$ can also be foliated in
a way that allows for representing a given cKdV
stationary system as a Stäckel system defined on leaves of this foliation. We will
therefore call this foliation (see its definition below) of $\mathcal{M}_{n}$
the \emph{Stäckel foliation}. In fact, we will construct $N+1$ different
Stäckel foliations of $\mathcal{M}_{n}$, one for each choice of $m\in\left\{ 0,\ldots,N\right\} $.

We start by observing that the $(n+1)$-th stationary flow \eqref{r1}
of the cKdV hierarchy can be written as
\begin{equation}
\mathbb{Q}_{t_{n+1}} = (\mathbb{P}_{n+1})_{x}\mathbb{Q}+\frac{1}{2}\mathbb{P}_{n+1}\mathbb{Q}_{x}+\frac{1}{4}(\mathbb{P}_{n+1})_{3x}\equiv J\mathbb{P}_{n+1}=0.\label{s1}
\end{equation}
Note that this equation contains not only the stationary flow $K_{n+1}=0$
but also the first $n$ equations of the infinite recursion (\ref{c5})
on $P_{k}$. The stationary condition \eqref{s1} can be integrated
once to the form:
\begin{equation}
-\frac{1}{8}\mathbb{P}_{n+1}(\mathbb{P}_{n+1})_{xx}+\frac{1}{16}(\mathbb{P}_{n+1})_{x}^{2}-\frac{1}{4}\mathbb{Q}\mathbb{P}_{n+1}^{\,2}=C(\lambda),\label{s2a}
\end{equation}
or in our shorthand notation, using \eqref{Jpisane}, as
\begin{equation}
\mathcal{J}(\mathbb{P}_{n+1},\mathbb{P}_{n+1})=C(\lambda),\label{s2b}
\end{equation}
where $C(\lambda)$ is an appropriate polynomial in $\lambda$ with
coefficients being integration constants that follow from the next
proposition. Below we investigate the left hand side of (\ref{s2b})
more thoroughly.
\begin{prop}
\label{prophk} The left-hand side of \eqref{s2b}, or \eqref{s2a},
takes the following explicit form
\begin{equation}
\mathcal{J}(\mathbb{P}_{n+1},\mathbb{P}_{n+1})\equiv\lambda^{2n+N}+\sum_{k=0}^{n+N-1}h_{k}\lambda^{k},\label{s13}
\end{equation}
where the coefficients $h_{k}$ are differential functions of $\bm{u}$
given by
\begin{equation}
h_{k}=\sum_{i=0}^{N}\sum_{j=i}^{k}\mathcal{J}_{i}(P_{n-k+j},P_{n+i-j}).\label{s14}
\end{equation}
For $n+N\leq k<2n+N$ the coefficients $h_{k}$ in (\ref{s14}) vanish
and $h_{2n+N}=1$.\footnotemark\footnotetext{It is worth here to
compare \eqref{s2b} with $\mathcal{J}(\lambda^{n}\mathcal{P},\lambda^{n}\mathcal{P})=\lambda^{2n+N}$,
which follows from \eqref{c8}.}
\end{prop}
\begin{proof}
For $k=n+1$ \eqref{c3} has the form
\[
\mathbb{P}_{n+1}\equiv\left[\lambda^{n}\mathcal{P}\right]_{\geqslant0}=\sum_{i=0}^{n}P_{n-i}\lambda^{i}.
\]
Thus,
\begin{align*}
\mathcal{J}(\mathbb{P}_{n+1},\mathbb{P}_{n+1}) & =\sum_{i=0}^{N}\sum_{j=0}^{n}\sum_{k=0}^{n}\mathcal{J}_{i}(P_{n-j},P_{n-k})\lambda^{i+j+k} = \sum_{i=0}^{N}\sum_{j=i}^{n+i}\sum_{k=j}^{n+j}\mathcal{J}_{i}(P_{n-k+j},P_{n+i-j})\lambda^{k}\\
 & \equiv\sum_{k=0}^{2n+N}\sum_{i=0}^{N}\sum_{j=i}^{k}\mathcal{J}_{i}(P_{n-k+j},P_{n+i-j})\lambda^{k}=\sum_{k=0}^{2n+N}h_{k}\lambda^{k},
\end{align*}
Taking into account that $P_{n-k-j}=0$ for $j<k-n$ and $P_{n+i-j}=0$ for
$j>n+i$, we obtain that for $k=2n+N$
\[
h_{2n+N}=\sum_{i=0}^{N}\sum_{j=i}^{2n+N}\mathcal{J}_{i}(P_{j-n-N},P_{n+i-j})=\mathcal{J}_{N}(P_{0},P_{0})=1,
\]
while for $n+N\leqslant k\leqslant2n+N-1$ we obtain
\[
h_{k}=\sum_{i=0}^{N}\sum_{j=k-n}^{n+i}\mathcal{J}_{i}(P_{n-k+j},P_{n+i-j}) = \sum_{i=0}^{N}\sum_{j=0}^{i+2n-k}\mathcal{J}_{i}(P_{j},P_{i-j+2n-k})=0,
\]
were we used the equality \eqref{c8a}. The remaining coefficients
$h_{k}$ are non-vanishing.
\end{proof}
From Proposition~\ref{prophk} it follows that
\[
C(\lambda)\equiv\lambda^{2n+N}+\sum_{k=0}^{n+N-1}\varepsilon_{k}\lambda^{k},
\]
where $\varepsilon_{k}$ are some integration constants. Due to
(\ref{s2b}) the functions \eqref{s14} are integrals of motion of
all the flows of the $n$-th stationary system \eqref{r2}. The
integration constants $\varepsilon_{k}$ are thus the values of
integrals of motion $h_{k}$, depending on particular initial
conditions.

In two lemmas below we analyze the functions $h_{k}$ more closely.
\begin{lem}
\label{lemmahkx} For $k<N$:
\begin{subequations}
\label{s16}
\begin{equation}
(h_{k})_{x}=-\frac{1}{2}\sum_{j=0}^{k}P_{n-k+j}\sum_{i=0}^{j}J_{i}P_{n+i-j}\equiv-\frac{1}{2}\sum_{j=0}^{k}P_{n-k+j}(\bm{K}_{n+1})_{j+1},\label{s16a}
\end{equation}
and if $k\geqslant N$:
\begin{equation}
(h_{k})_{x}=-\frac{1}{2}\sum_{j=0}^{N-1}P_{n-k+j}\sum_{i=j+1}^{N}J_{i}P_{n+i-j}\equiv-\frac{1}{2}\sum_{j=0}^{N-1}P_{n-k+j}(\bm{K}_{n+1})_{j+1},\label{s16b}
\end{equation}
moreover if additionally $k\geqslant n$:
\begin{equation}
(h_{k})_{x}=-\frac{1}{2}\sum_{j=k-n}^{N-1}P_{n-k+j}\sum_{i=j+1}^{N}J_{i}P_{n+i-j}\equiv-\frac{1}{2}\sum_{j=k-n}^{N-1}P_{n-k+j}(\bm{K}_{n+1})_{j+1}.\label{s16c}
\end{equation}
\end{subequations}
\end{lem}
\begin{proof}
Differentiating \eqref{s14} with respect to $x$ and using the relation
\eqref{s12b} one finds that
\begin{equation}
(h_{k})_{x}=-\frac{1}{4}\sum_{i=0}^{N}\sum_{j=i}^{k}\bigl[P_{n+i-j}J_{i}P_{n-k+j}+P_{n-k+j}J_{i}P_{n+i-j}\bigr]\equiv-\frac{1}{2}\sum_{i=0}^{N}\sum_{j=i}^{k}P_{n-k+j}J_{i}P_{n+i-j}.\label{s17}
\end{equation}
Thus, for $k<N$, using \eqref{c11b}, we obtain \eqref{s16a}. For
$k\geqslant N$ one finds that \eqref{s17} takes the form
\[
(h_{k})_{x}\equiv-\frac{1}{2}\sum_{j=0}^{N-1}P_{n-k+j}\sum_{i=0}^{j}J_{i}P_{n+i-j}-\frac{1}{2}\sum_{j=N}^{k}P_{n-k+j}\sum_{i=0}^{N}J_{i}P_{n+i-j},
\]
where the second term vanishes by the equality \eqref{c6}, now using
\eqref{c11b} we get \eqref{s16b}.
\end{proof}
We will now prove that the manifold $\mathcal{M}_{n}$ can be reconstructed,
on $N+1$ different ways, from appropriate subsets of the set of all functions
$h_{k}$. This is the content of the next theorem.
\begin{thm}
\label{intBm} Let us fix $m\in\{0,\ldots,N\}$. Then, the set of solutions of
the system of equations
\begin{subequations}
\label{s19}
\begin{equation}
h_{k}=c_{k+1}\qquad k=0,\ldots,m-1\label{s19a}
\end{equation}
and
\begin{equation}
h_{k}=c_{k-n+1},\qquad k=n+m,\ldots ,n+N-1.\label{s19b}
\end{equation}
\end{subequations}
 where all $c_{k}$ vary over $\mathbb{R}$, coincide with the stationary
manifold $\mathcal{M}_{n}$.
\end{thm}
When we fix the values of all $c_{k}$ then the equations (\ref{s19})
define a particular leaf of a $2n$-dimensional foliation of the stationary manifold $\mathcal{M}_{n}$.
This foliation is parameterized by the vector $\bm{c}\equiv(c_{1},\ldots,c_{N})$.
Therefore, for each $m\in\{0,1,\ldots,N\}$, we define
\begin{equation}
\bar{\mathcal{M}}_{n,m}^{\bm{c}}:=\left\{ [\bm{u}]\in\mathcal{M}_{n}\,|\,\text{s.t.~(\ref{s19})}\right\} ,\label{Mna}
\end{equation}
and then $\mathcal{M}_{n}$ is foliated into $\bar{\mathcal{M}}_{n,m}^{\bm{c}}$:
\[
\mathcal{M}_{n}\equiv\bigcup\limits _{\bm{c}\in\mathbb{R}^{N}}\bar{\mathcal{M}}_{n,m}^{\bm{c}}.
\]
Note that foliations (\ref{Mna}) for different $m$ are transversal
to each other. We will refer to the foliation (\ref{Mna}) as $m$-th
\emph{Stäckel foliation} of $\mathcal{M}_{n}$. By construction, the
leaves $\bar{\mathcal{M}}_{n,m}^{\bm{c}}$ are invariant with respect
to the evolution flows from the $n$-th stationary system \eqref{r2}.
This means that the mutually commuting vector fields $K_{k}$ from
the $n$-th stationary cKdV system \eqref{r2} are tangent to all the
leaves $\bar{\mathcal{M}}_{n,m}^{\bm{c}}$.
\begin{proof}
We have to show that for each $m\in\{0,1,\ldots,N\}$
the constraints \eqref{s19} imply the conditions \eqref{s15}. We
will show it by differentiating \eqref{s19} and using \eqref{s16}
taking into account the relation between $k$ and $N$. We need
to consider two cases.

Let us start with the generic case (i.e.~$n+m\geqslant N$). By \eqref{s19}:
\begin{subequations}
\begin{align}
h_{k} & =c_{k+1}\text{\ensuremath{\qquad}for}\qquad0\leqslant k\leqslant m-1,\label{c1a}\\
h_{k} & =c_{k-n+1}\text{ \ensuremath{\qquad}for}\qquad n+m\leqslant k\leqslant n+N-1.\label{c1b}
\end{align}
\end{subequations}
 Differentiating \eqref{c1a} and using \eqref{s16a} we obtain
the system
\[
(h_{k})_{x}\equiv-\frac{1}{2}\sum_{j=0}^{k}P_{n-k+j}(\bm{K}_{n+1})_{j+1}=0,\qquad0\leqslant k\leqslant m-1,
\]
which recursively implies that
\[
(\bm{K}_{n+1})_{j}=0\qquad\text{for}\qquad1\leqslant j\leqslant m.
\]
Similarly, differentiating \eqref{c1b} and using \eqref{s16c}
we obtain
\[
(h_{k})_{x}\equiv-\frac{1}{2}\sum_{j=k-n}^{N-1}P_{n-k+j}(\bm{K}_{n+1})_{j+1}=0,\qquad n+m\leqslant k\leqslant N+n-1,
\]
from which it follows that
\[
(\bm{K}_{n+1})_{j}=0\qquad\text{for}\qquad m+1\leqslant j\leqslant N.
\]
Thus, for $n+m\geqslant N$ the relations \eqref{s19} imply all the conditions \eqref{s15}.

In the non-generic case $n+m<N$ we first rewrite \eqref{s19}
in the form:
\begin{subequations}
\begin{align}
h_{k} & =c_{k+1}\qquad\text{for}\qquad0\leqslant k\leqslant m-1,\label{c2a}\\
h_{k} & =c_{k-n+1}\qquad\text{for}\qquad N\leqslant k\leqslant n+N-1.\label{c2b}\\
h_{k} & =c_{k-n+1}\qquad\text{for}\qquad n+m\leqslant k<N.\label{c2c}
\end{align}
\end{subequations}
 The relations \eqref{c2a} together with \eqref{s16a}, similarly as  \eqref{c1a},
imply that
\[
(\bm{K}_{n+1})_{j}=0\qquad\text{for}\qquad1\leqslant j\leqslant m.
\]
Differentiating \eqref{c2b} and using \eqref{s16c} we get the system
\[
(h_{k})_{x}\equiv-\frac{1}{2}\sum_{j=k-n}^{N-1}P_{n-k+j}(\bm{K}_{n+1})_{j+1}=0,\qquad N\leqslant k\leqslant N+n-1,
\]
which now implies that
\begin{equation}
(\bm{K}_{n+1})_{j}=0\qquad\text{for}\qquad N-n+1\leqslant j\leqslant N.\label{p2b}
\end{equation}
We obtain the remaining cases differentiating \eqref{c2c} and using,
this time, \eqref{s16a}:
\begin{equation}
(h_{k})_{x}\equiv-\frac{1}{2}\sum_{j=k-n}^{k}P_{n-k+j}(\bm{K}_{n+1})_{j+1}=0,\qquad n+m\leqslant k<N,\label{p2c}
\end{equation}
since $P_{n-k+j}\neq0$ only for $n-k+j\geqslant0$. Taking into account
(\ref{p2b}) the system (\ref{p2c}) implies that
\[
(\bm{K}_{n+1})_{j}=0\qquad\text{for}\qquad m+1\leqslant j\leqslant N-n.
\]
Thus, appropriately gathering the above cases we actually see that
also for $n+m\geqslant N$ the relations \eqref{s19} imply all the conditions
in \eqref{s15}.
\end{proof}
\begin{lem}
\label{lemmahk} The coefficients $h_{k}$ in \eqref{s14} satisfy
the following relations:
\begin{subequations}
\begin{align}
h_{k} & =f_{k+1}\qquad\text{for}\qquad k<N,\label{s18a}\\
h_{k} & =g_{k-n+1}\qquad\text{for}\qquad k\geqslant\max(N-1,n)\label{s18b}\\
h_{k} & =g_{k-n+1}+\tilde{g}_{k+2}\qquad\text{for}\qquad n\leqslant k<N-1.\label{s18c}
\end{align}
\end{subequations}
\end{lem}
For the proof see Appendix.

\bigskip{}

Notice that, by Lemma~\ref{lemmahk}, in the generic case $n+m\geqslant N-1$,
the set of constraints (\ref{s19}) takes the form:
\begin{align*}
f_{k} & \equiv h_{k-1}=c_{k}\qquad\text{for}\qquad1\leqslant k\leqslant m,\\
g_{k} & \equiv h_{k+n-1}=c_{k}\qquad\text{for}\qquad m+1\leqslant k\leqslant N.
\end{align*}
Thus we can see that in the generic case the set of constraints (\ref{s19})
is identical with the `integrated' set of constraints (\ref{r3})
associated with $m$-th Hamiltonian representation of the $(n+1)$-th
stationary cKdV flow \eqref{r1}. By the same lemma, in the non-generic case $n+m<N-1$,
the set of constraints (\ref{s19}) takes the form:
\begin{align*}
h_{k-1}\equiv f_{k} & =c_{k}\qquad\text{for}\qquad1\leqslant k\leqslant m,\\
h_{k+n-1}\equiv g_{k}+\tilde{g}_{k+n+1} & =c_{k}\qquad\text{for}\qquad m+1\leqslant k\leqslant N-n,\\
h_{k+n-1}\equiv g_{k} & =c_{k}\qquad\text{for}\qquad N-n+1\leqslant k\leqslant N,
\end{align*}
which is equivalent with (\ref{r5}). The above results lead to the following theorem.
\begin{thm}
\label{foliacje}The Hamiltonian foliation $\mathcal{M}_{n,m}^{\bm{c}}$
and the Stäckel foliation $\bar{\mathcal{M}}_{n,m}^{\bm{c}}$ coincide.
\end{thm}
Thus, in the sequel, we will only use the notation $\mathcal{M}_{n,m}^{\bm{c}}$
for leaves of this foliation.

\subsection{Stationary cKdV system on the leaves $\mathcal{M}_{n,m}^{\bm{c}}$.}

\label{33S}

The
constraint (\ref{r1}), defining the $n$-th stationary manifold $\mathcal{M}_{n}$,
can also be obtained by imposing an appropriate constraint on the
Lax hierarchy (\ref{l1}). To see this, let us impose the following constraint 
\begin{subequations}\label{s3ab}%
\begin{equation}
\Psi_{t_{n+1}}=\lambda^{m}\mu\Psi\label{s3a}
\end{equation}
or equivalently
\begin{equation}
\mathbb{V}_{n+1}\Psi=\lambda^{m}\mu\Psi \label{s3b}
\end{equation}
\end{subequations}
on the linear problems \eqref{l1}. The factor $\lambda^{m}\mu$ is chosen here in order to relate the
stationary system \eqref{r2} with an appropriate Stäckel system.
The parameters $\mu$ and $\lambda$ are spectral parameters of the
linear problems (\ref{l1}) and \eqref{s3ab}, and they
are assumed to be isospectral, i.e.~independent of all evolution variables.
Then, the constraint \eqref{s3a} and the compatibility condition
$(\Psi_{t_{k}})_{t_{n+1}}=(\Psi_{t_{n+1}})_{t_{k}}$ yields
\begin{equation}
\frac{d}{dt_{n+1}}\mathbb{V}_{k}=0,\text{ \ }k=1,2,\ldots\label{s4}
\end{equation}
which due to the form of $\mathbb{V}_{1}$ in (\ref{l11}) is equivalent
to the $(n+1)$-th stationary flow of the cKdV hierarchy \eqref{r1}.
As a consequence, the equations (\ref{l2}) and (\ref{l3}),
after imposing (\ref{s4}), yield the following Lax representation of the stationary
cKdV system (\ref{r2}) on $\mathcal{M}_{n}$:
\begin{equation*}
\frac{d}{dt_{k}}\mathbb{V}_{n+1}=[\mathbb{V}_{k},\mathbb{V}_{n+1}],\qquad k=1,\ldots.,n.
\end{equation*}
In consequence, the stationary cKdV system (\ref{r2}) on $\mathcal{M}_{n,m}^{\bm{c}}$
has the following Lax representation
\begin{equation}
\frac{d}{dt_{k}}\mathbb{V}_{n+1}^{(m)}=[\mathbb{V}_{k}^{(m)},\mathbb{V}_{n+1}^{(m)}],\qquad k=1,2,\ldots,n, \label{s5}
\end{equation}
where $\mathbb{V}_{k}^{(m)}$ originate by inserting the constraints
(\ref{s19}) into $\mathbb{V}_{k}$ given by (\ref{l11}). Note that now it is the matrix $\mathbb{V}_{n+1}$ (and respectively $\mathbb{V}_{n+1}^{(m)}$) that plays the role of the Lax matrix of the stationary cKdV system on $\mathcal{M}_{n}$ (and, respectively, on $\mathcal{M}_{n,m}^{\bm{c}}$), while the matrices $\mathbb{V}_{k}$ (and $\mathbb{V}_{k}^{(m)}$, respectively), for $k=1,\ldots,n$, play the role of auxiliary matrices.

Nontrivial solutions for the linear problem (\ref{s3b}) exist provided
that the characteristic equation
\begin{equation}
\det\left(\mathbb{V}_{n+1}-\lambda^{m}\mu\mathbb{I}\right)=0\label{s6}
\end{equation}
is satisfied. Explicitly written, the characteristic equation (\ref{s6})
takes the form
\begin{equation}
-\frac{1}{8}\mathbb{P}_{n+1}(\mathbb{P}_{n+1})_{xx}+\frac{1}{16}(\mathbb{P}_{n+1})_{x}^{2}-\frac{1}{4}\mathbb{Q}\mathbb{P}_{n+1}^{\,2}=\lambda^{2m}\mu^{2},\label{s7}
\end{equation}
where the left-hand side coincides with with the 'integrated' form
\eqref{s2a} of the $(n+1)$-th stationary flow (\ref{s1}). The equation
(\ref{s7}) can be written in our shorthand notation as
\begin{equation}
\mathcal{J}(\mathbb{P}_{n+1},\mathbb{P}_{n+1})=\lambda^{2m}\mu^{2}.\label{s8}
\end{equation}
By Proposition~\ref{prophk}, the relation \eqref{s7} (or equivalently \eqref{s8})
attains the form of a spectral curve
\begin{equation}
\lambda^{2n+N}+\sum_{k=0}^{n+N-1}h_{k}\lambda^{k}=\lambda^{2m}\mu^{2},\label{s9}
\end{equation}
where $h_{k}$ are given by \eqref{s14}. Since the functions $h_{k}$
depend not only on $x$ but also on all evolution parameters $t_{1},\ldots,t_{n}$,
it follows from (\ref{s9}) that they are in fact integrals of motion
of all the flows of the stationary cKdV system (\ref{r2}) on $\mathcal{M}_{n}$, as we mentioned earlier.

With each foliation $\mathcal{M}_{n,m}^{\bm{c}}$ (one for each $m \in \{1,\ldots,N\}$) we now associate
the curve (\ref{s9}) with appropriate $h_{k}$ fixed by (\ref{s19}),
that is a curve (depending on $N$ parameters $c_{k}$)
\begin{equation}
\lambda^{2n+N}+\sum_{k=1}^{N-m}c_{m+k}\lambda^{n+m+k-1}+\sum_{k=1}^{n}H_{k}\lambda^{n+m-k}+\sum_{k=1}^{m}c_{k}\lambda^{k-1}=\lambda^{2m}\mu^{2},\label{SC}
\end{equation}
where we use the notation
\[
H_{i}=h_{n+m-i},\qquad i=1,\ldots,n,
\]
that will prove to be useful later. Dividing (\ref{SC}) by $\lambda^{m}$
we obtain
\begin{equation}
\lambda^{2n+N-m}+\sum_{k=1}^{N-m}c_{m+k}\lambda^{n+k-1}+\sum_{k=1}^{n}H_{k}\lambda^{n-k}+\sum_{k=1}^{m}c_{k}\lambda^{k-m-1}=\lambda^{m}\mu^{2}.\label{SCP}
\end{equation}
As we will prove in the subsequent sections, if we now treat the curve
(\ref{SCP}) as a separation curve for a particular Stäckel system
of Benenti type, we will be able to construct a map between jet
variables on $\mathcal{M}_{n}$ and the separation variables of this
Stäckel system such that it preserves the dynamics (i.e.~it maps the solutions of one systems onto solutions of the other system).

\section{Stäckel systems of Benenti type and their Lax representation\label{4S}}

\subsection{Stäckel systems in separation variables}

This section contains some basic information about Stäckel
separable systems of Benenti type. Consider the separation curve
in the form (cf. (\ref{SCP}))
\begin{equation}
\sigma(\lambda)+\sum_{k=1}^{n}H_{k}\lambda^{n-k}=\lambda^{m}\mu^{2},\qquad m\in\mathbb{Z},\label{S1}
\end{equation}
where $\sigma(\lambda)$ is a (Laurent) polynomial in the variable
$\lambda$. The separable systems associated with \eqref{S1} belong
to the so-called Benenti subclass of Stäckel systems \cite{ben2,blaszak2007}.
Consider now the coordinates $(\bm{\lambda},\bm{\mu})$ on a phase
space $M=T^{\ast}Q$, where $\bm{\lambda}=(\lambda_{1},\ldots,\lambda_{n})^{T}$
are coordinates on an $n$-dimensional configuration space $Q$ and
$\bm{\mu}=(\mu_{1},\ldots,\mu_{n})^{T}$ are the fibre (momenta)
coordinates. Taking $n$ copies of \eqref{S1} at $n$ points $(\lambda_{i},\mu_{i})$
we obtain the system
\begin{equation}\label{S1b}
\sigma(\lambda_{i})+\sum_{k=1}^{n}H_{k}\lambda_{i}^{n-k}=\lambda_{i}^{m}\mu_{i}^{2},\qquad i=1,\ldots,n
\end{equation}
that is linear with respect to $H_{k}$. It can thus be easily solved with respect to these variables yielding $n$ functions $H_{k}=H_{k}(\bm{\lambda},\bm{\mu})$ on $M$.
In result, we obtain $n$ quadratic in momenta (Stäckel) Hamiltonians:
\begin{equation}
H_{k}=\frac{1}{2}\bm{\mu}^{T}A_{k}G_{m}\bm{\mu}+V_{k},\qquad k=1,\ldots,n,\label{S3}
\end{equation}
on the phase space $T^{\ast}Q$, where $G_{m}$ are treated as contravariant
metrics on the configuration space $Q$. Explicitly
\[
G_{m}=2\,\mathrm{diag}\left(\frac{\lambda_{1}^{m}}{\Delta_{1}},\ldots,\frac{\lambda_{n}^{m}}{\Delta_{n}}\right),\qquad\Delta_{i}=\prod_{j\neq i}(\lambda_{i}-\lambda_{j}).
\]
Further, all $A_{k}$ are $(1,1)$-Killing tensors for all the metrics
$G_{m}$ and the are given by
\[
A_{k}=(-1)^{k+1}\mathrm{diag}\left(\frac{\partial s_{k}}{\partial\lambda_{1}},\ldots,\frac{\partial s_{k}}{\partial\lambda_{n}}\right),\qquad k=1,\ldots,n,
\]
where $s_{k}$ denotes the elementary symmetric polynomial
in variables $\lambda_{i}$ of degree $k$, e.g.
\[
s_{1}=\sum_{i}\lambda_{i},\qquad s_{2}=\sum_{i<j}\lambda_{i}\lambda_{j},\qquad\ldots,\qquad s_{n}=\lambda_{1}\lambda_{2}\cdots\lambda_{n}.
\]
Notice that the Stäckel matrix $S$ associated with the linear system \eqref{S1b} is
the Vandermote matrix, $S_{ij} = \lambda_i^{n-j}$, see \cite{blaszak2009}, with the determinant
and its inverse given by
\begin{equation}\label{S5}
\det S = \prod_{i<j}(\lambda_{j}-\lambda_{i}),\qquad (S^{-1})_{ij} = \frac{(-1)^{i+1}}{\Delta_j}\frac{\partial s_i}{\partial \lambda_j}.
\end{equation}
Each metric $G_{m}$ in \eqref{S3} can be generated through $G_{m}=L^{m}G_{0}$,
where $L:=\mathrm{diag}(\lambda_{1},\ldots,\lambda_{n})$ is a special
conformal Killing tensor \cite{C2001}. Due to (\ref{S5}), the potential
functions $V_{k}$ in \eqref{S3} are given by
\begin{equation}
V_{k}=(-1)^{k+1}\sum_{i=1}^{n}\frac{\partial s_{k}}{\partial\lambda_{i}}\frac{\sigma(\lambda_{i})}{\Delta_{i}},\qquad k=1,\ldots,n.\label{S6}
\end{equation}
By construction, the Hamiltonians \eqref{S3} are in involution with
respect to the Poisson bracket
\[
\{f,g\}=\pi(df,dg),\qquad f,g\in C^{\infty}(M),
\]
with the Poisson bi-vector $\pi=\sum_{i}\frac{\partial}{\partial\lambda_{i}}\wedge\frac{\partial}{\partial\mu_{i}}$
(thus, $(\bm{\lambda},\bm{\mu})$ are Darboux coordinates for $\pi$).
The time evolution of any observable $\xi\in C^{\infty}(M)$ with
respect to the Hamiltonian $H_{k}$ has the form $\xi_{t_{k}}=\{\xi,H_{k}\}$
and the Hamiltonian evolution equations are
\begin{equation}
\bm{\lambda}_{t_{k}}=\{\bm{\lambda},H_{k}\}\equiv\frac{\partial H_{k}}{\partial\bm{\mu}},\qquad\bm{\mu}_{t_{k}}=\{\bm{\mu},H_{k}\}\equiv-\frac{\partial H_{k}}{\partial\bm{\lambda}},\qquad k=1,\ldots,n.\label{Hamr}
\end{equation}
By construction, the variables $(\bm{\lambda},\bm{\mu})$ are separation
variables for all the Stäckel Hamiltonians $H_{k}$ in (\ref{S3}).

\subsection{Stäckel systems in Viète coordinates}

Apart from the separation coordinates $(\bm{\lambda},\bm{\mu})$ we
will also work with Viète coordinates, defined by
\begin{equation}
q_{i}=(-1)^{i}s_{i},\qquad p_{i}=-\sum_{k=1}^{n}\frac{\lambda_{k}^{n-i}\mu_{k}}{\Delta_{k}},\qquad i=1,\ldots,n.\label{V1}
\end{equation}
The transformation (\ref{V1}) between the separation coordinates
and the Viète coordinates is a canonical transformation (since it
is a point transformation) so that $\pi=\sum_{i=1}^{n}\frac{\partial}{\partial q_{i}}\wedge\frac{\partial}{\partial p_{i}}$.
In Viète coordinates the Stäckel Hamiltonians \eqref{S3} take the form
\begin{equation}
H_{k}=\frac{1}{2}\bm{p}^{T}A_{k}G_{m}\bm{p}+V_{k},\qquad k=1,\ldots,n, \label{Hk}
\end{equation}
(where $\bm{p}=(p_{1},\ldots,p_{n})^{T}$ and $\bm{q}=(q_{1},\ldots,q_{n})^{T}$)
and the respective Hamiltonian equations attain the form
\begin{equation}
\bm{q}_{t_{k}}=\{\bm{q},H_{k}\}\equiv\frac{\partial H_{k}}{\partial\bm{p}},\qquad\bm{p}_{t_{k}}=\{\bm{p},H_{k}\}\equiv-\frac{\partial H_{k}}{\partial\bm{q}},\qquad k=1,\ldots,n.\label{Hamk}
\end{equation}

If $\sigma(\lambda)=\sum_{i}a_{i}\lambda^{i}$ the potential functions
\eqref{S6} are given by $V_{k}=\sum_{i}a_{i}\mathcal{V}_{k}^{(i)}$,
where the so-called elementary separable potentials $\mathcal{V}_{k}^{(i)}$
can be explicitly constructed from the recursion formula \cite{blaszak2011}
\begin{equation*}
\mathcal{V}^{(i)}=R^{i}\mathcal{V}^{(0)},\qquad\mathcal{V}^{(i)}=\bigl(\mathcal{V}_{1}^{(i)},\ldots,\mathcal{V}_{n}^{(i)}\bigr)^{T},\qquad\mathcal{V}^{(0)}=(0,\ldots,0,-1)^{T}, 
\end{equation*}
where
\begin{equation}
R=\begin{pmatrix}-q_{1} & 1 & 0 & 0\\
\vdots & 0 & \ddots & 0\\
\vdots & 0 & 0 & 1\\
-q_{n} & 0 & 0 & 0
\end{pmatrix},\qquad R^{-1}=\begin{pmatrix}0 & 0 & 0 & -\frac{1}{q_{n}}\\
1 & 0 & 0 & \vdots\\
0 & \ddots & 0 & \vdots\\
0 & 0 & 1 & -\frac{q_{n-1}}{q_{n}}
\end{pmatrix}.\label{V5}
\end{equation}

In Viète coordinates the metric $G_{0}$ has the form $(G_{0})^{ij}=2q_{i+j-n-1}$,
where, for convenience, we set $q_{0}=1$ and $q_{i}=0$ for $i<0$
or $i>n$. The metrics for arbitrary $m$ are given again by $G_{m}=L^{m}G_{0}$.
An interesting fact is that in Viète coordinates the special conformal
Killing tensor $L$ has the matrix representation identical to $R$
defined by \eqref{V5} \cite{blaszak2007}. Further, the Killing tensors
$A_{k}$, are in these coordinates given by
\[
(A_{k})_{j}^{i}=\begin{cases}
q_{i-j+k-1} & \text{if}\quad i\leqslant j\quad\text{and}\quad k\leqslant j,\\
-q_{i-j+k-1} & \text{if}\quad i>j\quad\text{and}\quad k>j,\\
0 & \text{otherwise},
\end{cases}
\]
where $k=1,\ldots,n$. Notice that $(A_{1})_{j}^{i}=\delta_{j}^{i}$
in any coordinate system.

Reversing the fiber part of the map (\ref{V1}) we find that
\[
\mu_{i}=\sum_{k=1}^{n}(-1)^{k}\frac{\partial s_{k}}{\partial\lambda_{i}}p_{k}.
\]
Hence, we can obtain the formula
\begin{equation}
\sum_{i=1}^{n}\frac{\lambda_{i}^{k}\mu_{i}}{\Delta_{i}}=-\sum_{j=1}^{n}\mathcal{V}_{j}^{(k)}p_{j},\qquad k\in\mathbb{Z},\label{nice}
\end{equation}
that we will be useful later. It follows directly by \eqref{S5} and
the fact that the elementary potentials in separation coordinates
are given by
\[
\mathcal{V}_{j}^{(k)}=(-1)^{j+1}\sum_{i=1}^{n}\frac{\partial s_{j}}{\partial\lambda_{i}}\frac{\lambda_{i}^{k}}{\Delta_{i}},
\]
which in turn is an immediate consequence of \eqref{S6}.

\subsection{Lax representation of Stäckel systems}

Let us now, following \cite{blaszak2019}, present the Lax formulation
of Stäckel systems of Benenti type. The results below are necessary
for the proof of the main theorem of this paper, Theorem \ref{great}.

The Hamiltonian evolution equations \eqref{Hamr} or \eqref{Hamk},
associated with the separation (spectral) curves \eqref{S1}, have the following
(isospectral) Lax representation
\begin{equation}
\frac{d}{dt_{k}}\mathbb{L}=\left[\mathbb{U}_{k},\mathbb{L}\right],\qquad k=1,\ldots,n,\label{T1}
\end{equation}
with $\mathbb{L}$ and $\mathbb{U}_{k}$ being $2\times2$ traceless
matrices depending rationally on the spectral parameter~$\lambda$.
The Lax matrix $\mathbb{L}$ has the form\footnote{In general the Lax matrices for Benenti systems are parametrized by
two arbitrary functions $f(\lambda)$ and $g(\lambda)$, see \cite{blaszak2019}.
In this paper we choose these functions as $g=\frac{1}{2}f=\lambda^{m}$.}
\begin{equation}
\mathbb{L}=\begin{pmatrix}\mathfrak{v} & \mathfrak{u}\\
\mathfrak{w} & -\mathfrak{v}
\end{pmatrix},\label{laxm}
\end{equation}
where $\mathfrak{u}$ is a function on $Q$ given in the separation coordinates $\bm{\lambda}$ by
\[
\mathfrak{u}:=\prod_{k=1}^{n}(\lambda-\lambda_{k})\equiv\lambda^{n}+\sum_{k=1}^{n}(-1)^{k}s_{k}\lambda^{n-k},
\]
while $\mathfrak{v}$ and $\mathfrak{w}$ are functions on $T^{\ast}Q$
given in the separation coordinates $(\bm{\lambda},\bm{\mu})$ by
\[
\mathfrak{v}:=\sum_{k=1}^{n}(-1)^{k+1}\left[\sum_{i=1}^{n}\frac{\partial s_{k}}{\partial\lambda_{i}}\frac{\lambda_{i}^{m}\mu_{i}}{\Delta_{i}}\right]\lambda^{n-k}
\]
and
\begin{equation}
\mathfrak{w}:=\frac{1}{\mathfrak{u}}\biggl[\lambda^{m}\Bigl(\sigma(\lambda)+\sum_{k=1}^{n}H_{k}\lambda^{n-k}\Bigr)-\mathfrak{v}^{2}\biggr].\label{T3c}
\end{equation}
In fact, $\mathfrak{w}$ is defined so that the spectral curve \eqref{S1}
can be reconstructed from the characteristic equation of $\mathbb{L}$:
\begin{equation}
0=\det\bigl[\mathbb{L}-\lambda^{m}\mu\mathbb{I}\bigr]=-(\mathfrak{v}^{2}+\mathfrak{u}\mathfrak{w})+\lambda^{2m}\mu^{2}=-\lambda^{m}\Bigl(\sigma(\lambda)+\sum_{k=1}^{n}H_{k}\lambda^{n-k}\Bigr)+\lambda^{2m}\mu^{2}.\label{T4}
\end{equation}
One can show that the expression in the quadratic bracket in \eqref{T3c}
factorizes so that $\mathfrak{w}$ takes the form of a Laurent polynomial
in $\lambda$:
\begin{equation}
\mathfrak{w}=\lambda^{m}\left[\frac{\sigma(\lambda)-\lambda^{-m}\mathfrak{v}^{2}}{\mathfrak{u}}\right]_{+}.\label{T5}
\end{equation}
Here, the operation $[\cdot]_{+}$ is defined as follows. Given an
analytic function $\psi$ and a (pure) polynomial $\mathfrak{u}$,
we define $\left[\frac{\psi}{\mathfrak{u}}\right]_{+}$ so that the
following decomposition holds:
\[
\psi=\left[\frac{\psi}{\mathfrak{u}}\right]_{+}\mathfrak{u}+r,
\]
where the (unique) remainder $r$ is a lower degree polynomial than
the polynomial $\mathfrak{u}$, see \cite{blaszak2019} for details.
In particular, when $\psi$ is a Laurent polynomial, we have
\begin{equation}
\left[\frac{\psi}{\mathfrak{u}}\right]_{+}\equiv\left[\frac{[\psi]_{\geqslant0}}{\mathfrak{u}}\right]_{\geqslant0}+\left[\frac{[\psi]_{<0}}{\mathfrak{u}}\right]_{<0},\label{T7}
\end{equation}
where $[\cdot]_{\geqslant0}$ is the projection on the part consisting
of non-negative degree terms in the expansion into its Laurent series
at $\infty$ and $[\cdot]_{<0}$ is the projection on the part consisting
of negative degree terms in the expansion into its Laurent series
at $0$.

The auxiliary matrices $\mathbb{U}_{k}$ are defined by
\begin{equation}
\mathbb{U}_{k}:=\left[\frac{\mathfrak{u}_{k}\mathbb{L}}{\mathfrak{u}}\right]_{+}\equiv\begin{pmatrix}\left[\frac{\mathfrak{u}_{k}\mathfrak{v}}{\mathfrak{u}}\right]_{+} & \mathfrak{u}_{k}\\
\left[\frac{\mathfrak{u}_{k}\mathfrak{w}}{\mathfrak{u}}\right]_{+} & -\left[\frac{\mathfrak{u}_{k}\mathfrak{v}}{\mathfrak{u}}\right]_{+}
\end{pmatrix},\qquad k=1,\ldots,n,\label{T8}
\end{equation}
where
\[
\mathfrak{u}_{k}:=\left[\frac{\mathfrak{u}}{\lambda^{n-k+1}}\right]_{+}\equiv\lambda^{k-1}+\sum_{i=1}^{k-1}(-1)^{k}s_{k}\lambda^{k-i-1}.
\]

The Lax equations \eqref{T1} (describing the evolution of the Lax
matrix \eqref{laxm} with respect to Hamiltonian equations \eqref{Hamr}),
can be derived from the evolution equations for $\mathfrak{u}$,$\mathfrak{v}$,$\mathfrak{w}$:
\begin{subequations}\label{T10}
\begin{align}
\mathfrak{u}_{t_{k}}\equiv\left\{ \mathfrak{u},H_{k}\right\}  & =-2\mathfrak{u}_{k}\mathfrak{v}+2\mathfrak{u}\left[\frac{\mathfrak{u}_{k}\mathfrak{v}}{\mathfrak{u}}\right]_{+},\label{T10a}\\
\mathfrak{v}_{t_{k}}\equiv\left\{ \mathfrak{v},H_{k}\right\}  & =\mathfrak{u}_{k}\mathfrak{w}-\mathfrak{u}\left[\frac{\mathfrak{u}_{k}\mathfrak{w}}{\mathfrak{u}}\right]_{+},\label{T10b}\\
\mathfrak{w}_{t_{k}}\equiv\left\{ \mathfrak{w},H_{k}\right\}  & =-2\mathfrak{w}\left[\frac{\mathfrak{u}_{k}\mathfrak{v}}{\mathfrak{u}}\right]_{+}+2\mathfrak{v}\left[\frac{\mathfrak{u}_{k}\mathfrak{w}}{\mathfrak{u}}\right]_{+},\label{T10c}
\end{align}
\end{subequations}
which were obtained in \cite{blaszak2019}.

Since $\mathfrak{u}_{1}=1$ and $\left[\frac{\mathfrak{v}}{\mathfrak{u}}\right]_{+}=0$,
the equations (\ref{T10a}) and (\ref{T10b}) for $k=1$ read
\[
\dot{\mathfrak{u}}=-2\mathfrak{v},\qquad\dot{\mathfrak{v}}=\mathfrak{w}+\mathfrak{u}\mathcal{Q},\qquad\mathcal{Q}:=-\left[\frac{\mathfrak{w}}{\mathfrak{u}}\right]_{+}.
\]
Here and in what follows, the dot means the derivative with respect
to the first Hamiltonian flow, i.e.~$\dot{\xi}\equiv\xi_{t_{1}}$.
Hence, we can rewrite the Lax matrix \eqref{laxm} as
\begin{equation}
\mathbb{L}=\begin{pmatrix}-\frac{1}{2}\dot{\mathfrak{u}} & \mathfrak{u}\\
-\frac{1}{2}\ddot{\mathfrak{u}}-\mathfrak{u}\mathcal{Q} & \frac{1}{2}\dot{\mathfrak{u}}
\end{pmatrix},\label{T11}
\end{equation}
(see also \cite{ee}). Further, observing that
\[
\dot{\mathfrak{u}}_{k}=\left[\frac{\mathfrak{u}_{k}\dot{\mathfrak{u}}}{\mathfrak{u}}\right]_{+}=-2\left[\frac{\mathfrak{u}_{k}\mathfrak{v}}{\mathfrak{u}}\right]_{+}
\]
and
\[
\ddot{\mathfrak{u}}_{k}=\left[\frac{\mathfrak{u}_{k}\ddot{\mathfrak{u}}}{\mathfrak{u}}\right]_{+}=-2\left[\frac{\mathfrak{u}_{k}\dot{\mathfrak{v}}}{\mathfrak{u}}\right]_{+}=-2\left[\frac{\mathfrak{u}_{k}\mathfrak{w}}{\mathfrak{u}}\right]_{+}-2\left[\frac{\mathfrak{u}_{k}\mathfrak{u}\mathcal{Q}}{\mathfrak{u}}\right]_{+}=-2\left[\frac{\mathfrak{u}_{k}\mathfrak{w}}{\mathfrak{u}}\right]_{+}-2\mathfrak{u}_{k}\mathcal{Q},
\]
we can rewrite the auxiliary matrices \eqref{T8} in the form
\begin{equation}
\mathbb{U}_{k}=\begin{pmatrix}-\frac{1}{2}\dot{\mathfrak{u}}_{k} & \mathfrak{u}_{k}\\
-\frac{1}{2}\ddot{\mathfrak{u}}_{k}-\mathfrak{u}_{k}\mathcal{Q} & \frac{1}{2}\dot{\mathfrak{u}}_{k}
\end{pmatrix},\qquad k=1,\ldots,n.\label{T12}
\end{equation}
(note that also in this notation $\mathbb{U}_{1}=\mathbb{L}$, as
it should be). The characteristic equation \eqref{T4} for $\mathbb{L}$
in the form \eqref{T11} is given by
\begin{equation}
-\frac{1}{2}\mathfrak{u}\ddot{\mathfrak{u}}+\frac{1}{4}\dot{\mathfrak{u}}^{2}-\mathfrak{u}^{2}\mathcal{Q}=\lambda^{2m}\mu^{2}.\label{T13}
\end{equation}
Using the form \eqref{T11} of $\mathbb{L}$ and the form \eqref{T12}
of $\mathbb{U}_{k}$ we see that the Lax equation \eqref{T1} , for
$k=1$ yields the following ODE:
\begin{equation}
\dddot{\mathfrak{u}}+4\dot{\mathfrak{u}}\mathcal{Q}+2\mathfrak{u}\dot{\mathcal{Q}}=0,\label{T14}
\end{equation}
while the remaining (i.e.~for $k=2,\ldots,n$) Lax equations in \eqref{T1}
yield, due to \eqref{T14}, the following hierarchy of PDE's:
\begin{align}
\mathfrak{u}_{t_{k}} & =\dot{\mathfrak{u}}\mathfrak{u}_{k}-\mathfrak{u}\dot{\mathfrak{u}}_{k},\label{T15}\\
\mathcal{Q}_{t_{k}} & =\frac{1}{2}\dddot{\mathfrak{u}}_{k}+2\dot{\mathfrak{u}}_{k}\mathcal{Q}+\mathfrak{u}_{k}\dot{\mathcal{Q}}. \label{T16}
\end{align}
Notice that the equation \eqref{T14} can also be obtained by differentiating
\eqref{T13} with respect to time~$t_{1}$. Besides, using \eqref{T10a},
equations \eqref{T15} are identically satisfied.

Now, comparing (\ref{T11}--\ref{T12}) and \eqref{l11} opens the
possibility of constructing a map between Stäckel systems defined
by the curve (\ref{SCP}) and an appropriate stationary cKdV system.
This map will be constructed in the next section.

The functions $\mathfrak{u}$, $\mathfrak{v}$ and $\mathfrak{w}$
can also be written in Viète's coordinates \eqref{V1}. The function
$\mathfrak{u}$ is given by
\[
\mathfrak{u}=\lambda^{n}+\sum_{k=1}^{n}q_{k}\lambda^{n-k},
\]
\begin{equation*}
\mathfrak{v}=-\frac{1}{2}\sum_{k=1}^{n}\Bigl[\sum_{l=1}^{n}(G_{m})^{kl}p_{l}\Bigr]\lambda^{n-k}\equiv-\frac{1}{2}\dot{\mathfrak{u}}. 
\end{equation*}
while the function $\mathfrak{w}$ can still be obtained from the
formula \eqref{T3c} or \eqref{T5}.

\section{Equivalent Stäckel representations of stationary cKdV systems on
$\mathcal{M}_{n}$\label{5S}}

We are finally in position to formulate and prove the main result
of this paper: the two seemingly distinct objects, stationary cKdV
systems and Stäckel systems of Benenti type, are in fact the same,
i.e.~represent the same finite dimensional integrable system written
in two different systems of coordinates. More precisely, we will prove
below that each stationary $N$-field cKdV system on $\mathcal{M}_{n}$
can be identified, on $N+1$ equivalent ways, with an appropriate
Stäckel system of Benenti type.

Let us thus fix $m\in\{0,1,\ldots,N\}$ and consider the separation
curve (\ref{SCP}),
\begin{equation}
\lambda^{2n+N-m}+\sum_{k=1}^{N-m}c_{m+k}\lambda^{n+k-1}+\sum_{k=1}^{n}H_{k}\lambda^{n-k}+\sum_{k=1}^{m}c_{k}\lambda^{k-m-1}=\lambda^{m}\mu^{2},\label{sc}
\end{equation}
associated with $m$-th foliation \eqref{fol} of the $n$-th stationary
cKdV system \eqref{r2}. Comparing \eqref{sc} with \eqref{S1} we
see that
\[
\sigma(\lambda)=\lambda^{2n+N-m}+\sum_{k=1}^{N-m}c_{m+k}\lambda^{n+k-1}+\sum_{k=1}^{m}c_{k}\lambda^{k-m-1}.
\]
The identification we want to achieve will be done through an appropriate
map between the jet variables $[\bm{u}]$ of the stationary cKdV system
and coordinates of the Stäckel system (\ref{sc}) defined on the (extended)
Poisson manifold of dimension $2n+N$, spanned by the coordinates
$(\bm{\lambda},\bm{\mu},\bm{c})$ and equipped with the Poisson bracket
$\pi=\sum_{i}\frac{\partial}{\partial\lambda_{i}}\wedge\frac{\partial}{\partial\mu_{i}}$
(so that $\bm{c}=(c_{1},\ldots,c_{N})$ are Casimir variables for
$\pi$). Thus, this identification will depend on the choice of $m\in\{0,1,\ldots,N\}$
and can therefore be obtained on $N+1$ different ways.

Let us now explicitly write the Stäckel system defined by the
curve (\ref{sc}). Solving (\ref{sc}) with respect to $H_{k}$ and
passing to (extended) Viète coordinates $(\bm{q},\bm{p},\bm{c})$
we obtain $n$ Stäckel Hamiltonians
$H_{k}=H_{k}(\bm{q},\bm{p},\bm{c})$, $k=1,\ldots,n$, of the form
(\ref{Hk}) with potentials $V_{k}$ given by
\[
V_{k}(\bm{q},\bm{c})=\mathcal{V}_{k}^{(2n+N-m)}+\sum_{i=1}^{N-m}c_{m+i}\mathcal{V}_{k}^{(n+i-1)}+\sum_{i=1}^{m}c_{k}\mathcal{V}_{k}^{(i-m-1)}.
\]
The Stäckel Hamiltonians $H_{k}$ on the extended phase space generate
the following Stäckel system
\begin{equation}
\bm{q}_{t_{k}}=\{\bm{q},H_{k}\}\equiv\frac{\partial H_{k}}{\partial\bm{p}},\qquad\bm{p}_{t_{k}}=\{\bm{p},H_{k}\}\equiv-\frac{\partial H_{k}}{\partial\bm{q}},\qquad\bm{c}_{t_{k}}=\{\bm{c},H_{k}\}\equiv0,\qquad k=1,\ldots,n.\label{StE}
\end{equation}

In what follows we will need the following lemma.
\begin{lem}
\label{beata}The element $\mathcal{Q}$ from the Lax matrices \eqref{T11}
and \eqref{T12} associated with the separation curve \eqref{sc}
is a polynomial of order $N$:
\begin{equation}
\mathcal{Q}\equiv-\left[\frac{\mathfrak{w}}{\mathfrak{u}}\right]_{+}=-\left[\frac{\mathfrak{w}}{\mathfrak{u}}\right]_{\geqslant0}=-\lambda^{N}+l.d.t.,\label{qq2}
\end{equation}
where $l.d.t.$ denotes the lower degree terms with coefficients being
functions of $(\bm{q},\bm{p},\bm{c})$.
\end{lem}
\begin{proof}
The relation \eqref{qq2} is a direct consequence of the identity
\eqref{T7} and the fact that $\mathfrak{w}$ is a polynomial.
\end{proof}
We are now in position to formulate and prove the main theorem of
this paper.
\begin{thm}
\label{great} For a given $m\in\{0,1,\ldots,N\}$, the transformation
between the jet variables $[\mathbf{{u}]}$ on the stationary manifold
$\mathcal{M}_{n}$ and the Viète coordinates $(\bm{q},\bm{p},\bm{c})$
given by
\begin{subequations}
\label{mapa}
\begin{equation}
q_{i}=\frac{1}{2}P_{i},\qquad p_{i}=\frac{1}{2}\sum_{j=1}^{n}\bigl(G_{m}^{-1}\bigr)_{ij}(P_{j})_{x},\qquad i=1,\ldots,n,\label{mapaa}
\end{equation}
and
\begin{equation}
c_{i}=h_{i-1},\qquad i=1,\ldots,m,\qquad c_{i}=h_{n+i-1},\qquad i=m+1,...,N,\label{mapac}
\end{equation}
\end{subequations}
 maps (after the identification $t_{1}\equiv x$) the $r$-th flow
$\bm{u}_{t_{r}}=\bm{K}_{r}$ of the stationary cKdV system (\ref{r2})
on $\mathcal{M}_{n}$ onto the $r$-th flow of the Stäckel system
(\ref{StE}). The metric $G_{m}$ in (\ref{mapa}) is expressed in
coordinates $q_{i}=q_{i}[\mathbf{{u}]}$ that are given by the first
formula in (\ref{mapaa}).
\end{thm}
\begin{proof}
Let us make the following identification
\begin{equation*}
\mathfrak{u}\equiv\frac{1}{2}\mathbb{P}_{n+1} 
\end{equation*}
between the variables $q_{i}$ and the jet variables $[\bm{u}]$.
This formula immediately yields the first part of the map \eqref{mapa}.
Extending it to the momenta part we immediately obtain the second
part (this part depends on $m$) of (\ref{mapa}). By Lemma \ref{beata}
also $\mathcal{Q}$ and $\mathbb{Q}$ must then coincide (the explicit
identification between $\mathcal{Q}$ and $\mathbb{Q}$ also depends
on $m$). Thus further, after the identification $t_{1}=x$ and by
comparing $\mathbb{L}$ in (\ref{T11}) with $\mathbb{V}_{n+1}$ in
(\ref{l11}) and $\mathbb{U}_{k}$ in (\ref{T12}) with respective
$\mathbb{V}_{k}$ in (\ref{l11}) we see that on the particular leaf
$\mathcal{M}_{n,m}^{\bm{c}}$ we must have
\[
\mathbb{L}\equiv\mathbb{V}_{n+1}^{(m)},\qquad\mathbb{U}_{k}\equiv\mathbb{V}_{k}^{(m)},\qquad k=1,\ldots,n,
\]
so that Lax formulations of both systems coincide. Moreover, (\ref{T13})
coincides with (\ref{s7}), (\ref{T14}) coincides with (\ref{s1})
and (\ref{T16}) coincides with (\ref{c2}). Finally, the map (\ref{mapac})
is given by (\ref{s19}). Thus, the whole map (\ref{mapa}) maps the
Stäckel system defined by the curve (\ref{sc}) to the corresponding
stationary cKdV system (\ref{r2}) on $\mathcal{M}_{n}$.
\end{proof}
\begin{rem}
Let us remark that the map (\ref{mapaa}) is defined on each leaf
$\mathcal{M}_{n,m}^{\bm{c}}$ while the whole map (\ref{mapa}) is
defined on the stationary manifold $\mathcal{M}_{n}$. Thus, \eqref{mapaa}
maps the $r$-th flow $\bm{u}_{t_{r}}=\bm{K}_{r}$ on the leaf $\mathcal{M}_{n,m}^{\bm{c}}$
onto the $r$-th flow of the Stäckel system (\ref{Hamk}) (we remind
the reader that each cKdV flow $\bm{u}_{t_{r}}=\bm{K}_{r}$ is tangent
not only to $\mathcal{M}_{n}$ but also to each leaf $\mathcal{M}_{n,m}^{\bm{c}}$
of the respective foliation).
\end{rem}

\section{Miura maps between Stäckel representations on $\mathcal{M}_{n}$\label{6S}}

In this chapter we prove that all $N+1$ Stäckel representations
(\ref{StE}) of the stationary cKdV system (\ref{r2}), when
considered on the whole stationary manifold $\mathcal{M}_{n}$ and
not only on the leaves $M_{n,m}^{\bm{c}}$ (that is on the extended
phase space when $\bm{c}$ are variables on $\mathcal{M}_{n}$
rather than parameters of the foliation $M_{n,m}^{\bm{c}}$) are
connected by appropriate Miura maps. To simplify the presentation,
we will only consider the case when one of the representations is
given by $m=0$ (the general case does not create any problems
whatsoever). Consider thus two Stäckel representations of the
same cKdV stationary system (\ref{r2}) on $\mathcal{M}_{n}$
generated by the curves
\begin{equation}
\lambda^{2n+N}+\sum_{k=1}^{N}c_{k}\lambda^{n+k-1}+\sum_{k=1}^{n}H_{k}\lambda^{n-k}=\mu^{2}\label{K0}
\end{equation}
and
\begin{equation}
\bar{\lambda}^{2n+N-m}+\sum_{k=1}^{N-m}\bar{c}_{m+k}\bar{\lambda}^{n+k-1}+\sum_{k=1}^{n}\bar{H}_{k}\bar{\lambda}^{n-k}+\sum_{k=1}^{m}\bar{c}_{k}\bar{\lambda}^{k-m-1}=\bar{\lambda}^{m}\bar{\mu}^{2},\label{Km}
\end{equation}
with $m\in\left\{ 1,\ldots,N\right\} $. We will thus consider the
Stäckel system generated by the curve (\ref{K0}) (it has the form
(\ref{StE})) in the extended Viète coordinates
$(\bm{q},\bm{p},\bm{c})$, with $(\bm{q},\bm{p})$ still given by
(\ref{V1}) and we will also consider the Stäckel system generated
by the curve (\ref{Km}) (it also has the form (\ref{StE})) in the
extended Viète coordinates
$(\bar{\bm{q}},\bar{\bm{p}},\bar{\bm{c}})$ where
$(\bar{\bm{q}},\bar{\bm{p}})$ are Viète coordinates defined by
$\bar{\lambda}_{k}$ and $\bar{\mu}_{k}$. The theorem below shows
that these systems are indeed two different parametrizations of
the same system, connected by a Miura map, as it should be, since
both are connected by two invertible maps (\ref{mapa}) with the
same stationary cKdV system (\ref{c11a}).
\begin{thm}
\label{Miury} The following Miura map on the stationary manifold
$\mathcal{M}_{n}$
\begin{align}
\bm{q} & =\bar{\bm{q}}\nonumber \\
\bm{p} & =\left(R^{T}\right)^{m}\bar{\bm{p}},\qquad\left[\left(R^{T}\right)^{m}\right]_{ij}=\mathcal{V}_{j}^{(n-i+m)}(\bar{\bm{q}}),\quad i,j=1,...,n,\label{mu1}\\
c_{i} & =\bar{H}_{m-i+1}(\bar{\bm{q}},\bar{\bm{p}},\bar{\bm{c}}),\qquad i=1,\ldots,m,\nonumber \\
c_{i} & =\bar{c}_{i}\text{, \ \ }i=m+1,\ldots,N\nonumber
\end{align}
transforms the Stäckel system defined by the curve (\ref{Km}) to
the Stäckel system generated by the curve (\ref{K0}). Similarly,
the inverse of this map, given by
\begin{align*}
\bar{\bm{q}} & =\bm{q}\nonumber \\
\bar{\bm{p}} & =\left(R^{T}\right)^{-m}\bm{p},\qquad\left[\left(R^{T}\right)^{m}\right]_{ij}=\mathcal{V}_{j}^{(n-i+m)}(\bm{q}),\quad i,j=1,...,n,\\ 
\bar{c}_{i} & =H_{n-i+1}(\bm{q},\bm{p},\bm{c}),\qquad i=1,\ldots,m\nonumber \\
\bar{c}_{i} & =c_{i},\qquad i=m+1,\ldots,N.\nonumber
\end{align*}
transforms the Stäckel system defined by the curve (\ref{K0}) to
the Stäckel system generated by the curve (\ref{Km}).
\end{thm}
\begin{proof}
For the fixed $m\in\{1,\ldots,N\}$, the map
\begin{equation}
\lambda=\bar{\lambda},\qquad\mu=\bar{\lambda}^{m}\bar{\mu}\label{pmap}
\end{equation}
transforms the curve (\ref{Km}) into the curve (\ref{K0}), provided
that
\begin{equation}
\begin{split}c_{i} & =\bar{c}_{i},\qquad i=m+1,\ldots,N\\
c_{i} & =\bar{H}_{m-i+1},\qquad i=1,\ldots,m\\
H_{i} & =\bar{H}_{m+i},\qquad i=1,\ldots,n-m\\
H_{i} & =\bar{c}_{n-i+1},\qquad i=n-m+1,\ldots,n.
\end{split}
\label{map}
\end{equation}
As result, the maps (\ref{pmap}) and (\ref{map}) induce the following
Miura map $\mathfrak{M}:\mathcal{M}_{n}\rightarrow\mathcal{M}_{n}$
\begin{equation}
\begin{split}\lambda_{i} & =\bar{\lambda}_{i},\qquad i=1,\ldots,n\\
\mu_{i} & =\bar{\lambda}_{i}^{m}\bar{\mu}_{i},\qquad i=1,\ldots,n\\
c_{i} & =\bar{H}_{m-i+1}(\bar{\bm{\lambda}},\bar{\bm{\mu}},\bar{\bm{c}}),\qquad i=1,\ldots,m\\
c_{i} & =\bar{c}_{i},\qquad i=m+1,\ldots.
\end{split}
\label{miura}
\end{equation}
Let us now write the Miura map (\ref{miura}) in Viète coordinates.
Since $\lambda_{i}=\bar{\lambda}_{i}$ we have $q_{i}=\bar{q}_{i}$.
Further, since $\mu_{i}=\bar{\lambda}_{i}^{m}\bar{\mu}_{i}$, we obtain
\[
p_{i}=-\sum_{k=1}^{n}\frac{\lambda_{k}^{n-i}\mu_{k}}{\Delta_{k}}=-\sum_{k=1}^{n}\frac{\bar{\lambda}_{k}^{n-i+m}\bar{\mu}_{k}}{\bar{\Delta}_{k}}=\sum_{k=1}^{n}\mathcal{V}_{k}^{(n-i+m)}(\bar{\bm{q}})\bar{p}_{k},
\]
where the last equality follows from (\ref{nice}). Alternatively,
we see that the first Hamiltonian flows of the corresponding Stäckel
systems are $\dot{\bm{q}}=G_{0}\bm{p}$ and $\dot{\bar{\bm{q}}}=G_{m}\bar{\bm{p}}$.\footnote{Note that we are not distinguishing between metrics defined in coordinates
$q_{i}$ and $\bar{q}_{i}$ because $q_{i}\equiv\bar{q}_{i}$, analogously
in the case of the matrix $R$.} Now, since $\dot{\bm{q}}=\dot{\bar{\bm{q}}}$ and $G_{m}=R^{m}G_{0}=G_{0}\left(R^{T}\right)^{m}$
we have
\[
\bm{p}=G_{0}^{-1}G_{m}\bar{\bm{p}}=\left(R^{T}\right)^{m}\bar{\bm{p}}.
\]
Rewriting the Hamiltonians $\bar{H}_{m}$ in the coordinates $(\bar{\bm{q}},\bar{\bm{p}},\bar{\bm{c}})$
completes expressing the Miura map (\ref{miura}) in Viéte coordinates.

The proof of the second statement of the theorem is analogous. The
relations (\ref{map}) can be inverted to
\begin{equation}
\begin{split}\bar{c}_{i} & =H_{n-i+1},\qquad i=1,\ldots,m\\
\bar{c}_{i} & =c_{i},\qquad i=m+1,\ldots,N\\
\bar{H}_{i} & =c_{m-i+1},\qquad i=1,\ldots,m\\
\bar{H}_{i} & =H_{i-m},\qquad i=m+1,\ldots,n.
\end{split}
\label{mapinv}
\end{equation}
The maps (\ref{pmap}) and (\ref{mapinv}) induce the inverse Miura
map $\mathfrak{M}^{-1}:\mathcal{M}_{n}\rightarrow\mathcal{M}_{n}$
given by
\begin{equation}
\begin{split}\bar{\lambda}_{i} & =\lambda_{i},\qquad i=1,\ldots,n\\
\bar{\mu}_{i} & =\lambda_{i}^{-m}\mu_{i},\qquad i=1,\ldots,n\\
\bar{c}_{i} & =H_{n-i+1}(\bm{\lambda},\bm{\mu},\bm{c}),\qquad i=1,\ldots,m\\
\bar{c}_{i} & =c_{i},\qquad i=m+1,\ldots,N,
\end{split}
\label{minv}
\end{equation}

The Miura map $\mathfrak{M}$ maps the Stäckel system generated by
the curve (\ref{Km}) into the Stäckel system generated by the
curve (\ref{K0}), while the map $\mathfrak{M}^{-1}$ maps this
system back into the system generated by (\ref{Km}).
\end{proof}

Obviously, $(\bm{q},\bm{p},\bm{c})$ are canonical coordinates with
respect to the Poisson bi-vector
$\pi=\sum_{i=1}^{n}\frac{\partial}{\partial
q_{i}}\wedge\frac{\partial}{\partial p_{i}}$ while the coordinates
$(\bar{\bm{q}},\bar{\bm{p}},\bar{\bm{c}})$ are canonical
coordinates with respect to the Poisson bi-vector
$\bar{\pi}=\sum_{i=1}^{n}\frac{\partial}{\partial\bar{q}_{i}}\wedge\frac{\partial}{\partial\bar{p}_{i}}$.
Since there exists a Miura map between (\ref{K0}) and (\ref{Km})\
for each $m\in\left\{ 1,\ldots,N\right\} $, we can construct, in a
standard way, \,$N+1$ Poisson bi-vectors for the Stäckel system
(\ref{K0}) on $\mathcal{M}_{n}$, each given by rewriting the
bi-vector $\bar{\pi}$ in the coordinates $(\bm{q},\bm{p},\bm{c})$,
yielding $N+1$ Hamiltonian representations of the Stäckel system
(\ref{K0}). Thus, we have proved that the $N+1$ Hamiltonian
representations of cKdV hierarchy leads to $N+1$ Hamiltonian
representations of its stationary system. Let us note that the
very existence of Miura maps between stationary systems (and hence
the multi-Hamiltonian structure of these systems) has been first
observed in \cite{ant87} for the KdV case.

\section{Examples}\label{7S}

In this final section we present examples illustrating all main ingredients of the
theory. Subsection \ref{P1} focuses on the DWW hierarchy (so that $N=2$) and contains two examples. The first example is for $n=2$ and all $m=0,1,2$ and is given with all details. In the second example of this subsection we present the case $n=3$ and $m=0$. In subsection \ref{ostatni} we illustrate our theory for the case $N=4$, $n=2$ and $m=0$. This example aims to show what happens in case we reduce the cKdV hierarchy to stationary manifolds $\mathcal{M}_{n}$ of low (in comparison
to the number $N$ of components) dimension.

\subsection{Dispersive Water Waves hierarchy} \label{P1}
  
Assume that $N=2$ and denote $\bm{u}:=(u_{0},u_{1})^{T}=(u,v)^{T}$.
Then the formulas from Section~\ref{S2} lead to the DWW hierarchy, the first members
of which are given by
\begin{subequations}\label{pole123}
\begin{equation}
\begin{pmatrix}u\\
v
\end{pmatrix}_{t_{1}}=\bm{K}_{1}\equiv\begin{pmatrix}u_{x}\\
v_{x}
\end{pmatrix},\qquad\begin{pmatrix}u\\
v
\end{pmatrix}_{t_{2}}=\bm{K}_{2}\equiv\begin{pmatrix}\frac{1}{2}vu_{x}+uv_{x}+\frac{1}{4}v_{3x}\\
u_{x}+\frac{3}{2}vv_{x}
\end{pmatrix},\label{pole12}
\end{equation}
\begin{equation}
\begin{pmatrix}u\\
v
\end{pmatrix}_{t_{3}}=\bm{K}_{3}\equiv\begin{pmatrix}\frac{3}{8}v^{2}u_{x}+\frac{3}{2}uvv_{x}+\frac{3}{2}uu_{x}+\frac{1}{4}u_{3x}+\frac{3}{8}vv_{3x}+\frac{9}{8}v_{x}v_{2x}\\
\frac{3}{2}vu_{x}+\frac{3}{2}uv_{x}+\frac{15}{8}v^{2}v_{x}+\frac{1}{4}v_{3x}
\end{pmatrix},\label{pole3}
\end{equation}
\begin{align}
 & \begin{pmatrix}u\\
v
\end{pmatrix}_{t_{4}}=\bm{K}_{4}\equiv\begin{pmatrix}(\bm{K}_{4})_1\\
\frac{15}{8}v^{2}u_{x}+\frac{15}{4}uvv_{x}+\frac{3}{2}uu_{x}+\frac{1}{4}u_{3x}+\frac{35}{16}v^{3}v_{x}+\frac{5}{8}vv_{3x}+\frac{5}{4}v_{x}v_{2x}
\end{pmatrix},\label{pole4}
\end{align}
\end{subequations}
where
\begin{align*}
(\bm{K}_{4})_1 &= \frac{3}{2}u^{2}v_{x}+\frac{5}{16}v^{3}u_{x}+\frac{15}{8}uv^{2}v_{x}+\frac{9}{4}uvu_{x}+\frac{3}{8}vu_{3x}+\frac{9}{8}u_{2x}v_{x}+\frac{5}{4}u_{x}v_{2x}+\frac{5}{8}uv_{3x}\\
&\quad+\frac{15}{32}v^{2}v_{3x}+\frac{45}{16}vv_{x}v_{2x}+\frac{15}{16}v_{x}^{3}+\frac{1}{16}v_{5x}
\end{align*}
 This hierarchy, according to (\ref{mham}), is three-hamiltonian
with the Hamiltonian operators given by
\begin{align*}
\mathbb{B}_{0} & =\begin{pmatrix}-\frac{1}{2}v\partial_{x}-\frac{1}{2}\partial_{x}v & \partial_{x}\\
\partial_{x} & 0
\end{pmatrix},\qquad\mathbb{B}_{1}=\begin{pmatrix}\frac{1}{4}\partial_{x}^{3}+\frac{1}{2}u\partial_{x}+\frac{1}{2}\partial_{x}u & 0\\
0 & \partial_{x}
\end{pmatrix},\\
\mathbb{B}_{2} & =\begin{pmatrix}0 & \frac{1}{4}\partial_{x}^{3}+\frac{1}{2}u\partial_{x}+\frac{1}{2}\partial_{x}u\\
\frac{1}{4}\partial_{x}^{3}+\frac{1}{2}u\partial_{x}+\frac{1}{2}\partial_{x}u & \frac{1}{2}v\partial_{x}+\frac{1}{2}\partial_{x}v
\end{pmatrix}.
\end{align*}
The cosymmetries $\bm{\gamma}_{k}=(P_{k},P_{k+1})^{T}$ are given
by
\begin{align*}
P_{0} & =2,\qquad P_{1}=v,\qquad P_{2}=u+\frac{3}{4}v^{2},\qquad P_{3}=\frac{3}{2}uv+\frac{5}{8}v^{3}+\frac{1}{4}v_{2x},\\
P_{4} & =\frac{3}{4}u^{2}+\frac{15}{8}uv^{2}+\frac{1}{4}u_{2x}+\frac{35}{64}v^{4}+\frac{5}{8}vv_{2x}+\frac{5}{16}v_{x}^{2},\\
P_{5} & =\frac{15}{8}u^{2}v+\frac{35}{16}uv^{3}+\frac{5}{8}vu_{2x}+\frac{5}{8}u_{x}v_{x}+\frac{5}{8}uv_{2x}+\frac{63}{128}v^{5}+\frac{35}{32}v^{2}v_{2x}+\frac{35}{32}vv_{x}^{2}+\frac{1}{16}v_{4x},\\
 & \vdots
\end{align*}
while the recursion operator (\ref{R}) attains the form
\[
\mathbb{R}=\left(\begin{array}{cc}
0 & \frac{1}{4}\partial_{x}^{2}+u+\frac{1}{2}u_{x}\partial_{x}^{-1}\\
1 & v+\frac{1}{2}v_{x}\partial_{x}^{-1}
\end{array}\right).
\]
Moreover, the Lax matrix $\mathbb{V}_{1}$ and the first three auxiliary
matrices $\mathbb{V}_{k}$ are as follows
\begin{align*}
\mathbb{V}_{1} & =\begin{pmatrix}0 & 1\\
\lambda^{2}-v\lambda-u & 0
\end{pmatrix},\qquad\mathbb{V}_{2}=\begin{pmatrix}-\frac{1}{4}v_{x} & \lambda+\frac{1}{2}v\\
\lambda^{3}-\frac{1}{2}v\lambda^{2}-\left(u+\frac{1}{2}v^{2}\right)\lambda-\frac{1}{2}uv-\frac{1}{4}v_{2x} & \frac{1}{4}v_{x}
\end{pmatrix},
\end{align*}
\begin{align*}
\mathbb{V}_{3} & =\begin{pmatrix}-\frac{1}{4}v_{x}\lambda-\frac{1}{4}u_{x}-\frac{3}{8}vv_{x} & \lambda^{2}+\frac{1}{2}v\lambda+\frac{3}{8}v^{2}+\frac{1}{2}u\\
\left(\mathbb{V}_{3}\right)_{21} & \frac{1}{4}v_{x}\lambda+\frac{1}{4}u_{x}+\frac{3}{8}vv_{x}
\end{pmatrix},
\end{align*}
 where
\[
(\mathbb{V}_{3})_{21}=\lambda^{4}-\frac{1}{2}v\lambda^{3}-\frac{1}{8}\left(4u+v^{2}\right)\lambda^{2}-\frac{1}{8}\left(8uv+3v^{3}+2v_{2x}\right)\lambda-\frac{1}{8}\left(4u^{2}+3uv^{2}+2u_{2x}+3v_{x}^{2}+3vv_{2x}\right),
\]
and
\[
\mathbb{V}_{4}=\begin{pmatrix}(\mathbb{V}_{4})_{11} & \lambda^{3}+\frac{1}{2}v\lambda^{2}+\frac{1}{8}\left(3v^{2}+4u\right)\lambda+\frac{1}{16}\left(5v^{3}+12uv+2v_{2x}\right)\\
(\mathbb{V}_{4})_{21} & -(\mathbb{V}_{4})_{11}
\end{pmatrix},
\]
where
\begin{align*}
(\mathbb{V}_{4})_{11} & =-\frac{1}{4}v_{x}\lambda^{2}-\frac{1}{8}\left(2u_{x}+3vv_{x}\right)\lambda-\frac{1}{32}\left(12vu_{x}+12uv_{x}+15v^{2}v_{x}+2v_{3x}\right),\\
(\mathbb{V}_{4})_{21} & =\lambda^{5}-\frac{1}{2}v\lambda^{4}-\frac{1}{8}\left(4u+v^{2}\right)\lambda^{3}-\frac{1}{16}\left(4uv+v^{3}+2v_{2x}\right)\lambda^{2}\\
 & \quad-\frac{1}{16}\left(8u^{2}+18uv^{2}+4u_{2x}+5v^{4}+8vv_{2x}+6v_{x}^{2}\right)\lambda\\
 & \quad-\frac{1}{16}\left(12u^{2}v+5uv^{3}+6vu_{2x}+12u_{x}v_{x}+8uv_{2x}+\frac{15}{2}v^{2}v_{2x}+15vv_{x}^{2}+v_{4x}\right).
\end{align*}

\subsubsection{The stationary reduction with $n=2$ and $m=0,1,2$} \label{dlugi}

Here we consider three Stäckel representations of the stationary DWW system for $n=2$.

\subsubsection*{Stationary system}

For $N=2$ and $n=2$ the stationary system (\ref{r2}) consists of two evolution equations \eqref{pole12}

\begin{equation}
\begin{pmatrix}u\\
v
\end{pmatrix}_{t_{1}}=\bm{K}_{1}\equiv\begin{pmatrix}u_{x}\\
v_{x}
\end{pmatrix},\qquad\begin{pmatrix}u\\
v
\end{pmatrix}_{t_{2}}=\bm{K}_{2}\equiv\begin{pmatrix}\frac{1}{2}vu_{x}+uv_{x}+\frac{1}{4}v_{3x}\\
u_{x}+\frac{3}{2}vv_{x}
\end{pmatrix}\label{pole12s}
\end{equation}
and of the stationary flow $\bm{K}_{3}=0$,
\[
\begin{pmatrix}\frac{3}{8}v^{2}u_{x}+\frac{3}{2}uvv_{x}+\frac{3}{2}uu_{x}+\frac{1}{4}u_{3x}+\frac{3}{8}vv_{3x}+\frac{9}{8}v_{x}v_{2x}\\
\frac{3}{2}vu_{x}+\frac{3}{2}uv_{x}+\frac{15}{8}v^{2}v_{x}+\frac{1}{4}v_{3x}
\end{pmatrix}=\begin{pmatrix}0\\
0
\end{pmatrix}.
\]
The normal form of the stationary flow $\bm{K}_{3}=0$ is
\begin{align*}
u_{3x} & =-6uu_{x}+\frac{15}{2}v^{2}u_{x}+3uvv_{x}+\frac{45}{4}v^{3}v_{x}-\frac{9}{2}v_{x}v_{2x}\\
v_{3x} & =-6vu_{x}-6uv_{x}-\frac{15}{2}v^{2}v_{x}
\end{align*}
and yields us the stationary manifold $\mathcal{M}_{2}$
parameterized by the jet variables
$[\bm{u}]=(u,u_{x},u_{2x},v,v_{x},v_{2x})$. The vector fields
(\ref{pole12s}) in this parametrization attain on
$\mathcal{M}_{2}$ the form
\begin{equation}
\begin{pmatrix}u\\
v
\end{pmatrix}_{t_{1}}=\begin{pmatrix}u_{x}\\
v_{x}
\end{pmatrix},\qquad\begin{pmatrix}u\\
v
\end{pmatrix}_{t_{2}}=\begin{pmatrix}-vu_{x}-\frac{1}{2}uv_{x}-\frac{15}{8}v^{2}v_{x}\\
u_{x}+\frac{3}{2}vv_{x}
\end{pmatrix}.\label{s2}
\end{equation}
The corresponding separation curve (\ref{s9}) attains the form
\[
\lambda^{6}+h_{3}\lambda^{3}+h_{2}\lambda^{2}+h_{1}\lambda+h_{0}=\lambda^{2m}\mu^{2},
\]
which yields the following $h_{k}$ in (\ref{s14}) on $\mathcal{M}_{2}$
\begin{align*}
h_{3} & =-\frac{3}{2}uv-\frac{1}{4}v_{2x}-\frac{5}{8}v^{3},\\
h_{2} & =-\frac{9}{8}uv^{2}-\frac{1}{4}u_{2x}-\frac{3}{4}u^{2}-\frac{1}{2}vv_{2x}-\frac{5}{16}v_{x}^{2}-\frac{15}{64}v,\\
h_{1} & =-\frac{1}{8}vu_{2x}+\frac{1}{8}u_{x}v_{x}-\frac{1}{8}uv_{2x}-\frac{3}{4}uv^{3}-\frac{3}{4}u^{2}v-\frac{9}{32}v^{2}v_{2x}-\frac{9}{64}v^{5},\\
h_{0} & =-\frac{3}{32}v^{2}u_{2x}+\frac{3}{16}vu_{x}v_{x}-\frac{3}{16}uvv_{2x}-\frac{3}{16}uv_{x}^{2}-\frac{9}{64}uv^{4}-\frac{3}{8}u^{2}v^{2}+\frac{1}{16}u_{x}^{2}-\frac{1}{8}uu_{2x}-\frac{1}{4}u^{3}-\frac{9}{64}v^{3}v_{2x}.
\end{align*}

\subsubsection*{Foliation for $m=0$}

Consider first the case $m=0$. Putting $h_{3}=c_{2}$ and $h_{2}=c_{1}$
we obtain the foliation of $\mathcal{M}_{2}$ into leaves $\mathcal{M}_{2,0}^{\bm{c}}$.
Solving these relations with respect to $u_{2x}$ and $v_{2x}$,
\begin{align*}
u_{2x}=8c_{2}v-4c_{1}-3u^{2}+\frac{15}{2}uv^{2}+\frac{65}{16}v^{4}-\frac{5}{4}v_{x}^{2},\qquad v_{2x}=-4c_{2}-6uv-\frac{5}{2}v^{3},
\end{align*}
we arrive at the curve (\ref{SCP}) for the leaves $\mathcal{M}_{2,0}^{\bm{c}}$.
It has the form
\begin{equation}
\lambda^{6}+c_{2}\lambda^{3}+c_{1}\lambda^{2}+H_{1}\lambda+H_{2}=\mu^{2}\label{k1}
\end{equation}
with $H_{i}$ given by
\begin{align*}
H_{1} & =\frac{3}{8}u^{2}v+\frac{5}{16}uv^{3}+\frac{1}{8}u_{x}v_{x}+\frac{7}{128}v^{5}+\frac{5}{32}vv_{x}^{2}+\frac{1}{2}c_{2}u+\frac{1}{8}c_{2}v^{2}+\frac{1}{2}c_{1}v,\\
H_{2} & =+\frac{1}{8}u^{3}+\frac{3}{32}u^{2}v^{2}-\frac{5}{128}uv^{4}+\frac{3}{16}vu_{x}v_{x}-\frac{1}{32}uv_{x}^{2}+\frac{1}{16}u_{x}^{2}-\frac{15}{512}v^{6}\\
 & \quad+\frac{15}{128}v^{2}v_{x}^{2}-\frac{1}{4}c_{2}uv+\frac{1}{2}c_{1}u-\frac{3}{16}c_{2}v^{3}+\frac{3}{8}c_{1}v^{2}.
\end{align*}
The Lax matrices $\mathbb{V}_{k}^{(0)}$ in (\ref{s5}) are as follows,
$\mathbb{V}_{1}^{(0)}=\mathbb{V}_{1}$,
\[
\mathbb{V}_{2}^{(0)}=\begin{pmatrix}-\frac{1}{4}v_{x} & \frac{1}{2}v+\lambda\\
\lambda^{3}-\frac{1}{2}v\lambda^{2}-\left(u+\frac{1}{2}v^{2}\right)\lambda+uv+\frac{5}{8}v^{3}+c_{2} & \frac{1}{4}v_{x}
\end{pmatrix}
\]
\[
\mathbb{V}_{3}^{(0)}=\begin{pmatrix}-\frac{1}{4}v_{x}\lambda-\frac{1}{4}u_{x}-\frac{3}{8}vv_{x} & \lambda^{2}+\frac{1}{2}v\lambda+\frac{3}{8}v^{2}+\frac{1}{2}u\\
\left(\mathbb{V}_{3}^{(0)}\right)_{21} & \frac{1}{4}v_{x}\lambda+\frac{1}{4}u_{x}+\frac{3}{8}vv_{x}
\end{pmatrix},
\]
where
\[
\left(\mathbb{V}_{3}^{(0)}\right)_{21}=\lambda^{4}-\frac{1}{2}v\lambda^{3}-\frac{1}{2}\left(u-\frac{1}{4}v^{2}\right)\lambda^{2}+\left(\frac{1}{2}uv+\frac{1}{4}v^{3}+c_{2}\right)\lambda-\frac{5}{64}v^{4}+\frac{1}{4}u^{2}-\frac{1}{16}v_{x}^{2}+c_{1}-\frac{1}{2}c_{2}v.
\]

\subsubsection*{Foliation for $m=1$}

For the case $m=1$, putting $h_{0}=c_{1}$ and $h_{3}=c_{2}$ we
obtain the foliation of $\mathcal{M}_{2}$ into leaves $\mathcal{M}_{2,1}^{\bm{c}}$.
Solving these relations with respect to $u_{2x}$ and $v_{2x}$,
\begin{align*}
u_{2x}=6c_{2}v-2u^{2}+\frac{15}{2}uv^{2}+\frac{15}{4}v^{4}+\frac{6vu_{x}v_{x}-6uv_{x}^{2}+2u_{x}^{2}-32c_{1}}{4u+3v^{2}},\qquad v_{2x}=-4c_{2}-6uv-\frac{5}{2}v^{3},
\end{align*}
we arrive at the curve (\ref{SCP}) for the leaves $\mathcal{M}_{2,1}^{\bm{c}}$,
\begin{equation}
\lambda^{5}+c_{2}\lambda^{2}+c_{1}\lambda^{-1}+H_{1}\lambda+H_{2}=\lambda\mu^{2},\label{k2}
\end{equation}
where
\begin{align*}
H_{1} & =\frac{5}{64}v^{4}-\frac{1}{4}u^{2}+\frac{1}{16}v_{x}^{2}-\frac{(2u_{x}+3vv_{x})^{2}}{8\left(4u+3v^{2}\right)}+\frac{8c_{1}}{4u+3v^{2}}+\frac{1}{2}c_{2}v,\\
H_{2} & =\frac{1}{32}v(v^{2}+2u)\left(4u+3v^{2}\right)-\frac{(2u_{x}+3vv_{x})(vu_{x}-2uv_{x})}{8\left(4u+3v^{2}\right)}+\frac{4c_{1}v}{4u+3v^{2}}+\frac{1}{2}c_{2}u+\frac{3}{8}c_{2}v^{2}.
\end{align*}
The Lax matrices $\mathbb{V}_{k}^{(1)}$ in (\ref{s5}) are as follows,
$\mathbb{V}_{1}^{(1)}=\mathbb{V}_{1}$,
\[
\mathbb{V}_{2}^{(1)}=\begin{pmatrix}-\frac{1}{4}v_{x} & \frac{1}{2}v+\lambda\\
\lambda^{3}-\frac{1}{2}v\lambda^{2}-\left(u+\frac{1}{2}v^{2}\right)\lambda+uv+\frac{5}{8}v^{3}+c_{2} & \frac{1}{4}v_{x}
\end{pmatrix}
\]
\[
\mathbb{V}_{3}^{(1)}=\begin{pmatrix}-\frac{1}{4}v_{x}\lambda-\frac{1}{4}u_{x}-\frac{3}{8}vv_{x} & \lambda^{2}+\frac{1}{2}v\lambda+\frac{3}{8}v^{2}+\frac{1}{2}u\\
\left(\mathbb{V}_{3}^{(1)}\right)_{21} & \frac{1}{4}v_{x}\lambda+\frac{1}{4}u_{x}+\frac{3}{8}vv_{x}
\end{pmatrix},
\]
where
\[
\left(\mathbb{V}_{3}^{(1)}\right)_{21}=\lambda^{4}-\frac{1}{2}v\lambda^{3}-\frac{1}{2}\left(u-\frac{1}{4}v^{2}\right)\lambda^{2}+\left(\frac{1}{2}uv+\frac{1}{4}v^{3}+c_{2}\right)\lambda-\frac{9v^{2}v_{x}^{2}+12vu_{x}v_{x}+4u_{x}^{2}-64c_{1}}{8(3v^{2}+4u)}.
\]

\subsubsection*{Foliation for $m=2$}

Finally, for $m=2$ we put $h_{0}=c_{1}$ and $h_{1}=c_{2}$ which
leads to the foliation $\mathcal{M}_{2,2}^{\bm{c}}$ of $\mathcal{M}_{2}$.
Solving these relations again with respect to $u_{2x}$ and $v_{2x}$,
\begin{align*}
u_{2x} & =\frac{48c_{2}v}{4u+3v^{2}}-\frac{2\left(4u+9v^{2}\right)\left(16c_{1}+3uv_{x}^{2}-u_{x}^{2}\right)-36v^{3}u_{x}v_{x}}{\left(4u+3v^{2}\right)^{2}}-2u^{2}+\frac{9}{2}uv^{2}+\frac{9}{4}v^{4},\\
v_{2x} & =\frac{256c_{1}v-12v^{2}u_{x}v_{x}-16vu_{x}^{2}+48uvv_{x}^{2}+16uu_{x}v_{x}}{\left(4u+3v^{2}\right)^{2}}-\frac{32c_{2}}{4u+3v^{2}}-4uv-\frac{3}{2}v^{3},
\end{align*}
we arrive at the curve (\ref{SCP}) for the leaves $\mathcal{M}_{2,2}^{\bm{c}}$.
It has the form
\begin{equation*}
\lambda^{4}+c_{2}\lambda^{-1}+c_{1}\lambda^{-2}+H_{1}\lambda+H_{2}=\lambda^{2}\mu^{2} 
\end{equation*}
while
\begin{align*}
H_{1} & =-\frac{\left(2uv_{x}-vu_{x}\right)\left(2u_{x}+3vv_{x}\right)}{\left(4u+3v^{2}\right)^{2}}-\frac{1}{4}v\left(2u+v^{2}\right)+\frac{8c_{2}}{4u+3v^{2}}-\frac{32c_{1}v}{\left(4u+3v^{2}\right)^{2}},\\
H_{2} & =\frac{\left(\left(3v^{2}-4u\right)v_{x}+4vu_{x}\right)^{2}}{16\left(4u+3v^{2}\right)^{2}}-\frac{\left(2u_{x}+3vv_{x}\right)^{2}}{8\left(4u+3v^{2}\right)}-\frac{1}{64}\left(4u+v^{2}\right)\left(4u+3v^{2}\right)+\frac{4c_{2}v}{4u+3v^{2}}+\frac{8c_{1}\left(4u+v^{2}\right)}{\left(4u+3v^{2}\right)^{2}}.
\end{align*}
The Lax matrices $\mathbb{V}_{k}^{(2)}$ in (\ref{s5}) are as follows,
$\mathbb{V}_{1}^{(2)}=\mathbb{V}_{1}$,
\[
\mathbb{V}_{2}^{(2)}=\begin{pmatrix}-\frac{1}{4}v_{x} & \lambda+\frac{1}{2}v\\
\lambda^{3}-\frac{1}{2}v\lambda^{2}-\left(u+\frac{1}{2}v^{2}\right)\lambda-\bm{\kappa} +\frac{1}{8}v(3v^{2}+4u) & \frac{1}{4}v_{x}
\end{pmatrix},
\]
\[
\mathbb{V}_{3}^{(2)}=\begin{pmatrix}-\frac{1}{4}v_{x}\lambda-\frac{1}{4}u_{x}-\frac{3}{8}vv_{x} & \lambda^{2}+\frac{1}{2}v\lambda+\frac{3}{8}v^{2}+\frac{1}{2}u\\
\lambda^{4}-\frac{1}{2}v\lambda^{3}-\frac{1}{2}\left(u+\frac{1}{4}v^{2}\right)\lambda^{2}
-\bm{\kappa}\,\lambda-\frac{1}{8}\frac{(3v\,v_{x}+2u_{x})^{2}}{3v^{2}+4u}+c_{1}\frac{8}{3v^{2}+4u} & \frac{1}{4}v_{x}\lambda+\frac{1}{4}u_{x}+\frac{3}{8}vv_{x}
\end{pmatrix},
\]
where
\[
\bm{\kappa} =\frac{(3v\,v_{x}+2u_{x})(2uv_{x}-vu_{x})}{(3v^{2}+4u)^{2}}+c_{1}\frac{32v}{(3v^{2}+4u)^{2}}-c_{2}\frac{8}{(3v^{2}+4u)}.
\]

\subsubsection*{Stäckel system for $m=0$}

Take first $m=0$. The separation curve (\ref{k1})
\[
\lambda^{6}+c_{2}\lambda^{3}+c_{1}\lambda^{2}+H_{1}\lambda+H_{2}=\mu^{2}
\]
yields the following Stäckel Hamiltonians in Viéte coordinates
\begin{subequations}
\label{1H}
\begin{align}
H_{1} & =p_{2}^{2}q_{1}+2p_{1}p_{2}+q_{1}^{5}-4q_{2}q_{1}^{3}+3q_{2}^{2}q_{1}-c_{2}q_{1}^{2}+c_{1}q_{1}+c_{2}q_{2},\\
H_{2} & =p_{2}^{2}q_{1}^{2}+2p_{1}p_{2}q_{1}-p_{2}^{2}q_{2}+p_{1}^{2}+q_{2}q_{1}^{4}-3q_{2}^{2}q_{1}^{2}+q_{2}^{3}-c_{2}q_{2}q_{1}+c_{1}q_{2},
\end{align}
\end{subequations}
 that in turn lead to the following Stäckel system
\begin{equation}
\begin{pmatrix}q_{1}\\
q_{2}\\
p_{1}\\
p_{2}
\end{pmatrix}_{t_{1}}=\begin{pmatrix}2p_{2}\\
2p_{2}q_{1}+2p_{1}\\
2c_{1}q_{1}-c_{1}-p_{2}^{2}-5q_{1}^{4}+12q_{2}q_{1}^{2}-3q_{2}^{2}\\
-c_{2}+4q_{1}^{3}-6q_{2}q_{1}
\end{pmatrix},\label{s01}
\end{equation}
\begin{equation}
\begin{pmatrix}q_{1}\\
q_{2}\\
p_{1}\\
p_{2}
\end{pmatrix}_{t_{2}}=\begin{pmatrix}2p_{2}q_{1}+2p_{1}\\
2p_{2}q_{1}^{2}+2p_{1}q_{1}-2p_{2}q_{2}\\
c_{2}q_{2}-2p_{2}^{2}q_{1}-2p_{1}p_{2}-4q_{2}q_{1}^{3}+6q_{2}^{2}q_{1}\\
c_{2}q_{1}-c_{1}+p_{2}^{2}-q_{1}^{4}+6q_{2}q_{1}^{2}-3q_{2}^{2}
\end{pmatrix}.\label{s02}
\end{equation}
The map (\ref{mapa}) attains the form given by
\begin{subequations}
\label{m1}
\begin{align}
q_{1} & =\frac{1}{2}v,\qquad q_{2}=\frac{1}{2}u+\frac{3}{8}v^{2},\qquad p_{1}=\frac{1}{4}u_{x}+\frac{1}{4}vv_{x},\qquad p_{2}=\frac{1}{4}v_{x},
\end{align}
and by
\begin{align}
c_{1} & =h_{2}\equiv-\frac{9}{8}uv^{2}-\frac{1}{4}u_{2x}-\frac{3}{4}u^{2}-\frac{1}{2}vv_{2x}-\frac{5}{16}v_{x}^{2}-\frac{15}{64}v^{4},\qquad c_{2}=h_{3}\equiv-\frac{3}{2}uv-\frac{1}{4}v_{2x}-\frac{5}{8}v^{3},
\end{align}
\end{subequations}
 It maps the first two flows in (\ref{s2}) onto the first two components
in (\ref{s01}) and (\ref{s02}), respectively. The remaining two
components become identities on $\mathcal{M}_{2}$ due to the second
part of the map (\ref{m1}). The Lax matrices $\mathbb{V}_{3}^{(0)}$,
$\mathbb{V}_{1}^{(0)}$ and $\mathbb{V}_{2}^{(0)}$ transforms by
(\ref{m1}) respectively onto $\mathbb{L}$, $\mathbb{U}_{1}$ and
$\mathbb{U}_{2}$ given by
\begin{align*}
\mathbb{L}=\begin{pmatrix}-p_{2}\lambda-p_{1}-p_{2}q_{1} & \lambda^{2}+q_{1}\lambda+q_{2}\\
\lambda^{4}-q_{1}\lambda^{3}+\left(q_{1}^{2}-q_{2}\right)\lambda^{2}+\left(-q_{1}^{3}+c_{2}+2q_{1}q_{2}\right)\lambda + \bm{\kappa} & p_{2}\lambda+p_{1}+p_{2}q_{1}
\end{pmatrix},
\end{align*}
\[
\mathbb{U}_{1}=\begin{pmatrix}0 & 1\\
\lambda^{2}-2q_{1}\lambda+3q_{1}^{2}-2q_{2} & 0
\end{pmatrix},\qquad\mathbb{U}_{2}=\begin{pmatrix}-p_{2} & \lambda+q_{1}\\
\lambda^{3}-q_{1}\lambda^{2}+\left(q_{1}^{2}-2q_{2}\right)\lambda-q_{1}^{3}+4q_{1}q_{2}+c_{2} & p_{2}
\end{pmatrix},
\]
where
\[
\bm{\kappa} = q_{1}^{4}-p_{2}^{2}+q_{2}^{2}-3q_{1}^{2}q_{2}+c_{1}-c_{2}q_{1}.
\]
The inverse of the map \eqref{m1} is given by
\begin{align*}
u=2q_{2}-3q_{1}^{2},\qquad u_{x}=4p_{1}-8p_{2}q_{1},\qquad v=2q_{1},\qquad v_{x}=4p_{2}
\end{align*}
and
\begin{align*}
u_{2x}=16c_{2}q_{1}-4c_{1}-20p_{2}^{2}-52q_{1}^{4}+96q_{2}q_{1}^{2}-12q_{2}^{2},\qquad v_{2x}=-4c_{2}+16q_{1}^{3}-24q_{2}q_{1}.
\end{align*}

\subsubsection*{Stäckel system for $m=1$}

For $m=1$ the curve (\ref{k2}) is
\[
\lambda^{5}+c_{2}\lambda^{2}+c_{1}\lambda^{-1}+H_{1}\lambda+H_{2}=\lambda\mu^{2}
\]
yielding the following Stäckel Hamiltonians in Viéte coordinates
\begin{subequations}
\label{2H}
\begin{align}
H_{1} & =-p_{2}^{2}q_{2}+p_{1}^{2}-q_{1}^{4}+3q_{2}q_{1}^{2}-q_{2}^{2}+c_{1}\frac{1}{q_{2}}+c_{2}q_{1},\\
H_{2} & =-p_{2}^{2}q_{2}q_{1}-2p_{1}p_{2}q_{2}-q_{2}q_{1}^{3}+2q_{2}^{2}q_{1}+c_{1}\frac{q_{1}}{q_{2}}+c_{2}q_{2},
\end{align}
\end{subequations}
 giving rise to the following Stäckel system
\begin{equation}
\begin{pmatrix}q_{1}\\
q_{2}\\
p_{1}\\
p_{2}
\end{pmatrix}_{t_{1}}=\begin{pmatrix}2p_{1}\\
-2p_{2}q_{2}\\
4q_{1}^{3}-6q_{2}q_{1}-c_{2}\\
p_{2}^{2}-3q_{1}^{2}+2q_{2}+c_{1}\frac{1}{q_{2}^{2}}
\end{pmatrix},\label{s11}
\end{equation}
\begin{equation}
\begin{pmatrix}q_{1}\\
q_{2}\\
p_{1}\\
p_{2}
\end{pmatrix}_{t_{2}}=\begin{pmatrix}-2p_{2}q_{2}\\
-2p_{1}q_{2}-2p_{2}q_{1}q_{2}\\
p_{2}^{2}q_{2}-2q_{2}^{2}+3q_{1}^{2}q_{2}-c_{1}\frac{1}{q_{2}}\\
p_{2}^{2}q_{1}+2p_{1}p_{2}+q_{1}^{3}-4q_{2}q_{1}+c_{1}\frac{q_{1}}{q_{2}^{2}}-c_{2}
\end{pmatrix}.\label{s12}
\end{equation}
The map (\ref{mapa}) is given by
\begin{subequations}
\label{m2}
\begin{align}
q_{1} & =\frac{1}{2}v,\qquad q_{2}=\frac{1}{2}u+\frac{3}{8}v^{2},\qquad p_{1}=\frac{1}{4}v_{x},\qquad p_{2}=-\frac{2u_{x}+3vv_{x}}{4u+3v^{2}},
\end{align}
together with
\begin{align}
\begin{split}c_{1}=h_{0} & \equiv-\frac{3}{32}v^{2}u_{2x}+\frac{3}{16}vu_{x}v_{x}-\frac{3}{16}uvv_{2x}-\frac{3}{16}uv_{x}^{2}-\frac{9}{64}uv^{4}-\frac{3}{8}u^{2}v^{2}+\frac{1}{16}u_{x}^{2}\\
 & \quad-\frac{1}{8}uu_{2x}-\frac{1}{4}u^{3}-\frac{9}{64}v^{3}v_{2x},\\
c_{2}=h_{3} & \equiv-\frac{3}{2}uv-\frac{1}{4}v_{2x}-\frac{5}{8}v^{3},
\end{split}
\end{align}
\end{subequations}
 and it maps the first two flows in (\ref{s2}) onto the first two
components of (\ref{s11}) and (\ref{s12}), respectively. The remaining
two components become identities on $\mathcal{M}_{2}$ due to (\ref{m2}).
The Lax matrices $\mathbb{V}_{3}^{(1)}$, $\mathbb{V}_{1}^{(1)}$
and $\mathbb{V}_{2}^{(1)}$ transforms by (\ref{m2}) respectively
onto $\mathbb{L}$, $\mathbb{U}_{1}$ and $\mathbb{U}_{2}$ given
by
\[
\mathbb{L}=\begin{pmatrix}p_{2}q_{2}-p_{1}\lambda & \lambda^{2}+q_{1}\lambda+q_{2}\\
\lambda^{4}-q_{1}\lambda^{3}+\left(q_{1}^{2}-q_{2}\right)\lambda^{2}-\left(q_{1}^{3}-c_{2}-2q_{1}q_{2}\right)\lambda-p_{2}^{2}q_{2}+\frac{c_{1}}{q_{2}} & \lambda p_{1}-p_{2}q_{2}
\end{pmatrix},
\]
\[
\mathbb{U}_{1}=\begin{pmatrix}0 & 1\\
\lambda^{2}-2q_{1}\lambda+3q_{1}^{2}-2q_{2} & 0
\end{pmatrix},\qquad\mathbb{U}_{2}=\begin{pmatrix}-p_{1} & \lambda+q_{1}\\
\lambda^{3}-q_{1}\lambda^{2}+\left(q_{1}^{2}-2q_{2}\right)\lambda-q_{1}^{3}+c_{2}+4q_{1}q_{2} & p_{1}
\end{pmatrix}.
\]
The inverse of the map \eqref{m2} is given by
\begin{align*}
u=2q_{2}-3q_{1}^{2},\qquad u_{x}=-12p_{1}q_{1}-4p_{2}q_{2},\qquad v=2q_{1},\qquad v_{x}=4p_{1}
\end{align*}
together with
\begin{align*}
u_{2x}=12c_{2}q_{1}-\frac{4c_{1}}{q_{2}}+4p_{2}^{2}q_{2}-24p_{1}^{2}-48q_{1}^{4}+84q_{2}q_{1}^{2}-8q_{2}^{2},\qquad v_{2x}=-4c_{2}+16q_{1}^{3}-24q_{2}q_{1}.
\end{align*}

\subsubsection*{Stäckel system for $m=2$}

Finally, for $m=2$ the curve (\ref{k2}) is
\[
\lambda^{4}+c_{2}\lambda^{-1}+c_{1}\lambda^{-2}+H_{1}\lambda+H_{2}=\lambda^{2}\mu^{2},
\]
yielding the Stäckel Hamiltonians
\begin{subequations}
\label{3H}
\begin{align}
H_{1} & =-p_{1}^{2}q_{1}-2p_{1}p_{2}q_{2}+q_{1}^{3}-2q_{2}q_{1}-c_{1}\frac{q_{1}}{q_{2}^{2}}+c_{2}\frac{1}{q_{2}},\\
H_{2} & =-q_{2}p_{1}^{2}+p_{2}^{2}q_{2}^{2}-q_{2}^{2}+q_{1}^{2}q_{2}+c_{1}\left(\frac{1}{q_{2}}-\frac{q_{1}^{2}}{q_{2}^{2}}\right)+c_{2}\frac{q_{1}}{q_{2}},
\end{align}
\end{subequations}
 that generate the following Stäckel system:
\begin{equation}
\begin{pmatrix}q_{1}\\
q_{2}\\
p_{1}\\
p_{2}
\end{pmatrix}_{t_{1}}=\begin{pmatrix}-2p_{1}q_{1}-2p_{2}q_{2}\\
-2p_{1}q_{2}\\
p_{1}^{2}-3q_{1}^{2}+2q_{2}+c_{1}\frac{1}{q_{2}^{2}}\\
2p_{1}p_{2}+2q_{1}+c_{2}\frac{1}{q_{2}^{2}}-2c_{1}\frac{q_{1}}{q_{2}^{3}}
\end{pmatrix},\label{s21}
\end{equation}
\begin{equation}
\begin{pmatrix}q_{1}\\
q_{2}\\
p_{1}\\
p_{2}
\end{pmatrix}_{t_{2}}=\begin{pmatrix}-2p_{1}q_{2}\\
2p_{2}q_{2}^{2}\\
-2q_{1}q_{2}+2c_{1}\frac{q_{1}}{q_{2}^{2}}-c_{2}\frac{1}{q_{2}}\\
-2p_{2}^{2}q_{2}+p_{1}^{2}-q_{1}^{2}+2q_{2}+c_{1}\left(\frac{1}{q_{2}^{2}}-2\frac{q_{1}^{2}}{q_{2}^{3}}\right)+c_{2}\frac{q_{1}}{q_{2}^{2}}
\end{pmatrix}.\label{s22}
\end{equation}
The map (\ref{mapa}) is given by
\begin{subequations}
\label{m3}
\begin{align}
q_{1} & =\frac{1}{2}v,\qquad q_{2}=\frac{1}{2}u+\frac{3}{8}v^{2},\qquad p_{1}=-\frac{2u_{x}+3vv_{x}}{4u+3v^{2}},\qquad p_{2}=-\frac{2\left(-4vu_{x}+4uv_{x}-3v^{2}v_{x}\right)}{\left(4u+3v^{2}\right)^{2}},
\end{align}
and by
\begin{align}
\begin{split}c_{1}=h_{0} & \equiv-\frac{3}{32}v^{2}u_{2x}+\frac{3}{16}vu_{x}v_{x}-\frac{3}{16}uvv_{2x}-\frac{3}{16}uv_{x}^{2}-\frac{9}{64}uv^{4}-\frac{3}{8}u^{2}v^{2}+\frac{1}{16}u_{x}^{2}\\
 & \quad-\frac{1}{8}uu_{2x}-\frac{1}{4}u^{3}-\frac{9}{64}v^{3}v_{2x},\\
c_{2}=h_{1} & \equiv-\frac{1}{8}vu_{2x}+\frac{1}{8}u_{x}v_{x}-\frac{1}{8}uv_{2x}-\frac{3}{4}uv^{3}-\frac{3}{4}u^{2}v-\frac{9}{32}v^{2}v_{2x}-\frac{9}{64}v^{5},
\end{split}
\end{align}
\end{subequations}
 and it maps the first two flows in (\ref{s2}) onto the first two
components of (\ref{s21}) and (\ref{s22}), respectively. The remaining
two components become identities on $\mathcal{M}_{2}$ due to (\ref{m3}).
The Lax matrices $\mathbb{V}_{3}^{(2)}$, $\mathbb{V}_{1}^{(2)}$
and $\mathbb{V}_{2}^{(2)}$ transforms by (\ref{m2}) respectively
onto $\mathbb{L}$, $\mathbb{U}_{1}$ and $\mathbb{U}_{2}$ given
by
\[
\mathbb{L}=\begin{pmatrix}\left(q_{1}p_{1}+q_{2}p_{2}\right)\lambda+q_{2}p_{1} & \lambda^{2}+q_{1}\lambda+q_{2}\\
\lambda^{4}-q_{1}\lambda^{3}+\left(q_{1}^{2}-q_{2}\right)\lambda^{2}-\left(q_{1}p_{1}^{2}+2q_{2}p_{1}p_{2}-c_{2}\frac{1}{q_{2}}+c_{1}\frac{q_{1}}{q_{2}^{2}}\right)\lambda +\bm{\kappa} & -\left(q_{1}p_{1}+q_{2}p_{2}\right)\lambda-q_{2}p_{1}
\end{pmatrix},
\]
\begin{align*}
\mathbb{U}_{1} & =\begin{pmatrix}0 & 1\\
\lambda^{2}-2q_{1}\lambda+3q_{1}^{2}-2q_{2} & 0
\end{pmatrix},\qquad\\
\mathbb{U}_{2} & =\begin{pmatrix}q_{1}p_{1}+q_{2}p_{2} & \lambda+q_{1}\\
\lambda^{3}-q_{1}\lambda^{2}+\left(q_{1}^{2}-2q_{2}\right)\lambda-(q_{1}p_{1}^{2}+2q_{2}p_{1}p_{2}-2q_{1}q_{2}-c_2\frac{1}{q_{2}}+c_{1}\frac{q_{1}}{q_{2}^{2}}) & -q_{1}p_{1}-q_{2}p_{2}
\end{pmatrix},
\end{align*}
where
\[
\bm{\kappa}= -q_{2}p_{1}^{2}+c_{1}\frac{1}{q_{2}}.
\]
The inverse of the map \eqref{m3} is given by
\begin{align*}
u=2q_{2}-3q_{1}^{2},\qquad u_{x}=12p_{1}q_{1}^{2}+12p_{2}q_{2}q_{1}-4p_{1}q_{2},\qquad v=2q_{1},\qquad v_{x}=-4p_{1}q_{1}-4p_{2}q_{2},
\end{align*}
together with
\begin{align*}
u_{2x} & =12c_{2}q_{1}-\frac{4c_{1}}{q_{2}}-24p_{1}^{2}q_{1}^{2}-48p_{1}p_{2}q_{2}q_{1}-24p_{2}^{2}q_{2}^{2}+4p_{1}^{2}q_{2}-48q_{1}^{4}+84q_{2}q_{1}^{2}-8q_{2}^{2},\\
v_{2x} & =-4c_{2}+16q_{1}^{3}-24q_{2}q_{1}.
\end{align*}

\subsubsection*{Miura maps}

Consider again the case $N=n=2$. For $m=1$ the Miura
map (\ref{mu1}) between parametrizations $(q_{1},q_{2},p_{1},p_{2},c_{1},c_{2})$
and $(\bar{q}_{1},\bar{q}_{2},\bar{p}_{1},\bar{p}_{2},\bar{c}_{1},\bar{c}_{2})$
of the stationary manifold $\mathcal{M}_{6}$ attains the form
\begin{align}
q_{1} & =\bar{q}_{1},\qquad q_{2}=\bar{q}_{2},\qquad p_{1}=-\bar{q}_{1}\bar{p}_{1}-\bar{q}_{2}\bar{p}_{2},\qquad p_{2}=\bar{p}_{1},\nonumber \\
c_{1} & =\bar{H}_{1}(\bar{\bm{q}},\bar{\bm{p}},\bar{\bm{c}})=-\bar{p}_{2}^{2}\bar{q}_{2}+\bar{p}_{1}^{2}-\bar{q}_{1}^{4}+3\bar{q}_{2}\bar{q}_{1}^{2}-\bar{q}_{2}^{2}+\bar{c}_{1}\frac{1}{\bar{q}_{2}}+\bar{c}_{2}\bar{q}_{1},\label{mp41}\\
c_{2} & =\bar{c}_{2},\nonumber
\end{align}
where $\bar{H}_{1}$ is given by (\ref{2H}) and for $m=2$ the respective
form
\begin{align}
q_{1} & =\bar{q}_{1},\qquad q_{2}=\bar{q}_{2},\qquad p_{1}=\left(\bar{q}_{1}^{2}-\bar{q}_{2}\right)\bar{p}_{1}+\bar{q}_{1}\bar{q}_{2}\bar{p}_{2},\qquad p_{2}=-\bar{q}_{1}\bar{p}_{1}-\bar{q}_{2}\bar{p}_{2},\nonumber \\
c_{1} & =\bar{H}_{2}(\bar{\bm{q}},\bar{\bm{p}},\bar{\bm{c}})=-\bar{q}_{2}\bar{p}_{1}^{2}+\bar{p}_{2}^{2}\bar{q}_{2}^{2}-\bar{q}_{2}^{2}+\bar{q}_{1}^{2}\bar{q}_{2}+\bar{c}_{1}\left(\frac{1}{\bar{q}_{2}}-\frac{\bar{q}_{1}^{2}}{\bar{q}_{2}^{2}}\right)+\bar{c}_{2}\frac{\bar{q}_{1}}{\bar{q}_{2}},\label{mp42}\\
c_{2} & =\bar{H}_{1}(\bar{\bm{q}},\bar{\bm{p}},\bar{\bm{c}})=-\bar{p}_{1}^{2}\bar{q}_{1}-2\bar{p}_{1}\bar{p}_{2}\bar{q}_{2}+\bar{q}_{1}^{3}-2\bar{q}_{2}\bar{q}_{1}-\bar{c}_{1}\frac{\bar{q}_{1}}{\bar{q}_{2}^{2}}+\bar{c}_{2}\frac{1}{\bar{q}_{2}}.\nonumber
\end{align}
where $\bar{H}_{1}$ and $\bar{H}_{2}$ are given by (\ref{3H}).
This yields the three-Hamiltonian representation of all the vector
fields of the Stäckel system (\ref{K0}) (and thus the three-Hamiltonian
representation of the stationary cKdV system (\ref{r2})) in the $(q_{1},q_{2},p_{1},p_{2},c_{1},c_{2})$
parametrization:
\begin{align*}
\pi_{0}dH_{1} & =\pi_{1}dc_{1}=\pi_{2}dc_{2},\\
\pi_{0}dH_{2} & =\pi_{1}dH_{1}=\pi_{2}dc_{1},
\end{align*}
where $\pi_{0}=\pi$, while $\pi_{1}$ is generated by the Miura map (\ref{mp41})
and $\pi_{2}$ is generated by the Miura map (\ref{mp42}). Explicitly,
the matrix representations of Poisson bi-vectors $\pi_{i}$ are as
follows
\[
\pi_{0}=\begin{pmatrix}0 & 0 & 1 & 0 & 0 & 0\\
0 & 0 & 0 & 1 & 0 & 0\\
-1 & 0 & 0 & 0 & 0 & 0\\
0 & -1 & 0 & 0 & 0 & 0\\
0 & 0 & 0 & 0 & 0 & 0\\
0 & 0 & 0 & 0 & 0 & 0
\end{pmatrix},\qquad\pi_{1}=\begin{pmatrix}0 & 0 & -q_{1} & 1 & \frac{\partial H_{1}}{\partial p_{1}} & 0\\
0 & 0 & -q_{2} & 0 & \frac{\partial H_{1}}{\partial p_{2}} & 0\\
q_{1} & q_{2} & 0 & -p_{2} & -\frac{\partial H_{1}}{\partial q_{1}} & 0\\
-1 & 0 & p_{2} & 0 & -\frac{\partial H_{1}}{\partial q_{2}} & 0\\
-\frac{\partial H_{1}}{\partial p_{1}} & -\frac{\partial H_{1}}{\partial p_{2}} & \frac{\partial H_{1}}{\partial q_{1}} & \frac{\partial H_{1}}{\partial q_{2}} & 0 & 0\\
0 & 0 & 0 & 0 & 0 & 0
\end{pmatrix},
\]
\[
\pi_{2}=\begin{pmatrix}0 & 0 & q_{1}^{2}-q_{2} & -q_{1} & \frac{\partial H_{2}}{\partial p_{1}} & \frac{\partial H_{1}}{\partial p_{1}}\\
0 & 0 & q_{1}q_{2} & -q_{2} & \frac{\partial H_{2}}{\partial p_{2}} & \frac{\partial H_{1}}{\partial p_{2}}\\
-q_{1}^{2}+q_{2} & -q_{1}q_{2} & 0 & q_{1}p_{2} & -\frac{\partial H_{2}}{\partial q_{1}} & -\frac{\partial H_{1}}{\partial q_{1}}\\
q_{1} & q_{2} & -q_{1}p_{2} & 0 & -\frac{\partial H_{2}}{\partial q_{2}} & -\frac{\partial H_{1}}{\partial q_{2}}\\
-\frac{\partial H_{2}}{\partial p_{1}} & -\frac{\partial H_{2}}{\partial p_{2}} & \frac{\partial H_{2}}{\partial q_{1}} & \frac{\partial H_{2}}{\partial q_{2}} & 0 & 0\\
-\frac{\partial H_{1}}{\partial p_{1}} & -\frac{\partial H_{1}}{\partial p_{2}} & \frac{\partial H_{1}}{\partial q_{1}} & \frac{\partial H_{1}}{\partial q_{2}} & 0 & 0
\end{pmatrix},
\]
where the Hamiltonians $H_{1}$, $H_{2}$ are given by (\ref{1H}). Now, the three Hamiltonian representations for the Stäckel systems associated with $m=1$ and $m=2$ can be obtained by means of
the Miura maps \eqref{mp41} and \eqref{mp42}.

Let us finally observe that this system has also been discussed in \cite{FH}, where the authors considered its three equivalent Ostrogradsky representations related with three different albeit equivalent Lagrangian formulations and where the above Poisson operators have also been presented.

\subsubsection{The stationary reduction with $n=3$ and $m=0$}\label{krotki}

Here we consider the $n=3$ stationary DWW system (see Section~\ref{P1}),
i.e.~the case $N=2$ and $n=3$.

\subsubsection*{Stationary system}

The third ($n=3$) DWW stationary system consists of first three evolution
equations from the hierarchy (given already in \eqref{pole123}):
\begin{equation}
\begin{pmatrix}u\\
v
\end{pmatrix}_{t_{1}}=\bm{K}_{1}\equiv\begin{pmatrix}u_{x}\\
v_{x}
\end{pmatrix},\qquad\begin{pmatrix}u\\
v
\end{pmatrix}_{t_{2}}=\bm{K}_{2}\equiv\begin{pmatrix}\frac{1}{2}vu_{x}+uv_{x}+\frac{1}{4}v_{3x}\\
u_{x}+\frac{3}{2}vv_{x}
\end{pmatrix},\label{ex1}
\end{equation}
\begin{equation}
\begin{pmatrix}u\\
v
\end{pmatrix}_{t_{3}}=\bm{K}_{3}\equiv\begin{pmatrix}\frac{3}{8}v^{2}u_{x}+\frac{3}{2}uvv_{x}+\frac{3}{2}uu_{x}+\frac{1}{4}u_{3x}+\frac{3}{8}vv_{3x}+\frac{9}{8}v_{x}v_{2x}\\
\frac{3}{2}vu_{x}+\frac{3}{2}uv_{x}+\frac{15}{8}v^{2}v_{x}+\frac{1}{4}v_{3x}
\end{pmatrix},\label{ex2}
\end{equation}
and the fourth stationary flow, $\bm{K}_{4}=0$, (cf. and (\ref{pole4}))
which in the normal form is given by
\begin{equation}\label{wiezy}
\begin{split}
u_{3x} & =-\frac{15}{2}v^{2}u_{x}-15uvv_{x}-6uu_{x}-\frac{1}{4}35v^{3}v_{x}-\frac{5}{2}vv_{3x}-5v_{x}v_{2x},\\
v_{5x} & =-24u^{2}v_{x}+40v^{3}u_{x}+60uv^{2}v_{x}-18u_{2x}v_{x}-20u_{x}v_{2x}-10uv_{3x}+\frac{105}{2}v^{4}v_{x}+\frac{15}{2}v^{2}v_{3x}\\
&\quad-15vv_{x}v_{2x}-15v_{x}^{3}.
\end{split}
\end{equation}
The constraints \eqref{wiezy} define the $2n+N=8$-th
dimensional stationary manifold $\mathcal{M}_{3}$ parametrized by
the jet coordinates $[\bm{u}]=(u,u_{x},u_{2x}$,$v,v_{x}\ldots,v_{4x})$.
The two vector fields \eqref{ex1} preserve their form in the above
parametrization of $\mathcal{M}_{3}$, while the vector field \eqref{ex2}
attains on $\mathcal{M}_{3}$, parametrized as above, the form
\begin{equation}
\begin{pmatrix}u\\
v
\end{pmatrix}_{t_{3}}=\begin{pmatrix}-\frac{3}{2}v^{2}u_{x}-\frac{9}{4}uvv_{x}-\frac{35}{16}v^{3}v_{x}-\frac{1}{4}vv_{3x}-\frac{1}{8}v_{x}v_{2x}\\
\frac{3vu_{x}}{2}vu_{x}+\frac{3}{2}uv_{x}+\frac{15}{8}v^{2}v_{x}+\frac{1}{4}v_{3x}
\end{pmatrix}.\label{ex3}
\end{equation}
The corresponding spectral curve \eqref{s9} is
\begin{equation}
\lambda^{8}+h_{4}\lambda^{4}+h_{3}\lambda^{3}+h_{2}\lambda^{2}+h_{1}\lambda+h_{0}=\mu^{2}.\label{ex4}
\end{equation}
The corresponding functions $h_{r}$ in (\ref{s14}) attain the explicit
form
\begin{align*}
h_{4} & =-\frac{35}{64}v^{4}-\frac{15}{8}uv^{2}-\frac{5}{8}v_{2x}v-\frac{3}{4}u^{2}-\frac{5}{16}v_{x}^{2}-\frac{1}{4}u_{2x},\\
h_{3} & =-\frac{7}{32}v^{5}-\frac{5}{4}uv^{3}-\frac{25}{32}v_{2x}v^{2}-\frac{3}{2}u^{2}v-\frac{15}{16}v_{x}^{2}v-\frac{1}{2}u_{2x}v-\frac{5}{8}u_{x}v_{x}-\frac{5}{8}uv_{2x}-\frac{1}{16}v_{4x},\\
h_{2} & =-\frac{35}{256}v^{6}-\frac{55}{64}uv^{4}-\frac{15}{32}v_{2x}v^{3}-\frac{21}{16}u^{2}v^{2}-\frac{15}{64}v_{x}^{2}v^{2}-\frac{9}{32}u_{2x}v^{2}-\frac{5}{8}uv_{2x}v-\frac{1}{32}v_{4x}v-\frac{1}{4}u^{3}+\frac{1}{16}u_{x}^{2}\\
 & \quad-\frac{1}{64}v_{2x}^{2}-\frac{1}{8}uu_{2x}+\frac{1}{32}v_{x}v_{3x},\\
h_{1} & =-\frac{25}{256}v^{7}-\frac{45}{64}uv^{5}-\frac{95}{256}v_{2x}v^{4}-\frac{23}{16}u^{2}v^{3}-\frac{15}{128}v_{x}^{2}v^{3}-\frac{7}{32}u_{2x}v^{3}+\frac{15}{64}u_{x}v_{x}v^{2}-\frac{15}{16}uv_{2x}v^{2}\\
 & \quad-\frac{3}{128}v_{4x}v^{2}-\frac{3}{4}u^{3}v+\frac{3}{16}u_{x}^{2}v-\frac{15}{32}uv_{x}^{2}v-\frac{1}{16}v_{2x}^{2}v-\frac{3}{8}uu_{2x}v+\frac{3}{64}v_{x}v_{3x}v-\frac{3}{16}uu_{x}v_{x}-\frac{5}{16}u^{2}v_{2x}\\
 & \quad-\frac{3}{64}v_{x}^{2}v_{2x}-\frac{1}{32}u_{2x}v_{2x}+\frac{1}{32}u_{x}v_{3x}-\frac{1}{32}uv_{4x},\\
h_{0} & =-\frac{25}{256}uv^{6}-\frac{75}{512}v_{2x}v^{5}-\frac{15}{32}u^{2}v^{4}-\frac{75}{1024}v_{x}^{2}v^{4}-\frac{15}{128}u_{2x}v^{4}+\frac{15}{128}u_{x}v_{x}v^{3}-\frac{35}{64}uv_{2x}v^{3}-\frac{5}{256}v_{4x}v^{3}\\
 & \quad-\frac{9}{16}u^{3}v^{2}+\frac{9}{64}u_{x}^{2}v^{2}-\frac{45}{128}uv_{x}^{2}v^{2}-\frac{15}{256}v_{2x}^{2}v^{2}-\frac{9}{32}uu_{2x}v^{2}+\frac{15}{256}v_{x}v_{3x}v^{2}-\frac{9}{32}uu_{x}v_{x}v\\
 & \quad-\frac{15}{32}u^{2}v_{2x}v-\frac{15}{128}v_{x}^{2}v_{2x}v-\frac{3}{64}u_{2x}v_{2x}v+\frac{3}{64}u_{x}v_{3x}v-\frac{3}{64}uv_{4x}v+\frac{9}{64}u^{2}v_{x}^{2}-\frac{1}{16}uv_{2x}^{2}
 +\frac{1}{256}v_{3x}^{2}\\
 & \quad-\frac{3}{32}u_{x}v_{x}v_{2x}+\frac{3}{64}uv_{x}v_{3x}-\frac{1}{128}v_{2x}v_{4x}.
\end{align*}

\subsubsection*{Foliation for $m=0$}

We will now consider the Hamiltonian foliation of $\mathcal{M}_{3}$
defined by (\ref{s19}) for $m=0$, that is $h_{4}=c_{2}$ and $h_{3}=c_{1}$,
which yields the $2n=6$-dimensional leaves $\mathcal{M}_{3,0}^{\bm{c}}$
defined by the conditions
\begin{equation}
\begin{split}u_{2x} & =-4c_{2}-3u^{2}-\frac{15}{2}uv^{2}-\frac{35}{16}v^{4}-\frac{5}{2}vv_{2x}-\frac{5}{4}v_{x}^{2},\\
v_{4x} & =32c_{2}v-16c_{1}+40uv^{3}-10u_{x}v_{x}-10uv_{2x}+14v^{5}+\frac{15}{2}v^{2}v_{2x}-5vv_{x}^{2}.
\end{split}
\label{ex5}
\end{equation}
Each leaf $\mathcal{M}_{3,0}^{\bm{c}}$ is now parametrized by the
jet coordinates $[\bm{u}]=(u,u_{x},v,v_{x},v_{2x},v_{3x})$. All three
vector fields $\bm{K}_{i}$ preserve their explicit form when reducing
from $\mathcal{M}_{3}$ to the leaf $\mathcal{M}_{3,0}^{\bm{c}}$,
due to the fact that they do not contain neither $u_{2x}$ nor $v_{4x}$
that define this leaf within the manifold $\mathcal{M}_{3}$. Taking
into account conditions \eqref{ex5} the spectral curve \eqref{ex4}
attains the form (cf.~\eqref{SCP})
\begin{equation}
\lambda^{8}+c_{2}\lambda^{4}+c_{1}\lambda^{3}+H_{1}\lambda^{2}+H_{2}\lambda+H_{3}=\mu^{2},\label{ex6}
\end{equation}
while $H_{i}$ are explicitly given by
\begin{align*}
H_{1} & =\frac{1}{16}u_{x}^{2}+\frac{5}{16}vv_{x}u_{x}+\frac{35}{128}v^{2}v_{x}^{2}+\frac{5}{32}uv_{x}^{2}+\frac{1}{32}v_{x}v_{3x}+\frac{1}{8}u^{3}+\frac{15}{32}u^{2}v^{2}+\frac{35}{128}uv^{4}+\frac{21}{512}v^{6}-\frac{1}{64}v_{2x}^{2}\\
 & \quad+\frac{1}{2}c_{2}u+\frac{1}{8}c_{2}v^{2}+\frac{1}{2}c_{1}v,\\
H_{2} & =\frac{35}{128}v_{x}^{2}v^{3}+\frac{15}{32}u_{x}v_{x}v^{2}+\frac{3}{16}u_{x}^{2}v+\frac{5}{32}uv_{x}^{2}v+\frac{3}{64}v_{x}v_{3x}v+\frac{1}{8}uu_{x}v_{x}-\frac{1}{128}v_{x}^{2}v_{2x}+\frac{1}{32}u_{x}v_{3x}\\
 & \quad+\frac{1}{2}c_{2}uv+\frac{1}{8}c_{2}v^{3}+\frac{1}{8}c_{2}v_{2x}+\frac{1}{2}c_{1}u+\frac{3}{8}c_{1}v^{2},
\end{align*}
and
\begin{align*}
H_{3} & =\frac{175}{1024}v_{x}^{2}v^{4}+\frac{5}{16}u_{x}v_{x}v^{3}+\frac{9}{64}u_{x}^{2}v^{2}+\frac{15}{64}uv_{x}^{2}v^{2}+\frac{15}{256}v_{x}v_{3x}v^{2}+\frac{3}{16}uu_{x}v_{x}v-\frac{5}{256}v_{x}^{2}v_{2x}v\\
 & \quad+\frac{3}{64}u_{x}v_{3x}v+\frac{9}{64}u^{2}v_{x}^{2}+\frac{1}{256}v_{3x}^{2}-\frac{1}{64}u_{x}v_{x}v_{2x}+\frac{3}{64}uv_{x}v_{3x}+\frac{9}{32}u^{3}v^{2}+\frac{15}{128}u^{2}v^{4}\\
 & \quad +\frac{9}{64}u^{2}vv_{2x}
-\frac{21}{512}uv^{6}+\frac{5}{128}uv^{3}v_{2x}+\frac{1}{64}uv_{2x}^{2}-\frac{35}{2048}v^{8}-\frac{7}{1024}v^{5}v_{2x}-\frac{3}{8}c_{2}uv^{2}\\
 & \quad-\frac{1}{32}5c_{2}v^{4}-\frac{1}{16}c_{2}vv_{2x}+\frac{3}{4}c_{1}uv+\frac{5}{16}c_{1}v^{3}+\frac{1}{8}c_{1}v_{2x}.
\end{align*}

\subsubsection*{Stäckel system}

The relation \eqref{ex6}, if treated as the separation curve~\eqref{sc},
yields the Stäckel Hamiltonians (given here directly in Viète coordinates
(\ref{V1}))
\begin{align*}
H_{1} & =2p_{3}p_{2}q_{1}+p_{3}^{2}q_{2}+p_{2}^{2}+2p_{1}p_{3}-q_{1}^{6}+5q_{2}q_{1}^{4}-4q_{3}q_{1}^{3}-6q_{2}^{2}q_{1}^{2}+6q_{2}q_{3}q_{1}+q_{2}^{3}-q_{3}^{2}+c_{2}q_{2}-c_{2}q_{1}^{2}+c_{1}q_{1},\\
H_{2} & =2p_{2}^{2}q_{1}+2p_{3}p_{2}q_{1}^{2}+2p_{1}p_{3}q_{1}+p_{3}^{2}q_{1}q_{2}-p_{3}^{2}q_{3}+2p_{1}p_{2}+q_{3}q_{1}^{4}+4q_{2}^{2}q_{1}^{3}-6q_{2}q_{3}q_{1}^{2}-3q_{2}^{3}q_{1}+2q_{3}^{2}q_{1},\\
 & \quad-q_{1}^{5}q_{2}+3q_{2}^{2}q_{3}+c_{2}q_{3}-c_{2}q_{1}q_{2}+c_{1}q_{2}\\
H_{3} & =2p_{2}p_{1}q_{1}+2p_{3}p_{1}q_{2}+p_{2}^{2}q_{1}^{2}+p_{3}^{2}q_{2}^{2}+2p_{2}p_{3}q_{1}q_{2}-2p_{2}p_{3}q_{3}-p_{3}^{2}q_{1}q_{3}+p_{1}^{2}+4q_{2}q_{3}q_{1}^{3}-3q_{3}^{2}q_{1}^{2}-3q_{2}^{2}q_{3}q_{1}\\
 & \quad+2q_{2}q_{3}^{2}-q_{1}^{5}q_{3}-c_{2}q_{1}q_{3}+c_{1}q_{3},
\end{align*}
which generate the following Stäckel system (cf.~\eqref{Hamk}) on
$\mathcal{M}_{3,0}^{\bm{c}}$:
\[
\begin{pmatrix}q_{1}\\
q_{2}\\
q_{3}\\
p_{1}\\
p_{2}\\
p_{3}
\end{pmatrix}_{t_{1}}=\begin{pmatrix}2p_{3}\\
2p_{2}+2p_{3}q_{1}\\
2p_{1}+2p_{2}q_{1}+2p_{3}q_{2}\\
6q_{1}^{5}-20q_{2}q_{1}^{3}+12q_{3}q_{1}^{2}+12q_{2}^{2}q_{1}+2c_{2}q_{1}-c_{1}-2p_{2}p_{3}-6q_{2}q_{3}\\
-5q_{1}^{4}+12q_{2}q_{1}^{2}-6q_{3}q_{1}-p_{3}^{2}-3q_{2}^{2}-c_{2}\\
4q_{1}^{3}-6q_{2}q_{1}+2q_{3}
\end{pmatrix},
\]
\[
\begin{pmatrix}q_{1}\\
q_{2}\\
q_{3}\\
p_{1}\\
p_{2}\\
p_{3}
\end{pmatrix}_{t_{2}}=\begin{pmatrix}2p_{2}+2p_{3}q_{1}\\
2p_{3}q_{1}^{2}+4p_{2}q_{1}+2p_{1}\\
2p_{2}q_{1}^{2}+2p_{1}q_{1}+2p_{3}q_{2}q_{1}-2p_{3}q_{3}\\
5q_{2}q_{1}^{4}-4q_{3}q_{1}^{3}-12q_{2}^{2}q_{1}^{2}-4p_{2}p_{3}q_{1}+12q_{2}q_{3}q_{1}+3q_{2}^{3}-2p_{2}^{2}-2q_{3}^{2}-2p_{1}p_{3}-p_{3}^{2}q_{2}+c_{2}q_{2}\\
q_{1}^{5}-8q_{2}q_{1}^{3}+6q_{3}q_{1}^{2}-p_{3}^{2}q_{1}+9q_{2}^{2}q_{1}+c_{2}q_{1}-c_{1}-6q_{2}q_{3}\\
-q_{1}^{4}+6q_{2}q_{1}^{2}-4q_{3}q_{1}+p_{3}^{2}-3q_{2}^{2}-c_{2}
\end{pmatrix},
\]
\[
\begin{pmatrix}q_{1}\\
q_{2}\\
q_{3}\\
p_{1}\\
p_{2}\\
p_{3}
\end{pmatrix}_{t_{3}}=\begin{pmatrix}2p_{1}+2p_{2}q_{1}+2p_{3}q_{2}\\
2p_{2}q_{1}^{2}+2p_{1}q_{1}+2p_{3}q_{2}q_{1}-2p_{3}q_{3}\\
2p_{3}q_{2}^{2}+2p_{1}q_{2}+2p_{2}q_{1}q_{2}-2p_{2}q_{3}-2p_{3}q_{1}q_{3}\\
5q_{3}q_{1}^{4}-12q_{2}q_{3}q_{1}^{2}-2p_{2}^{2}q_{1}+6q_{3}^{2}q_{1}-2p_{1}p_{2}-2p_{2}p_{3}q_{2}+p_{3}^{2}q_{3}+3q_{2}^{2}q_{3}+c_{2}q_{3}\\
-4q_{3}q_{1}^{3}-2p_{2}p_{3}q_{1}+6q_{2}q_{3}q_{1}-2q_{3}^{2}-2p_{1}p_{3}-2p_{3}^{2}q_{2}\\
q_{1}^{5}-4q_{2}q_{1}^{3}+6q_{3}q_{1}^{2}+p_{3}^{2}q_{1}+3q_{2}^{2}q_{1}+c_{2}q_{1}-c_{1}+2p_{2}p_{3}-4q_{2}q_{3}
\end{pmatrix}.
\]
Their Lax representation \eqref{T1} is given by the Lax matrix
(\ref{laxm}), given in Viète coordinates by
\[
\mathbb{L}=\begin{pmatrix}-p_{3}\lambda^{2}-\left(p_{2}+p_{3}q_{1}\right)\lambda-p_{1}-p_{2}q_{1}-p_{3}q_{2} & \lambda^{3}+q_{1}\lambda^{2}+q_{2}\lambda+q_{3}\\
\mathbb{L}_{21} & p_{3}\lambda^{2}+\left(p_{2}+p_{3}q_{1}\right)\lambda+p_{1}+p_{2}q_{1}+p_{3}q_{2}
\end{pmatrix},
\]
where
\begin{align*}
\mathbb{L}_{21} & =\lambda^{5}-q_{1}\lambda^{4}+\left(q_{1}^{2}-q_{2}\right)\lambda^{3}-\left(q_{1}^{3}-2q_{2}q_{1}+q_{3}\right)\lambda^{2}+\left(q_{1}^{4}-3q_{2}q_{1}^{2}+2q_{3}q_{1}-p_{3}^{2}+q_{2}^{2}+c_{2}\right)\lambda\\
 & \quad-q_{1}^{5}-3q_{1}q_{2}^{2}-2p_{2}p_{3}-p_{3}^{2}q_{1}+4q_{1}^{3}q_{2}-3q_{1}^{2}q_{3}+2q_{2}q_{3}+c_{1}-c_{2}q_{1},
\end{align*}
and by the auxiliary matrices (\ref{T8}), given in Viète
coordinates by
\[
\mathbb{U}_{1}=\begin{pmatrix}0 & 1\\
\lambda^{2}-2q_{1}\lambda+3q_{1}^{2}-2q_{2} & 0
\end{pmatrix},\qquad\mathbb{U}_{2}=\begin{pmatrix}-p_{3} & \lambda+q_{1}\\
\lambda^{3}-q_{1}\lambda^{2}+\left(q_{1}^{2}-2q_{2}\right)\lambda-q_{1}^{3}+4q_{1}q_{2}-2q_{3} & p_{3}
\end{pmatrix},
\]
\[
\mathbb{U}_{3}=\begin{pmatrix}-\lambda p_{3}-p_{2}-p_{3}q_{1} & \lambda^{2}+q_{1}\lambda+q_{2}\\
\lambda^{4}-q_{1}\lambda^{3}+\left(q_{1}^{2}-q_{2}\right)\lambda^{2}-\left(q_{1}^{3}-2q_{1}q_{2}+2q_{3}\right)\lambda+\bm{\kappa}& \lambda p_{3}+p_{2}+p_{3}q_{1}
\end{pmatrix},
\]
where
\[
\bm{\kappa} = +q_{1}^{4}-p_{3}^{2}+q_{2}^{2}+c_{2}-3q_{1}^{2}q_{2}+4q_{1}q_{3}.
\]

\subsubsection*{Transformation between jet and canonical coordinates}

Finally, let us present the map (\ref{mapa}) between the jet variables
$[\bm{u}]=(u,u_{x},u_{2x}$,$v,v_{x}\ldots,v_{4x})$ and the (extended
by Casimirs $c_{i}$) Viète coordinates on the stationary manifold
$\mathcal{M}_{3}$. It is explicitly given by
\begin{subequations}\label{ex134}
\begin{equation}
\begin{split}q_{1} & =\frac{1}{2}v,\qquad q_{2}=\frac{1}{2}u+\frac{3}{8}v^{2},\qquad q_{3}=\frac{3}{4}uv+\frac{5}{16}v^{3}+\frac{1}{8}v_{2x},\\
p_{1} & =\frac{1}{4}vu_{x}+\frac{1}{4}uv_{x}+\frac{1}{4}v^{2}v_{x}+\frac{1}{16}v_{3x},\qquad p_{2}=\frac{1}{4}u_{x}+\frac{1}{4}vv_{x},\qquad p_{3}=\frac{1}{4}v_{x}
\end{split}
\label{ex13}
\end{equation}
and by
\begin{equation}
\begin{split}c_{1} & =h_{3}\equiv-\frac{7}{32}v^{5}-\frac{5}{4}uv^{3}-\frac{25}{32}v_{2x}v^{2}-\frac{3}{2}u^{2}v-\frac{15}{16}v_{x}^{2}v-\frac{1}{2}u_{2x}v-\frac{5}{8}u_{x}v_{x}-\frac{5}{8}uv_{2x}-\frac{1}{16}v_{4x},\\
c_{2} & =h_{4}\equiv-\frac{35}{64}v^{4}-\frac{15}{8}uv^{2}-\frac{5}{8}v_{2x}v-\frac{3}{4}u^{2}-\frac{5}{16}v_{x}^{2}-\frac{1}{4}u_{2x}.
\end{split}
\label{ex14}
\end{equation}
\end{subequations}

Applying the map \eqref{ex13} to the above Stäckel system we reconstruct
the respective vector fields \eqref{ex1} and \eqref{ex3} together with
the constraints \eqref{ex5} or equivalently \eqref{ex14}.

The inverse of the map \eqref{ex134} is given by
\begin{align*}
u & =2q_{2}-3q_{1}^{2},\qquad u_{x}=4p_{2}-8p_{3}q_{1},\\
v & =2q_{1},\qquad v_{x}=4p_{3},\qquad v_{2x}=16q_{1}^{3}-24q_{2}q_{1}+8q_{3},\qquad v_{3x}=48p_{3}q_{1}^{2}-32p_{2}q_{1}-32p_{3}q_{2}+16p_{1},
\end{align*}
together with
\begin{align*}
u_{2x} & =-4c_{2}-20p_{3}^{2}-52q_{1}^{4}+96q_{2}q_{1}^{2}-40q_{3}q_{1}-12q_{2}^{2},\\
v_{4x} & =160p_{3}^{2}q_{1}-160p_{2}p_{3}+448q_{1}^{5}-1120q_{2}q_{1}^{3}+480q_{3}q_{1}^{2}+480q_{2}^{2}q_{1}-160q_{2}q_{3}-16c_{1}+64c_{2}q_{1}.
\end{align*}

\subsection{The case of $N=4$, $n=2$ and $m=0$} \label{ostatni}

Let us now take a closer look at four-component ($N=4$) stationary
cKdV system with $n=2$. We will again only consider the case $m=0$.

\subsubsection*{Four component cKdV hierarchy}

Assume thus that $N=4$ and denote $\bm{u}=(u_{0},u_{1},u_{2},u_{3})^{T}\equiv(u,v,w,r)^T$.
Then, the coefficients $P_{i}$ of the series \eqref{c4} are
\begin{align*}
P_{0} & =2,\qquad P_{1}=r,\qquad P_{2}=\frac{3}{4}r^{2}+w,\qquad P_{3}=\frac{5}{8}r^{3}+\frac{3}{2}rw+v,\\
P_{4} & =\frac{35}{64}r^{4}+\frac{15}{8}r^{2}w+\frac{3}{2}rv+u+\frac{3}{4}w^{2},\\
P_{5} & =\frac{63}{128}r^{5}+\frac{35}{16}r^{3}w+\frac{15}{8}r^{2}v+\frac{3}{2}ru+\frac{15}{8}rw^{2}+\frac{1}{4}r_{2x}+\frac{3}{2}vw\\
P_{6} & =\frac{231}{512}r^{6}+\frac{315}{128}r^{4}w+\frac{35}{16}r^{3}v+\frac{15}{8}r^{2}u+\frac{105}{32}r^{2}w^{2}+\frac{15}{4}rvw+\frac{5}{8}rr_{2x}+\frac{5}{16}r_{x}^{2}+\frac{3}{2}uw\\
&\quad+\frac{3}{4}v^{2}+\frac{5}{8}w^{3}+\frac{1}{4}w_{2x},\\
 & \vdots
\end{align*}
The first three members of the four component cKdV hierarchy (\ref{c11a})
have the form
\begin{equation}
\begin{pmatrix}u\\
v\\
w\\
r
\end{pmatrix}_{t_{1}}=\bm{K}_{1}\equiv\begin{pmatrix}u_{x}\\
v_{x}\\
w_{x}\\
r_{x}
\end{pmatrix},\qquad\begin{pmatrix}u\\
v\\
w\\
r
\end{pmatrix}_{t_{2}}=\bm{K}_{2}\equiv\begin{pmatrix}ur_{x}+\frac{1}{2}ru_{x}+\frac{1}{4}r_{3x}\\
vr_{x}+\frac{1}{2}rv_{x}+u_{x}\\
wr_{x}+\frac{1}{2}rw_{x}+v_{x}\\
\frac{3}{2}rr_{x}+w_{x}
\end{pmatrix},\label{ez1}
\end{equation}
\[
\begin{pmatrix}u\\
v\\
w\\
r
\end{pmatrix}_{t_{3}}=\bm{K}_{3}\equiv\begin{pmatrix}\frac{3}{8}r^{2}u_{x}+\frac{3}{2}rur_{x}+\frac{3}{8}rr_{3x}+\frac{9}{8}r_{x}r_{2x}+\frac{1}{2}wu_{x}+uw_{x}+\frac{1}{4}w_{3x}\\
\frac{3}{8}r^{2}v_{x}+\frac{1}{2}ru_{x}+ur_{x}+\frac{3}{2}rvr_{x}+\frac{1}{4}r_{3x}+\frac{1}{2}wv_{x}+vw_{x}\\
\frac{3}{8}r^{2}w_{x}+\frac{1}{2}rv_{x}+vr_{x}+\frac{3}{2}rwr_{x}+u_{x}+\frac{3}{2}ww_{x}\\
\frac{15}{8}r^{2}r_{x}+\frac{3}{2}rw_{x}+\frac{3}{2}wr_{x}+v_{x}.
\end{pmatrix}
\]
Their zero-curvature curvature representation \eqref{l2} is generated
by the Lax matrix
\[
\mathbb{V}_{1}=\begin{pmatrix}0 & 1\\
\lambda^{4}-r\lambda^{3}-w\lambda^{2}-v\lambda-u & 0
\end{pmatrix}
\]
and the auxiliary matrices
\begin{align*}
\mathbb{V}_{2} & =\begin{pmatrix}-\frac{1}{4}r_{x} & \lambda+\frac{1}{2}r\\
\lambda^{5}-\frac{1}{2}r\lambda^{4}-\left(\frac{1}{2}r^{2}+w\right)\lambda^{3}-\left(v+\frac{1}{2}wr\right)\lambda^{2}-\left(u+\frac{1}{2}vr\right)\lambda-\frac{1}{2}ur-\frac{1}{4}r_{2x} & \frac{1}{4}r_{x}
\end{pmatrix},\\
\mathbb{V}_{3} & =\begin{pmatrix}-\frac{1}{4}r_{x}\lambda-\frac{1}{4}w_{x}-\frac{3}{8}rr_{x} & \lambda^{2}+\frac{1}{2}r\lambda+\frac{3}{8}r^{2}+\frac{1}{2}w\\
(\mathbb{V}_{3})_{21} & \frac{1}{4}r_{x}\lambda+\frac{1}{4}w_{x}+\frac{3}{8}rr_{x}
\end{pmatrix},
\end{align*}
where
\begin{align*}
(\mathbb{V}_{3})_{21} & =\lambda^{6}-\frac{1}{2}r\lambda^{5}-\left(\frac{1}{8}r^{2}+\frac{1}{2}w\right)\lambda^{4}-\left(\frac{3}{8}r^{3}+wr+v\right)\lambda^{3}-\left(\frac{1}{2}w^{2}+\frac{3}{8}r^{2}w+u+\frac{1}{2}vr\right)\lambda^{2}\\
 & \quad-\left(\frac{3}{8}vr^{2}+\frac{1}{2}ur+\frac{1}{2}vw+\frac{1}{4}r_{2x}\right)\lambda-\frac{3}{8}ur^{2}-\frac{3}{8}r_{x}^{2}-\frac{1}{2}uw-\frac{1}{4}w_{2x}-\frac{3}{8}rr_{2x}.
\end{align*}

\subsubsection*{Stationary system}

In the case of $N=4$ the second cKdV stationary system consists of
the two evolution equations in (\ref{ez1}) and the third stationary
flow, $\bm{K}_{3}=0$, which in the normal form is given by
\begin{align*}
u_{x} & =\frac{15}{16}r^{3}r_{x}+\frac{3}{8}r^{2}w_{x}-vr_{x}-\frac{3}{4}rwr_{x}-\frac{3}{2}ww_{x},\qquad v_{x}=-\frac{15}{8}r^{2}r_{x}-\frac{3}{2}rw_{x}-\frac{3}{2}wr_{x},\\
w_{3x} & =-\frac{45}{16}r^{5}r_{x}-\frac{45}{16}r^{4}w_{x}-12r^{3}wr_{x}+\frac{15}{2}r^{2}vr_{x}-\frac{15}{2}r^{2}ww_{x}+6rvw_{x}+2vwr_{x}-3rw^{2}r_{x}-\frac{9}{2}r_{x}r_{2x}\\
 & \quad-4uw_{x}+3w^{2}w_{x},\\
r_{3x} & =\frac{15}{16}r^{4}r_{x}+\frac{3}{2}r^{3}w_{x}+\frac{15}{2}r^{2}wr_{x}-4ur_{x}-4rvr_{x}+3w^{2}r_{x}+6rww_{x}-4vw_{x}.
\end{align*}
These constraints define the $2n+N=8$-th dimensional stationary manifold
$\mathcal{M}_{2}$, parametrized by the jet coordinates $[\bm{u}]=(u,v,w,w_{x},w_{2x},r,r_{x},r_{2x})$.
The vector fields \eqref{ez1} attain on $\mathcal{M}_{2}$, parametrized
as above, the form:
\begin{align*}
\begin{pmatrix}u\\
v\\
w\\
r
\end{pmatrix}_{t_{1}} & =\begin{pmatrix}\frac{15}{16}r^{3}r_{x}+\frac{3}{8}r^{2}w_{x}-vr_{x}-\frac{3}{4}rwr_{x}-\frac{3}{2}ww_{x}\\
-\frac{15}{8}r^{2}r_{x}-\frac{3}{2}rw_{x}-\frac{3}{2}wr_{x}\\
w_{x}\\
r_{x}
\end{pmatrix},\\
\begin{pmatrix}u\\
v\\
w\\
r
\end{pmatrix}_{t_{2}} & =\begin{pmatrix}\frac{45}{64}r^{4}r_{x}+\frac{9}{16}r^{3}w_{x}+\frac{3}{2}r^{2}wr_{x}-\frac{3}{2}rvr_{x}+\frac{3}{4}w^{2}r_{x}+\frac{3}{4}rww_{x}-vw_{x}\\
-\frac{3}{8}r^{2}w_{x}-\frac{3}{2}rwr_{x}-\frac{3}{2}ww_{x}\\
-\frac{15}{8}r^{2}r_{x}-rw_{x}-\frac{1}{2}wr_{x}\\
\frac{3}{2}rr_{x}+w_{x}
\end{pmatrix}.
\end{align*}
The spectral curve \eqref{s9} corresponding to $\mathcal{M}_{3}$
is
\begin{equation}
\lambda^{8}+h_{5}\lambda^{5}+h_{4}\lambda^{4}+h_{3}\lambda^{3}+h_{2}\lambda^{2}+h_{1}\lambda+h_{0}=\mu^{2},\label{ez7}
\end{equation}
while the functions $h_{k}$ in (\ref{s14}) are
\[
\begin{split}h_{5} & =-\frac{5}{8}r^{3}-\frac{3}{2}wr-v,\qquad h_{4}=-\frac{15}{64}r^{4}-\frac{9}{8}wr^{2}-vr-\frac{3}{4}w^{2}-u,\\
h_{3} & =-\frac{9}{64}r^{5}-\frac{3}{4}wr^{3}-vr^{2}-\frac{3}{4}w^{2}r-ur-vw-\frac{1}{4}r_{2x},\\
h_{2} & =-\frac{9}{64}wr^{4}-\frac{3}{8}vr^{3}-\frac{3}{8}w^{2}r^{2}-ur^{2}-\frac{1}{2}vwr-\frac{1}{2}r_{2x}r-\frac{1}{4}w^{3}-\frac{5}{16}r_{x}^{2}-uw-\frac{1}{4}w_{2x},\\
h_{1} & =-\frac{9}{64}vr^{4}-\frac{3}{8}ur^{3}-\frac{3}{8}vwr^{2}-\frac{9}{32}r_{2x}r^{2}-\frac{1}{2}uwr-\frac{1}{8}w_{2x}r-\frac{1}{4}vw^{2}+\frac{1}{8}w_{x}r_{x}-\frac{1}{8}wr_{2x},\\
h_{0} & =-\frac{9}{64}ur^{4}-\frac{9}{64}r_{2x}r^{3}-\frac{3}{8}uwr^{2}-\frac{3}{32}w_{2x}r^{2}+\frac{3}{16}w_{x}r_{x}r-\frac{3}{16}wr_{2x}r-\frac{1}{4}uw^{2}+\frac{1}{16}w_{x}^{2}\\
&\quad-\frac{3}{16}wr_{x}^{2}-\frac{1}{8}ww_{2x}.
\end{split}
\]

\subsubsection*{Foliation}

Let us consider the Hamiltonian foliation of $\mathcal{M}_{2}$ defined
by (\ref{s19}) for $m=0$. That is, we consider the foliation of
$\mathcal{M}_{2}$ given by the conditions
\[
h_{5}=c_{4},\quad h_{4}=c_{3},\quad h_{3}=c_{2},\quad h_{2}=c_{1},
\]
which in the normal form are given by
\begin{equation}
\begin{split}u & =c_{4}r-c_{3}+\frac{25}{64}r^{4}+\frac{3}{8}r^{2}w-\frac{3}{4}w^{2},\qquad v=-c_{4}-\frac{5}{8}r^{3}-\frac{3}{2}rw,\\
w_{2x} & =-\frac{5}{2}c_{4}r^{3}-4c_{3}r^{2}-10c_{4}rw+8c_{2}r+4c_{3}w-4c_{1}-\frac{11}{8}r^{6}-\frac{65}{8}r^{4}w-9r^{2}w^{2}-\frac{5}{4}r_{x}^{2}+2w^{3},\\
r_{2x} & =4c_{3}r+4c_{4}w-4c_{2}+\frac{3}{8}r^{5}+4r^{3}w+6rw^{2}.
\end{split}
\label{ez8}
\end{equation}
The above conditions define leaves $\mathcal{M}_{2,0}^{\bm{c}}$ parametrized
by the jet coordinates $[\bm{u}]=(w,w_{x},r,r_{x})$. Notice that
the field variables $u$ and $v$ are entirely eliminated using \eqref{ez8}.
In consequence, the vector fields \eqref{ez1} on the leave $\mathcal{M}_{2,0}^{\bm{c}}$
attains the \emph{two-component} form
\begin{equation}
\begin{pmatrix}w\\
r
\end{pmatrix}_{t_{1}}=\begin{pmatrix}w_{x}\\
r_{x}
\end{pmatrix},\qquad\begin{pmatrix}w\\
r
\end{pmatrix}_{t_{2}}=\begin{pmatrix}-\frac{15}{8}r^{2}r_{x}-rw_{x}-\frac{1}{2}wr_{x}\\
\frac{3}{2}rr_{x}+w_{x}.
\end{pmatrix}\label{ez9}
\end{equation}
as the remaining components in these vector fields become identities.

Taking into account \eqref{ez8} the spectral curve \eqref{ez7} takes
the form, cf.~\eqref{SCP},
\begin{equation}
\lambda^{8}+c_{4}\lambda^{5}+c_{3}\lambda^{4}+c_{2}\lambda^{3}+c_{1}\lambda^{2}+H_{1}\lambda+H_{2}=\mu^{2},\label{ez10}
\end{equation}
where
\begin{align*}
H_{1} & =\frac{1}{8}r_{x}w_{x}+\frac{5}{32}rr_{x}^{2}+\frac{1}{128}r^{7}-\frac{3}{64}r^{5}w-\frac{1}{4}r^{3}w^{2}-\frac{1}{4}rw^{3}+\frac{5}{64}c_{4}r^{4}-\frac{1}{4}c_{4}w^{2}-\frac{1}{4}c_{3}r^{3}-\frac{1}{2}c_{3}rw\\
 & \quad+\frac{1}{8}c_{2}r^{2}+\frac{1}{2}c_{2}w+\frac{1}{2}c_{1}r,\\
H_{2} & =\frac{15}{128}r^{2}r_{x}^{2}-\frac{1}{32}wr_{x}^{2}+\frac{3}{16}rr_{x}w_{x}+\frac{1}{16}w_{x}^{2}+\frac{87}{4096}r^{8}+\frac{13}{128}r^{6}w+\frac{17}{128}r^{4}w^{2}-\frac{1}{16}w^{4}+\frac{3}{32}c_{4}r^{5}\\
 & \quad+\frac{5}{16}c_{4}r^{3}w+\frac{1}{4}c_{4}rw^{2}-\frac{1}{64}3c_{3}r^{4}-\frac{1}{4}c_{3}r^{2}w-\frac{1}{4}c_{3}w^{2}-\frac{1}{16}3c_{2}r^{3}-\frac{1}{4}c_{2}rw+\frac{3}{8}c_{1}r^{2}+\frac{1}{2}c_{1}w.
\end{align*}

\subsubsection*{Stäckel system}

The Stäckel Hamiltonians, defined by the separation curve \eqref{ez10}
and written in Viète coordinates are
\begin{align*}
H_{1} & =q_{1}p_{2}^{2}+2p_{1}p_{2}+q_{1}^{7}-6q_{2}q_{1}^{5}+10q_{2}^{2}q_{1}^{3}-4q_{2}^{3}q_{1}-c_{4}q_{1}^{4}+3c_{4}q_{2}q_{1}^{2}-c_{4}q_{2}^{2}+c_{3}q_{1}^{3}-2c_{3}q_{1}q_{2}\\
 & \quad+c_{2}q_{2}-c_{2}q_{1}^{2}+c_{1}q_{1},\\
H_{2} & =p_{1}^{2}+2p_{2}q_{1}p_{1}+p_{2}^{2}q_{1}^{2}-p_{2}^{2}q_{2}+q_{2}q_{1}^{6}-5q_{2}^{2}q_{1}^{4}+6q_{2}^{3}q_{1}^{2}-q_{2}^{4}+2c_{4}q_{1}q_{2}^{2}-c_{4}q_{1}^{3}q_{2}+c_{3}q_{1}^{2}q_{2}-c_{3}q_{2}^{2}\\
 & \quad-c_{2}q_{1}q_{2}+c_{1}q_{2}.
\end{align*}
They generate the following Stäckel system
\[
\begin{pmatrix}q_{1}\\
q_{2}\\
p_{1}\\
p_{2}
\end{pmatrix}_{t_{1}}=\begin{pmatrix}2p_{2}\\
2p_{1}+2p_{2}q_{1}\\
-7q_{1}^{6}+30q_{2}q_{1}^{4}+4c_{4}q_{1}^{3}-30q_{2}^{2}q_{1}^{2}-3c_{3}q_{1}^{2}+2c_{2}q_{1}-6c_{4}q_{2}q_{1}+4q_{2}^{3}-p_{2}^{2}-c_{1}+2c_{3}q_{2}\\
6q_{1}^{5}-20q_{2}q_{1}^{3}-3c_{4}q_{1}^{2}+12q_{2}^{2}q_{1}+2c_{3}q_{1}-c_{2}+2c_{4}q_{2}
\end{pmatrix},
\]
\[
\begin{pmatrix}q_{1}\\
q_{2}\\
p_{1}\\
p_{2}
\end{pmatrix}_{t_{2}}=\begin{pmatrix}2p_{1}+2p_{2}q_{1}\\
2p_{2}q_{1}^{2}+2p_{1}q_{1}-2p_{2}q_{2}\\
-6q_{2}q_{1}^{5}+20q_{2}^{2}q_{1}^{3}+3c_{4}q_{2}q_{1}^{2}-12q_{2}^{3}q_{1}-2p_{2}^{2}q_{1}-2c_{3}q_{2}q_{1}-2c_{4}q_{2}^{2}-2p_{1}p_{2}+c_{2}q_{2}\\
-q_{1}^{6}+10q_{2}q_{1}^{4}+c_{4}q_{1}^{3}-18q_{2}^{2}q_{1}^{2}-c_{3}q_{1}^{2}+c_{2}q_{1}-4c_{4}q_{2}q_{1}+4q_{2}^{3}+p_{2}^{2}-c_{1}+2c_{3}q_{2}
\end{pmatrix}.
\]
Their Lax representation \eqref{T1} contains the Lax matrix (\ref{laxm}),
given in Viète coordinates by
\[
\mathbb{L}=\begin{pmatrix}-p_{2}\lambda-p_{1}-p_{2}q_{1} & \lambda^{2}+q_{1}\lambda+q_{2}\\
\mathbb{L}_{21} & p_{2}\lambda+p_{1}+p_{2}q_{1}
\end{pmatrix},
\]
where
\begin{align*}
\mathbb{L}_{21} & =\lambda^{6}-q_{1}\lambda^{5}+\left(q_{1}^{2}-q_{2}\right)\lambda^{4}+\left(-q_{1}^{3}+2q_{2}q_{1}+c_{4}\right)\lambda^{3}+\left(q_{1}^{4}-3q_{2}q_{1}^{2}-c_{4}q_{1}+q_{2}^{2}+c_{3}\right)\lambda^{2}\\
 & \quad+\left(-q_{1}^{5}+4q_{2}q_{1}^{3}+c_{4}q_{1}^{2}-3q_{2}^{2}q_{1}-c_{3}q_{1}+c_{2}-c_{4}q_{2}\right)\lambda\\
 & \quad+q_{1}^{6}-c_{4}q_{1}^{3}-q_{2}^{3}-p_{2}^{2}+c_{3}q_{1}^{2}+6q_{1}^{2}q_{2}^{2}+c_{1}-c_{2}q_{1}-5q_{1}^{4}q_{2}-c_{3}q_{2}+2c_{4}q_{1}q_{2},
\end{align*}
and the auxiliary matrices that are given by
\[
\mathbb{U}_{1}=\begin{pmatrix}0 & 1\\
\lambda^{4}-2q_{1}\lambda^{3}+\left(3q_{1}^{2}-2q_{2}\right)\lambda^{2}+\left(-4q_{1}^{3}+6q_{2}q_{1}+c_{4}\right)\lambda+5q_{1}^{4}+3q_{2}^{2}+c_{3}-2c_{4}q_{1}-12q_{1}^{2}q_{2} & 0
\end{pmatrix},
\]
\[
\mathbb{U}_{2}=\begin{pmatrix}-p_{2} & \lambda+q_{1}\\
(\mathbb{U}_{2})_{21} & p_{2}
\end{pmatrix},
\]
where
\begin{align*}
(\mathbb{U}_{2})_{21} & =\lambda^{5}-q_{1}\lambda^{4}+\left(q_{1}^{2}-2q_{2}\right)\lambda^{3}+\left(-q_{1}^{3}+4q_{2}q_{1}+c_{4}\right)\lambda^{2}+\left(q_{1}^{4}-6q_{2}q_{1}^{2}-c_{4}q_{1}+3q_{2}^{2}+c_{3}\right)\lambda\\
 & \quad-q_{1}^{5}+c_{4}q_{1}^{2}-9q_{1}q_{2}^{2}+c_{2}-c_{3}q_{1}+8q_{1}^{3}q_{2}-2c_{4}q_{2}.
\end{align*}

\subsubsection*{Transformation between jet and canonical coordinates}

Finally, the map (\ref{mapa}) between the jet variables and the
(extended by Casimir variables $c_{i}$) Viète coordinates on the
stationary manifold $\mathcal{M}_{2}$ is given by
\begin{subequations}\label{ez156}
\begin{equation}
q_{1}=\frac{1}{2}r,\qquad q_{2}=\frac{3}{8}r^{2}+\frac{1}{2}w,\qquad p_{1}=\frac{1}{4}rr_{x}+\frac{1}{4}w_{x},\qquad p_{2}=\frac{1}{4}r_{x},\label{ez15}
\end{equation}
together with
\begin{equation}
\begin{split}c_{1} & =h_{2}\equiv-\frac{9}{64}wr^{4}-\frac{3}{8}vr^{3}-\frac{3}{8}w^{2}r^{2}-ur^{2}-\frac{1}{2}vwr-\frac{1}{2}r_{2x}r-\frac{1}{4}w^{3}-\frac{5}{16}r_{x}^{2}-uw-\frac{1}{4}w_{2x},\\
c_{2} & =h_{3}\equiv-\frac{9}{64}r^{5}-\frac{3}{4}wr^{3}-vr^{2}-\frac{3}{4}w^{2}r-ur-vw-\frac{1}{4}r_{2x},\\
c_{3} & =h_{4}\equiv-\frac{15}{64}r^{4}-\frac{9}{8}wr^{2}-vr-\frac{3}{4}w^{2}-u,\qquad c_{4}=h_{5}\equiv-\frac{5}{8}r^{3}-\frac{3}{2}wr-v.
\end{split}
\label{ez16}
\end{equation}
\end{subequations}
Notice that eliminating the fields $u$ and $v$ from \eqref{ez16}
one obtains the relations
\begin{equation}
\begin{split}c_{1} & =-\frac{5}{8}c_{4}r^{3}+c_{3}r^{2}-\frac{1}{2}c_{4}rw+c_{3}w-\frac{5}{32}r^{6}-\frac{1}{32}r^{4}w+\frac{3}{4}r^{2}w^{2}-\frac{1}{2}rr_{2x}-\frac{5}{16}r_{x}^{2}+\frac{1}{2}w^{3}-\frac{1}{4}w_{2x},\\
c_{2} & =c_{3}r+c_{4}w+\frac{3}{32}r^{5}+r^{3}w+\frac{3}{2}rw^{2}-\frac{1}{4}r_{2x}.
\end{split}
\label{ez17}
\end{equation}
Applying \eqref{ez15} to the above Hamiltonian equations we obtain
respective vector fields \eqref{ez9} together with the relations
\eqref{ez17}.

The inverse of the whole map \eqref{ez156} is given
by
\begin{align*}
w=2q_{2}-3q_{1}^{2},\qquad w_{x}=4p_{1}-8p_{2}q_{1},\qquad r=2q_{1},\qquad r_{x}=4p_{2},
\end{align*}
together with
\begin{align*}
u & =2c_{4}q_{1}-c_{3}-5q_{1}^{4}+12q_{2}q_{1}^{2}-3q_{2}^{2},\qquad v=-c_{4}+4q_{1}^{3}-6q_{2}q_{1},\\
w_{2x} & =40c_{4}q_{1}^{3}-28c_{3}q_{1}^{2}+16c_{2}q_{1}-40c_{4}q_{2}q_{1}+8c_{3}q_{2}-4c_{1}-20p_{2}^{2}-76q_{1}^{6}+280q_{2}q_{1}^{4}-216q_{2}^{2}q_{1}^{2}+16q_{2}^{3},\\
r_{2x} & =-12c_{4}q_{1}^{2}+8c_{3}q_{1}+8c_{4}q_{2}-4c_{2}+24q_{1}^{5}-80q_{2}q_{1}^{3}+48q_{2}^{2}q_{1}.
\end{align*}

\section{Conclusions}

In this article we investigated a surprising link between the stationary $N$-field cKdV system and a family of $N+1$ Stäckel systems. The result of this paper is that -- in a very precise sense -- each stationary cKdV system can be parameterized as a Stäckel system in exactly $N+1$ different ways. These different parametrizations are then shown to be equivalent through appropriate finite-dimensional analogues of Miura maps. One of the profits coming from our construction is the fact that $n$-time solutions of these Stäckel systems lead to a particular $n$-time solution of the first $n$ systems of the $N$-field cKdV hierarchy.

An open problem, worth further investigation, is to find other constraints on various soliton hierarchies, including cKdV, that also lead to some classes of Stäckel systems.

\setcounter{equation}{0}
\renewcommand{\theequation}{A.\arabic{equation}}

\section*{Appendix}\addcontentsline{toc}{section}{$\quad$ Appendix}

\paragraph*{Proof of Lemma \ref{lemmahk}.} It is immediate to see that
\eqref{s14} for $k<N$ has the form:
\[
h_{k}=\sum_{i=0}^{k}\sum_{j=i}^{k}\mathcal{J}_{i}(P_{n-k+j},P_{n+i-j})\equiv f_{k+1},
\]
hence \eqref{s18a}.

Notice that by \eqref{c8a} we have the equality
\begin{equation}
\sum_{i=0}^{N}\sum_{j=k-n}^{i+n}\mathcal{J}_{i}(P_{n-k+j},P_{n+i-j})=0\label{ll1}
\end{equation}
valid for $k\leqslant2n+N-1$. Taking into account that $P_{n-k+j}\neq0$
only for $j\geqslant k-n$ and $P_{n+i-j}\neq0$ only for $j\leqslant i+n$
one can see that if $k\geqslant N$ and $k\geqslant n$ the formula
\eqref{s14} takes the form:
\begin{align*}
h_{k} & =\sum_{i=0}^{N}\sum_{j=\max\{i,k-n\}}^{\min\{k,i+n\}}\mathcal{J}_{i}(P_{n-k+j},P_{n+i-j})\equiv\Biggl(\sum_{i=k-n+1}^{N}\sum_{j=i}^{k}+\sum_{i=0}^{k-n}\sum_{j=k-n}^{i+n}\Biggr)\mathcal{J}_{i}(P_{n-k+j},P_{n+i-j})\\
 & \overset{\text{by \eqref{ll1}}}{=}\sum_{i=k-n+1}^{N}\Biggl(\sum_{j=i}^{k}-\sum_{j=k-n}^{i+n}\Biggr)\mathcal{J}_{i}(P_{n-k+j},P_{n+i-j})\\
 &=-\sum_{i=k-n+1}^{N}\Biggl(\sum_{j=k+1}^{i+n}+\sum_{j=k-n}^{i-1}\Biggr)\mathcal{J}_{i}(P_{n-k+j},P_{n+i-j})\\
 & \equiv-2\sum_{i=k-n+1}^{N}\sum_{j=k-n}^{i-1}\mathcal{J}_{i}(P_{n-k+j},P_{n+i-j})\equiv g_{k-n+1}.
\end{align*}
Thus, comparing it with \eqref{gk}, we prove the formula \eqref{s18b}
for $k\geqslant N$. The remaining case $k=N-1$ will be obtained
as part of the following computation.

Similarly as above, the formula \eqref{s14} for $n\leqslant k\leqslant N-1$
takes the form:
\begin{align*}
h_{k} & =\sum_{i=0}^{k}\sum_{j=\max\{i,k-n\}}^{\min\{k,i+n\}}\mathcal{J}_{i}(P_{n-k+j},P_{n+i-j})\equiv\Biggl(\sum_{i=k-n+1}^{k}\sum_{j=i}^{k}+\sum_{i=0}^{k-n}\sum_{j=k-n}^{i+n}\Biggr)\mathcal{J}_{i}(P_{n-k+j},P_{n+i-j})\\
 & \overset{\text{by \eqref{ll1}}}{=}\Biggl(\sum_{i=k-n+1}^{k}\sum_{j=i}^{k}-\sum_{i=k-n+1}^{N}\sum_{j=k-n}^{i+n}\Biggr)\mathcal{J}_{i}(P_{n-k+j},P_{n+i-j})\\
 & \equiv\Biggl(\sum_{i=k-n+1}^{k}\sum_{j=i}^{k}-\sum_{i=k-n+1}^{k}\sum_{j=k-n}^{i+n}-\sum_{i=k+1}^{N}\sum_{j=k-n}^{i+n}\Biggr)\mathcal{J}_{i}(P_{n-k+j},P_{n+i-j})\\
 & =-\Biggl(\sum_{i=k-n+1}^{k}\sum_{j=k+1}^{i+n}+\sum_{i=k-n+1}^{k}\sum_{j=k-n}^{i-1}+\sum_{i=k+1}^{N}\sum_{j=k-n}^{i+n}\Biggr)\mathcal{J}_{i}(P_{n-k+j},P_{n+i-j})\\
 & \equiv-\Biggl(2\sum_{i=k-n+1}^{k}\sum_{j=k-n}^{i-1}+\sum_{i=k+1}^{N}\sum_{j=k-n}^{i+n}\Biggr)\mathcal{J}_{i}(P_{n-k+j},P_{n+i-j}).
\end{align*}
Now, comparing with \eqref{gk} and \eqref{ggk} we see that
\begin{align*}
h_{k} & \equiv g_{k-n+1}+\Biggl(2\sum_{i=k+1}^{N}\sum_{j=k-n}^{i-1}-\sum_{i=k+1}^{N}\sum_{j=k-n}^{i+n}\Biggr)\mathcal{J}_{i}(P_{n-k+j},P_{n+i-j})\\
 & \equiv g_{k-n+1}+\Biggl(\sum_{i=k+1}^{N}\sum_{j=k-n}^{i-1}-\sum_{i=k+1}^{N}\sum_{j=i}^{i+n}\Biggr)\mathcal{J}_{i}(P_{n-k+j},P_{n+i-j})\\
 & \equiv g_{k-n+1}+\Biggl(\sum_{i=k+1}^{N}\sum_{j=k-n}^{i-1}-\sum_{i=k+1}^{N}\sum_{j=k-n}^{k}\Biggr)\mathcal{J}_{i}(P_{n-k+j},P_{n+i-j})\\
 & =g_{k-n+1}+\sum_{i=k+1}^{N}\sum_{j=k+1}^{i-1}\mathcal{J}_{i}(P_{n-k+j},P_{n+i-j})\equiv g_{k-n+1}+\tilde{g}_{k+2}.
\end{align*}
Hence, we obtain \eqref{s18c} and the remaining case of \eqref{s18b},
since the sum in the last line of the above computation vanishes for
$k=N-1$.


\end{document}